\documentclass[aps,
prl,
superscriptaddress,
onecolumn,
floatfix]{revtex4-2}

\usepackage{amsmath,mathtools,amsthm,amssymb}

\usepackage{tabularx}
\usepackage{units}
\usepackage{complexity}
\usepackage{graphicx}
\usepackage{subfigure}
\usepackage{color}
\usepackage{xcolor}
\usepackage{bbold}
\usepackage{complexity}

\usepackage[linesnumbered,ruled,vlined]{algorithm2e}
\usepackage{listings}

\usepackage{lipsum} 
\usepackage{titlesec}
\usepackage{enumitem}

\usepackage{pifont}
\usepackage{times}
\usepackage{comment}
\usepackage{array}
\newcolumntype{H}{>{\setbox0=\hbox\bgroup}c<{\egroup}@{}}
\usepackage{graphics}
\usepackage[normalem]{ulem}
\usepackage{complexity}

\usepackage[
    breaklinks,
    pdftex,
    colorlinks=true,
    linkcolor=myrefcolor,
    citecolor=myurlcolor,
    urlcolor=myurlcolor
]{hyperref}

\hypersetup{
    colorlinks = true,
    linkcolor = [rgb]{0.1, 0.2, 0.4},  
    citecolor = [rgb]{0.2, 0.4, 0.2},  
    urlcolor = [rgb]{0.6, 0.2, 0.2}    
}

\definecolor{myrefcolor}{RGB}{242, 10, 10}  
\definecolor{myurlcolor}{RGB}{255, 138, 48}  

\usepackage[
    breaklinks,
    pdftex,
    colorlinks=true,
    linkcolor=myrefcolor,   
    citecolor=myurlcolor,   
    urlcolor=myurlcolor     
]{hyperref}

\hypersetup{
    colorlinks = true,
    linkcolor = [rgb]{0.15, 0.35, 0.55},  
    citecolor = [rgb]{0.2, 0.5, 0.3},     
    urlcolor = [rgb]{0.7, 0.2, 0.2}       
}

\usepackage[toc,page]{appendix}

\usepackage{braket}
\usepackage{physics}
\usepackage{csquotes}

\definecolor{antonio}{rgb}{.2,.5,.1}

\newtheorem{theorem}{Theorem}
\newtheorem{proposition}[theorem]{Proposition}

\newtheorem{definition}[theorem]{Definition}
\newtheorem{lemma}[theorem]{Lemma}
\newtheorem{corollary}[theorem]{Corollary}

\newtheorem{remark}[theorem]{Remark}

\newcommand{\MatC}[1]{\mathcal{L}\!\left(\mathbb{C}^{#1}\right)}

\newcommand{\hs}[2]{\langle #1, #2\rangle_{HS}}

\newcommand{\fu}{Dahlem Center for Complex Quantum Systems, Freie Universit\"{a}t Berlin, 14195 Berlin, Germany}

\begin{document}

\title{Efficient learning of quantum states prepared with few fermionic non-Gaussian gates}

\date{\today}

\author{Antonio Anna Mele}
\email[]{a.mele@fu-berlin.de}
\affiliation{\fu}

\author{Yaroslav Herasymenko}
\email[]{yaroslav@cwi.nl}
\affiliation{QuSoft and CWI, Science Park 123, 1098 XG Amsterdam, The Netherlands
}
\affiliation{QuTech, TU Delft, P.O. Box 5046, 2600 GA Delft, The Netherlands}
\affiliation{Delft Institute of Applied Mathematics, TU Delft, 2628 CD Delft, The Netherlands}

\begin{abstract}
The experimental realization of increasingly complex quantum states underscores the pressing need for new methods of state learning and verification. In one such framework, quantum state tomography, the aim is to learn the full quantum state from data obtained by measurements. Without prior assumptions on the state, this task is prohibitively hard.
Here, we present an efficient algorithm for learning states on $n$ fermion modes prepared by any number of Gaussian and at most $t$ non-Gaussian gates. By Jordan-Wigner mapping, this also includes $n$-qubit states prepared by nearest-neighbour matchgate circuits with at most $t$ SWAP-gates. Our algorithm is based exclusively on single-copy measurements and produces a classical representation of a state, guaranteed to be close in trace distance to the target state. The sample and time complexity of our algorithm is $\mathrm{poly}(n,2^t)$; 
thus if $t=\mathcal{O}(\log(n))$, it is efficient. We also show that, if $t$ scales \emph{slightly} more than logarithmically, any learning algorithm to solve the same task must be inefficient, under common cryptographic assumptions. 
We also provide an efficient property testing algorithm that, given access to copies of a state, determines whether such a state is far or close to the set of states for which our learning algorithm works. In addition to the outputs of quantum circuits, our tomography algorithm is efficient for some physical target states, such as those arising in time dynamics and low-energy physics of impurity models. 
Beyond tomography, our work sheds light on the structure of states prepared with few non-Gaussian gates and offers an improved upper bound on their circuit complexity, enabling an efficient circuit compilation method.
\end{abstract}
\maketitle

\section{Introduction}
Quantum state tomography is the task of reconstructing a classical description of a quantum state from experimental data~\cite{anshu2023survey,Cramer_2010}. Beyond its foundational significance in quantum information theory, it stands as the gold standard for verification and benchmarking of quantum devices~\cite{Cramer_2010}.
However, in the absence of any prior assumptions on the state to be learned, one encounters necessarily the \emph{curse of dimensionality} of the Hilbert space: learning the classical description of a generic quantum state demands resources that grow exponentially with the number of qubits~\cite{anshu2023survey,Haah_2017}. Simply storing and outputting the density matrix of a state already results in an exponential cost in time.
This raises the crucial question of identifying interesting classes of quantum states that can be efficiently learned using a number of state copies and time scaling at most \emph{polynomially} with the system size. Only a few classes of states are currently known to be efficiently learnable --- in particular, matrix product states~\cite{Cramer_2010,Lanyon_2017}, finitely-correlated states~\cite{fanizza2023learning}, high-temperature Gibbs states~\cite{rouzé2023learning}, states prepared by shallow quantum circuits~\cite{huang2024learning,landau2024learningquantumstatesprepared,kim2024learningstatepreparationcircuits}, stabilizer states~\cite{montanaro2017learning}, quantum phase states~\cite{arunachalam2023optimal}, and fermionic Gaussian states~\cite{aaronson2023efficient,Gluza_2018}.
The latter class of states comprises those prepared by fermionic Gaussian circuits~\cite{Surace_2022}, also referred to as free fermionic (non-interacting) evolutions or fermionic linear-optics circuits~\cite{Terhal_2002, knill2001fermionic}. Via Jordan-Wigner mapping, such states on $n$ fermionic modes can also be viewed as $n$-qubit states, prepared by generalized matchgate circuits~\cite{knill2001fermionic,Jozsa_2008,Valiant}.
Fermionic Gaussian states states play a key role in condensed matter physics and quantum chemistry, via the Hartree-Fock method and in the context of Fermi Liquid and Bardeen-Cooper-Schrieffer theories~\cite{Echenique_2007,Martin_2004,giuliani2008quantum, schrieffer2018theory}. These states are also essential in understanding many exactly solvable spin models~\cite{Baxter:1982zz,PhysRevB.83.075103,Kitaev_2006}.
In quantum computing, fermionic Gaussian states are primarily recognized for their efficient classical simulability~\cite{Terhal_2002,Valiant,Jozsa_2008}.
As in the case of Clifford circuits, for which the introduction of magic gates, such as T-gates, allows to reach universal quantum computation~\cite{Nielsen_Chuang_2010}, also for the case of Gaussian circuits the inclusion of certain magic gates~\cite{Hebenstreit_2019,Brod_2011,brod2014computational}, for example SWAP gates~\cite{Brod_2011}, allows to reach universality. If the number $t$ of T-gates in a Clifford circuit is low, the resulting states can still be efficiently simulated classically~\cite{gottesman1998heisenberg,Bravyi_2016,AarGot}; it has also been recently demonstrated that such states, termed as $t$-doped stabilizer states~\cite{Oliviero_2021,Leone_2021Chaos}, are still efficiently learnable~\cite{grewal2023efficient,leone2023learning,hangleiter2024bell}.
Similarly, in the past year, it has been shown that Gaussian circuits with a few magic gates are also classically simulable~\cite{dias2023classical,reardonsmith2023improved,cudby2023gaussian}. However, the learnability of such ``$t$-doped fermionic Gaussian states" remains unknown and this motivates the core-question of our work:
\begin{quote}\centering
    {\em Can we efficiently learn states prepared by Gaussian operations (e.g. matchgates) and a few magic gates?}
\end{quote}
We answer it by proposing a quantum algorithm of polynomial time and sample complexity that uses only single-copy measurements and learns a succint classical description of a $t$-doped fermionic Gaussian state; the learned state is guaranteed to be close to the true state in trace distance. Our presentation is framed in the language of qubits, but the results seamlessly translate into the fermionic formalism. Our learning algorithm may also be feasible to implement in near-term fermionic analog quantum simulators~\cite{Vijayan_2020,Mazurenko}, like cold atoms in optical lattices~\cite{Fermi-HUB}, since we only utilize time evolutions of simple few-body fermionic Hamiltonians~\cite{Naldesi_2023}.
The core of our algorithm relies on a result of independent interest, elucidating the structure of states in question. In particular, for any $t$-doped fermionic Gaussian state $\ket{\psi}$ we show that there exists a Gaussian operation $G$ such that
$
    G^\dag \ket{\psi} = \ket{\phi}\otimes\ket{0^{n-\kappa t}},
$
where $\ket{\phi}$ is supported on $\kappa t$ qubits and $\kappa$ is a small constant.
Informally, this says that all the magic (non-Gaussianity) of such states can be \emph{compressed} to a few qubits via a Gaussian operation. 
The proof of our compression theorem is constructive, which has implications for the circuit complexity of $\ket{\psi}$ and for improved preparation of doped fermionic Gaussian states.

\begin{figure*}[t]
    \centering
    \includegraphics[width=0.97\textwidth]{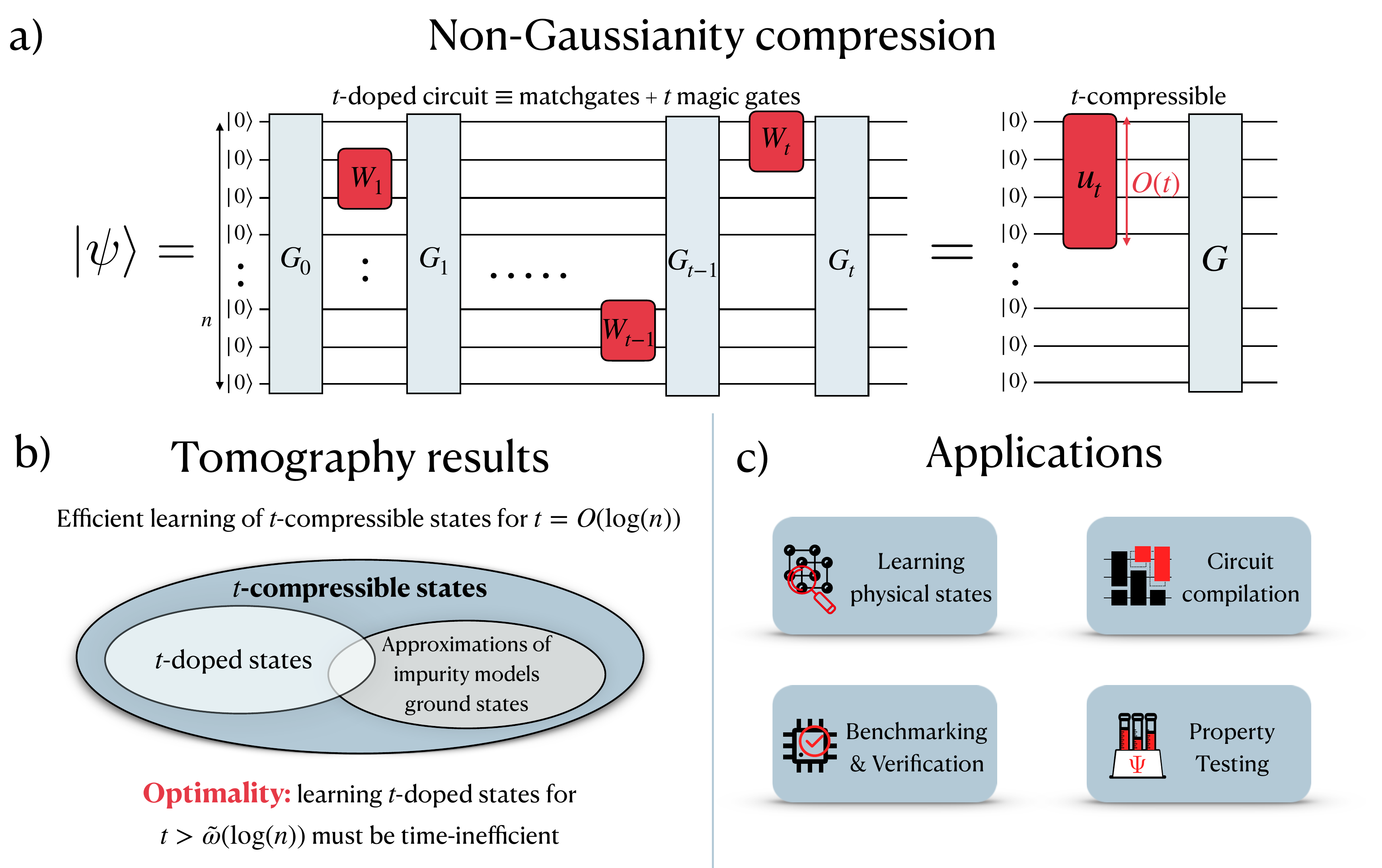}
    \caption{A visual summary of our work. a) We show that any $t$-doped state, i.e., a state prepared by an arbitrary number of Gaussian layers and at most $t$ local fermionic non-Gaussian gates (i.e., a $t$-doped matchgate circuit), can be equivalently prepared by a single non-Gaussian unitary $u_t$ acting solely on $O(t)$ fermionic modes (or qubits) and a global Gaussian unitary $G$. This shows that non-Gaussianity can be compressed into a localized part of the system. b) We provide a tomography algorithm to learn such $t$-compressible states, i.e., states of the form $G(\ket{\phi}\otimes\ket{0^{n-t}} )$, where $G$ is an arbitrary Gaussian unitary and $\ket{\phi}$ is an arbitrary quantum state supported solely on $t$ modes. The algorithm is efficient as long as $t = O(\log(n))$. This also directly implies an algorithm to learn $t$-doped states. Furthermore, up to common cryptographic assumptions, we prove that any tomography algorithm capable to learn a state prepared by slightly more than $\log(n)$ gates must be computationally inefficient. Moreover, we show that physical states, such as ground states of impurity models~\cite{Bravyi_2017}, are approximately $O(\log(n))$-compressible and thus efficiently learnable by our algorithm. c) Our results have various applications, including learning physical states; efficient circuit compilation (implied by the non-Gaussianity compression); benchmarking and verification of digital and analog quantum devices; and testing if the underlying unknown state is $t$-compressible --- of the form for which our algorithm is efficient.}
    \label{figmain}
\end{figure*}

The high level idea of the learning algorithm is to first learn a Gaussian unitary $G$ which \emph{compresses} the magic, apply it to the state, and then perform full state state tomography on the first few qubits alone.
Our learning algorithm has a time complexity $\mathcal{O}(\mathrm{poly}(n,2^t))$, i.e. it scales polynomially in the system size $n$ and exponentially in the number of non-Gaussian gates $t$. Thus it is efficient as long as the number of non-Gaussian gates is $t=\mathcal{O}(\log(n))$. Furthermore, we establish that the task of learning such states  is computationally intractable when the number of non-Gaussian gates scales \emph{slightly} more than logarithmically, under a common cryptography assumption~\cite{ShallowSrini,diakonikolas2022cryptographic,ananth2023revocable,gupte2022continuous}. We show the latter result using the theory of pseudorandom quantum states~\cite{Ji_2018,brakerski2019pseudo} and qubit-to-fermion mappings~\cite{Kitaev_2006}. In doing that, we bring pseudorandom quantum states, so far explored only for qubit-based systems, to the fermionic realm.
Our learning algorithm generalizes the one presented by Aaronson et al.~\cite{aaronson2023efficient}, which is tailored to learn only those states prepared by \emph{particle-number conserving} Gaussian gates and $t=0$ (in our work we relax both of these assumptions). 

Furthermore, our algorithm extends to all \emph{Gaussian compressible} states, i.e., those states which can be written as \(G(\sigma \otimes \ketbra{0^{n-t}}{0^{n-t}})G^{\dag}\), where \(G\) is a Gaussian unitary and \(\sigma\) is a possibly mixed quantum state supported on the first \(t\) qubits. We also propose an efficient method to \emph{test} if a given state is close or far from the set of compressible states, by showing an efficiently estimatable quantity that lower bounds the distance to this set. 
Moreover, we also demonstrate that our learning algorithm can learn states that are close to being compressible -- this feature is particularly significant as it emphasizes the noise-robustness of the algorithm. Additionally, we note that ground states of impurity models, a well-regarded class of quantum states in condensed matter physics~\cite{Kondo1964ResistanceMI,AndKond,AndersonLoc,Bravyi_2017}, are \emph{approximately} compressible, making our algorithm suitable for efficient tomography of such physically relevant states. We further provide numerical evidence showing that states prepared by time evolutions governed by impurity model Hamiltonians~\cite{Bravyi_2017} remain approximately compressible until constant evolution times, thus making them learnable by our tomography algorithm.

It should be noted that the concept of magic compression was first introduced in the context of Clifford+T circuits by Leone, Oliviero et al.~\cite{leone2023learning22PUBL1,leone2023learning22PUBL2} and later exploited for learning $t$-doped stabilizer states~\cite{grewal2023efficient,leone2023learning}. Our strategy of proving non-Gaussianity compression and applying it to quantum state tomography was inspired by these earlier works. It is an intriguing fact that a similar compression theorem holds in our context, even though the mathematical structures of stabilizer states and fermionic Gaussian states appear quite different.

See Figure \ref{figmain} for the visual summary of our results. In the next sections we summarize our findings, stating more precisely our results and the essential ideas that underlie them. In the Supplementary Material, we provide the technical details.

\subsection{Preliminaries}
Our work can be applied to two distinct and naively separate settings: a system of $n$ qubits with 1D matchgates circuits and their magic gates (e.g., SWAP gates), or a native fermionic system of $n$ modes with states prepared by fermionic Gaussian evolutions and local non-Gaussian evolutions. These two perspectives are mathematically related through the Jordan-Wigner mapping. We will use it now as a \emph{definition} of Majorana operators, thus directly aligning our discussion with the qubit language.
Majorana operators, denoted as $\gamma_{2k-1}$ and $\gamma_{2k}$ for $k \in [n]\coloneqq \{1,\dots,n\}$, are defined in terms of standard Pauli operators as $\gamma_{2k-1} \coloneqq (\prod_{j=1}^{k-1} Z_j) X_k$ and $\gamma_{2k} \coloneqq (\prod_{j=1}^{k-1} Z_j) Y_k$. Alternatively, they can be defined in the fermionic language through their anticommutation relations~\cite{Terhal_2002,Bravyi_2002}.
A fermionic Gaussian unitary $G$ is a unitary that satisfies $G^\dagger\gamma_\mu G= \sum^{2n}_{\nu=1} O_{\mu,\nu} \gamma_\nu$ for any $\mu \in [2n]$, where $O\in \mathrm{O}(2n)$ is an orthogonal matrix. The product of two Gaussian unitaries is Gaussian. Notably, a one-to-one correspondence exists between Gaussian unitaries up to a global phase and $\mathrm{O}(2n)$ orthogonal matrices. Given an orthogonal matrix, it is known how to exactly implement the associated Gaussian unitary using $\mathcal{O}(n^2)$ $2$-local qubits or $2$-local fermionic Gaussian operations~\cite{zhao2023PhD,dias2023classical,Zhao_2024}.
A pure fermionic Gaussian state can be defined as $\ket{\psi}=G\ket{0^n}$, where $G$ is a Gaussian unitary and $\ket{0^n}$ denotes the zero computational basis state.
Given a quantum state $\rho$, its correlation matrix $C(\rho)$ is defined as the real anti-symmetric $2n \times 2n$ matrix with elements $[C(\rho)]_{j,k} \coloneqq  -\frac{i}{2}\Tr\left(\gamma_j\gamma_k \rho\right)$, for any $j<k\in[2n]$.
We have that $C(G\rho G^{\dagger})= O  C(\rho) O^{T}$, for any Gaussian unitary $G$ associated with $O\in \mathrm{O}(2n)$.
A well-known result in linear algebra~\cite{BookLinAlg} asserts that any real anti-symmetric matrix $C$ can be decomposed in the so-called `normal form': 
\begin{align}
    C=O\bigoplus_{j = 1}^{n} \begin{pmatrix} 0 &  \lambda_j \\ -\lambda_j & 0 \end{pmatrix}O^T,
    \label{eq:decomAntisym}
\end{align}
where $O$ is an orthogonal matrix in $\mathrm{O}(2n)$ and $\lambda_j\ge 0 $, for any $j\in[n]$, are dubbed as `normal' eigenvalues, ordered in increasing order.
We denote the trace distance between two quantum states $\ket{\psi}$ and $\ket{\phi}$ as $d_{\mathrm{tr}}(\ket{\psi},\ket{\phi})  \coloneqq   \frac{1}{2}\norm{\ketbra{\psi}-\ketbra{\phi}}_1$. Given a matrix $A$, its operator norm $\norm{A}_{\infty}$ is defined as its largest singular value.
We refer to the Supplementary Material (SM) for more preliminaries.

\subsection{Structure of $t$-doped Gaussian states}
States prepared by Gaussian circuits applied to a computational basis state are efficiently simulable classically. However, by incorporating `non-Gaussian', or `magic' operations, such as $\operatorname{SWAP}$-gates~\cite{Brod_2011,Jozsa_2008}, one can render Gaussian circuits universal for quantum computation. The term `magic gate' comes from a loose parallel to Clifford circuits, which are efficiently simulable per se but become universal upon introduction of `magic' non-Clifford $\operatorname{T}$-gates.

Here we consider non-Gaussian operations generated by $\kappa$ Majorana operators $\{\gamma_{\mu(r)}\}^{\kappa}_{r=1}$, where $\mu(1),\dots,\mu(\kappa) \in [2n]$. Examples of such non-Gaussian operations for $\kappa=4$ are the $\operatorname{SWAP}$-gate or a unitary $\exp\!\left(i \theta \gamma_{1}\gamma_{5}\gamma_{6}\gamma_{8}\right)$ for $\theta \in \mathbb{R}$; for $\kappa=3$, an example is $\exp\!\left(\theta \gamma_{2}\gamma_{6}\gamma_{7}\right)$. We refer to $\kappa$ as the maximum Majorana locality of the employed non-Gaussian gates.

\begin{definition}[$t$-doped fermionic Gaussian state]
\label{def:tdoped}
A state $\ket{\psi}$ is a $(t,\kappa)$-doped Gaussian state if it can be prepared by Gaussian unitaries $\{G_i\}^{t}_{i=0}$ and $t$ non-Gaussian $\kappa$-local gates $\{W_i\}^{t}_{i=1}$, specifically
\begin{align}
    \ket{\psi} = G_{t}W_t\cdots G_1 W_1 G_0\ket{0^n},
\end{align}
where $\kappa$-local means that each non-Gaussian gate involves at most $\kappa$ Majorana operators. Informally, a state is $t$-doped Gaussian if it is $(t,\kappa)$-doped Gaussian for some fixed constant $\kappa$.
\end{definition}
Similarly, we denote the unitary $U_t  \coloneqq   G_{t}W_t \cdots G_1 W_1 G_0$ as a $t$-doped Gaussian unitary.
We now present our main result concerning the structure of $t$-doped Gaussian states: it is possible to compress all the `non-Gaussianity' of the state into a localized region of the system via a Gaussian operation. This motivates the following definition.
\begin{definition}[$t$-compressible Gaussian state]
\label{def:tcompr_maintext}
Let $t\in [n]$. A state $\ket{\psi}$ is (Gaussian) $t$-compressible if and only if 
\begin{align}
    \ket{\psi}=G(\ket{\phi}\otimes\ket{0^{n- t}}),
\end{align} 
where $G$ is a Gaussian operation, and $\ket{\phi}$ is a state supported solely on the first $t$ qubits.
\end{definition}
In the following, we assume $\kappa t \le n$.
\begin{theorem}[Magic compression in $t$-doped Gaussian states]
\label{th:1compr} 
Any $(t,\kappa)$-doped Gaussian state is $\kappa t$-compressible. 
\end{theorem}
\begin{proof}[Proof sketch]
Let $U_t \ket{0^n}$ be the $t$-doped state, where $U_t=(\prod^t_{t'=1}G_{t'}W_{t'})G_0$ is the $t$-doped unitary. We rearrange $U_t$ as $U_t=\tilde{G}_tG_{\mathrm{aux}}\prod^t_{t'=1}(G^{\dag}_{\mathrm{aux}}\tilde{W}_{t'}G_{\mathrm{aux}})G^{\dag}_{\mathrm{aux}}$, introducing a Gaussian operation $G_{\mathrm{aux}}$ to be fixed and defining $\tilde{W}_{t'} \coloneqq  \tilde{G}^\dag_{t'-1} W_{t'}\tilde{G}_{t'-1}$ and $\tilde{G}_{t'} \coloneqq  G_{t'}..G_0$. 
We require that $G_{\mathrm{aux}}$ satisfies $G^{\dag}_{\mathrm{aux}}\ket{0^n}=\ket{0^n}$, and that each $G^{\dag}_{\mathrm{aux}}\tilde{W}_{t'}G_{\mathrm{aux}}$ is supported non-trivially only on the first $\kappa t$ qubits. The latter is enforced by demanding that the Heisenberg evolution, via $\tilde{G}_{t'-1}G_{\mathrm{aux}}$, of each Majorana operator involved in the Hamiltonian generating $W_{t'}$, has non-trivial support exclusively on the first $\kappa t$ qubits.
The existence of $G_{\mathrm{aux}}$ is shown by demonstrating the existence of its associated orthogonal matrix $O_{\mathrm{aux}}$. The requirements on $G_{\mathrm{aux}}$ translate into the demand that $O_{\mathrm{aux}}$ must be symplectic and such that it sends $\kappa t$ fixed vectors to the span of the first $2\kappa t$ canonical basis vectors. The existence of such $O_{\mathrm{aux}}$ can be proven via the isomorphism between real $2n\times 2n$ symplectic orthogonal matrices and $n\times n$ unitaries~\cite{serafini2017quantum}. Additional details are provided in the Supplementary Material (see Theorem~\ref{th:comprStAPP}).
\end{proof}
This compressibility should not be confused with the compressed quantum computation result described in Ref.~\cite{jozsa2010matchgate}, which pertains to free-fermionic circuits ($t=0$) and demonstrates how such computations can be equivalently expressed in logarithmic space.

Note that while a $(t,\kappa)$-doped Gaussian state is a $\kappa t$-compressible Gaussian state, the reverse implication does not hold due to circuit complexity arguments.
Similarly to Theorem~\ref{th:1compr}, we show that any $t$-doped Gaussian unitary can be represented as:
\begin{align}
    U_t  =  G_A (u_t \otimes I) G_B,
    \label{eq:Ut}
\end{align}
where $G_A$ and $G_B$ denote Gaussian operations, and $u_t$ is a unitary operator supported on $\lceil\kappa t/2\rceil$ qubits (with $\lceil\cdot \rceil$ denoting rounding to the next integer), as elaborated in the Supplementary Material (Theorem~\ref{th:comprUni}). Notably, if $U_t$ is a particle number conserving unitary~\cite{Terhal_2002}, then $G_A$, $u_t$ and $G_B$ can also be chosen as such.

Our proof of Theorem~\ref{th:1compr} is constructive, i.e., given a classical description of the circuit that prepares $\ket{\psi}$, it provides an efficient procedure for finding the compressing Gaussian circuit $G$ and the state $\ket{\phi}$. 
The decomposition of $t$-doped Gaussian states (unitaries) reveals also that they have a circuit complexity, i.e., number of local gates needed for implementing the state (unitary), upper bounded by $\mathcal{O}(n^2+t^3)$ (Proposition~\ref{prop:complexity} in SM). This provides a better circuit complexity upper bound compared to the naive $\mathcal{O}(n^2 t)$ implied by definition~\ref{def:tdoped} for $\kappa=\mathcal{O}(1)$. Hence, our construction reveals also a method to compress the circuit depth (and not only  the magic), which might be used in practice for efficient circuit compilation of doped matchgate circuits. Remarkably, analogous results hold for the Clifford+T gates circuits~\cite{leone2023learning22PUBL2}.

In the context of Clifford operations, the notions of \emph{stabilizer dimension}~\cite{grewal2023efficient} and \emph{stabilizer nullity}~\cite{Beverland_2020,Jiang_2023} were introduced to quantify the degree of stabilizerness of a quantum state, and their connection to magic monotones was investigated. Analogously, we define the \emph{Gaussian dimension} of a state as the number of normal eigenvalues of its correlation matrix that are equal to one, and the \emph{Gaussian nullity} as the number of normal eigenvalues of its correlation matrix that are strictly less than one.

Using Eq.~\eqref{eq:decomAntisym}, we can show that a state has a Gaussian nullity of at most \(t\) (a Gaussian dimension of at least \(n-t\)) if and only if it is \(t\)-compressible. Consequently, \((t,\kappa)\)-doped Gaussian states have a Gaussian nullity of at most \(\kappa t\) (a Gaussian dimension of at least \(n-\kappa t\)).

\subsection{Learning Algorithm}
We present an algorithm for learning $t$-compressible Gaussian states, or, equivalently, quantum states with at least $n-t$ Gaussian dimension. 
Note that this is a broader class than $t$-doped Gaussian states; as an example unrelated to $t$-doped states, ground states of quantum impurity models are approximately of this form (as shown in~\cite{Bravyi_2017} and elaborated further in the subsequent subsections).
By definition, any $t$-compressible Gaussian state $\ket{\psi}$ can be factorized as $G^{\dag}\ket{\psi}=\ket{\phi}\otimes \ket{0^{n-t}}$, where $G^{\dag}$ is Gaussian and $\ket{\phi}$ is a state on $t$ qubits. At a high level, our strategy is to learn the Gaussian unitary $G^{\dag}$, apply it to $\ket{\psi}$, and then perform full state tomography solely on the first $t$ qubits to learn $\ket{\phi}$. Since full state tomography algorithms scale exponentially with the number of qubits~\cite{FastFranca}, for $t=\mathcal{O}(\log(n))$ our algorithm will be efficient.

\begin{algorithm}
\label{alg:algo}
\caption{Learning algorithm for $t$-compressible fermionic Gaussian states}
\KwIn{Accuracy $\varepsilon$, failure probability $\delta$, $N \coloneqq  \lceil \frac{256 n^5}{\varepsilon^{4}}\log\!\left(\frac{12n^2}{\delta}\right) + 2N_{\mathrm{tom}}\!\left(t,\frac{\varepsilon}{2},\frac{\delta}{3}\right) + 24 \log\!\left(\frac{3}{\delta}\right)\rceil$ copies of the $t$-compressible state $\ket{\psi}$, where $N_{\mathrm{tom}}$ is the number of copies needed for $t$-qubit pure state tomography with accuracy $\frac{\varepsilon}{2}$ and failure probability $\frac{\delta}{3}$.}
\KwOut{A classical description of $|\hat{\psi}\rangle$, ensuring $d_{\mathrm{tr}}(|\hat{\psi}\rangle, \ket{\psi})\le \varepsilon$ with probability at least $1-\delta$.}

Estimate the correlation matrix of $\ket{\psi}$ using $\lceil \frac{256 n^5}{\varepsilon^{4}}\log\!\left(\frac{12n^2}{\delta}\right)\rceil$ single-copy measurements (see Lemma~\ref{le:samplecompAPPcommuting}), obtaining $\hat{C}$\;

Express $\hat{C}$ in its normal form $\hat{C}=\hat{O}\hat{\Lambda}\hat{O}^T$ (Eq.\eqref{eq:decomAntisym}) and find the Gaussian unitary $G_{\hat{O}}$ associated with $\hat{O} \in \mathrm{O}(2n)$\;

\For{$i \leftarrow 1$
\KwTo $\lceil2N_{\mathrm{tom}}(t,\frac{\varepsilon}{2},\frac{\delta}{3})+24\log(\frac{3}{\delta})\rceil$}{
    Apply $G_{\hat{O}}^\dag$ to $\ket{\psi}$\;
    Measure the last $n-t$ qubits in the computational basis\;

    \If{the outcome corresponds to $\ket{0^{n-t}}$}{
        Proceed\;
    }
    \Else{
        Discard and move to the next iteration\;
    }

    Perform a step of pure state tomography~\cite{FastFranca} on the remaining $t$ qubits\;
}

Obtain the $t$-qubit state $|\hat{\phi}\rangle$ from tomography\;

\Return $\hat{O}$ and $|\hat{\phi}\rangle$, which identify $|\hat{\psi}\rangle \coloneqq  G_{\hat{O}}(|\hat{\phi}\rangle\otimes \ket{0^{n-t}})$\;
\end{algorithm}

To delve deeper, the initial phase of our learning algorithm entails estimating the correlation matrix entries through single-copy measurements. This can be achieved using different methods outlined in the Supplementary Material, such as measurements in the Pauli basis, global Clifford Gaussian measurements~\cite{PartitionBabbush}, or fermionic classical shadows~\cite{zhao2023PhD, Zhao_2021, wan2023matchgate}.
The estimated correlation matrix \(\hat{C}\) is subsequently transformed into its normal form in Eq.\eqref{eq:decomAntisym} to yield the corresponding orthogonal matrix \(\hat{O}\) associated with the Gaussian operation \(\hat{G}\). (We use the hat symbol to denote the objects estimated from the measurements.) Applying the inverse operation \(\hat{G}^{\dag}\) to $\ket{\psi}$ results in a state that exhibits high fidelity with a state, which is tensor product of an arbitrary state on the first \(t\) qubits and the zero computational basis state on the remaining \(n-t\) qubits. 
Consequently, the learning algorithm queries multiple copies of $\ket{\psi}$ (one at a time), applies $\hat{G}^{\dag}$ to them and measures the last $n-t$ qubits. If the outcome of such measurements correspond to $\ket{0^{n-t}}$, then the algorithm proceeds with a step of pure state tomography~\cite{FastFranca,guta2018fast} on the \(t\)-qubits state. The state tomography routine performed in the compressed space yields the state $|\hat{\phi}\rangle$. The final output of the learning algorithm is $|\hat{\psi}\rangle  \coloneqq   \hat{G}(|\hat{\phi}\rangle \otimes \ket{0^{n-t}})$, and an efficient classical representation can be provided if $t = \mathcal{O}(\log(n))$. Namely, to specify $|\hat{\psi}\rangle$, it is sufficient to provide the complete description
of the $t$-qubit state $|\hat{\phi}\rangle$ and the orthogonal matrix $\hat{O} \in \mathrm{O}(2n)$ associated with $\hat{G}$.

We now present a theorem which formalizes and proves the efficiency of the discussed procedure, outlined in Algorithm~\ref{alg:algo}, to learn doped Gaussian states or, more generally, $t$-compressible Gaussian states. 
\begin{theorem}[Learning algorithm guarantees]
\label{th:joiningpieces_informal}
Let $\ket{\psi}$ be a $t$-compressible Gaussian state, and $\varepsilon, \delta \in (0,1]$. Utilizing $\mathcal{O}\!\left(\mathrm{poly}\!\left(n,2^{t}\right)\right)$ single-copy measurements and computational time, Algorithm~\ref{alg:algo} outputs a classical representation of a state $|\hat{\psi}\rangle$, such that $d_{\mathrm{tr}}(|\hat{\psi}\rangle, \ket{\psi}) \le \varepsilon$ with probability $\ge 1-\delta$.
\end{theorem}
\begin{proof}[Proof sketch]
Using $\mathcal{O}(\mathrm{poly}(n))$ copies of $\ket{\psi}$, we estimate its correlation matrix $C$, yielding $\hat{C}$ such that $\|\hat{C}-C\|_\infty\le\varepsilon_c$ with a failure probability $\le \frac{\delta}{3}$, where $\varepsilon_c \coloneqq  \varepsilon^2/(4(n-t))$.
Expressing $\hat{C}$ in its normal form (Eq.\eqref{eq:decomAntisym}), we find the Gaussian unitary $\hat{G}$ associated to $\hat{O} \in \mathrm{O}(2n)$. 
Let $\ket{\psi^{\prime}} \coloneqq  \hat{G}^{\dag}\ket{\psi}$. As detailed in the Supplementary Material, we derive $\bra{\psi^{\prime}}Z_k\ket{\psi^{\prime}}\ge 1-2 \varepsilon_c$ for each $k \in \{t+1,\dots, n\}$ and, by Quantum Union Bound~\cite{Gao_2015} we get $d_{\mathrm{tr}}(\ket{\phi}\otimes\ket{0^{n-t}},\ket{\psi^{\prime}}) \le \frac{\varepsilon}{2}$,
where $\ket{\phi}\otimes\ket{0^{n-t}}$ corresponds to the state obtained by measuring the last $n-t$ qubits of $\hat{G}^{\dag}\ket{\psi}$ in the computational basis and obtaining the outcome corresponding to $\ket{0^{n-t}}$, an event which occurs with probability $\ge 1-\varepsilon^2/4$.
By querying $\lceil 2N_{\mathrm{tom}}(t,\frac{\varepsilon}{2},\frac{\delta}{3}) + 24 \log(\frac{3}{\delta})\rceil$ copies of $\ket{\psi}$, and, for each copy, applying $\hat{G}^{\dag}$ 
and measuring the last $n-t$ qubits, we get the outcome $\ket{0^{n-t}}$ at least $N_{\mathrm{tom}}(t,\frac{\varepsilon}{2},\frac{\delta}{3})$ times, with failure probability $\le \frac{\delta}{3}$ due to Chernoff bound. Here, $N_{\mathrm{tom}}(t,\frac{\varepsilon}{2},\frac{\delta}{3})$ is the number of copies sufficient for full state tomography~\cite{FastFranca} of a $t$-qubit state with an $\frac{\varepsilon}{2}$ accuracy and a failure probability $\le \frac{\delta}{3}$. 
Performing the $t$-qubit tomography on all the copies where the outcome $\ket{0^{n-t}}$ occurred yields $|\hat{\phi}\rangle$ such that $d_{\mathrm{tr}}(|\hat{\phi}\rangle,\ket{\phi})\le \frac{\varepsilon}{2}$, with a failure probability $\le \frac{\delta}{3}$.
Defining \(|\hat{\psi}\rangle \coloneqq  \hat{G}(|\hat{\phi}\rangle\otimes \ket{0^{n-t}})\), we have $
    d_{\mathrm{tr}}(|\hat{\psi}\rangle, \ket{\psi}) \le d_{\mathrm{tr}}(|\hat{\phi}\rangle,\ket{\phi}) + d_{\mathrm{tr}}(\ket{\phi}\otimes \ket{0^{n-t}},\hat{G}^{\dag}\ket{\psi})$.
This is $\le \varepsilon$ if the algorithm does not fail, an event occurring with probability $\ge 1-\delta$ due to the union bound.
\end{proof}
Theorem~\ref{th:joiningpieces_informal} is re-stated and rigorously proven in the Supplemental Material as Theorem~\ref{th:joiningpieces}. The sample, time and memory complexity of our algorithm for learning $t$-compressible states exhibits a polynomial dependence on \(n\) and an exponential dependence on \(t\): specifically, the $\mathrm{poly}(n)$ contribution (specifically an $\mathcal{O}(n^5)$ scaling) arises solely from estimating and post-processing the correlation matrix, while the $\mathrm{exp}(t)$ contribution arises from full state tomography on $t$-qubits.
It is easy to see that the dependence on \(t\) is optimal, because learning $t$-compressible states is at least as hard as learning an arbitrary pure state on \(t\) qubits and thus requires at least \(\exp(\Omega(t))\) copies of the state~\cite{Haah_2017}.

However, if we focus on the subclass of $t$-doped Gaussian states, a classical shadow tomography based algorithm presented in~\cite{abbas2023quantum,zhao2023learning} achieves $\mathcal{O}(\mathrm{poly}(n,t))$ sample complexity. Specifically, this algorithm requires a number of copies that scales polynomially with the circuit complexity of the state, and $t$-doped states have a circuit complexity $\mathcal{O}(\mathrm{poly}(n,t))$. However, the time complexity of the algorithm in~\cite{abbas2023quantum,zhao2023learning} scales exponentially with the number of qubits $n$, while our algorithm's time complexity scales only polynomially (although always exponentially in $t$). This observation also applies to $t$-doped stabilizer states learning analyzed in recent works~\cite{leone2023learning,grewal2023efficient}.

In our Supplementary Material, we extend our learning algorithm to handle mixed states. Specifically, we provide an algorithm to learn, in trace distance, possibly mixed quantum states that have at least \( n - t \) normal eigenvalues of their correlation matrix equal to one. 
In the more general mixed state scenario, the algorithm in Table~\ref{alg:algo} becomes significantly simpler. Notably, measuring the last \( n - t \) qubits and post-selecting on the outcome \( \ket{0^{n-t}} \) is not necessary, as we do not require the output state to be pure. Thus, it is sufficient to perform full state tomography~\cite{anshu2023survey} on the first \( t \) qubits right after the Gaussian operation \( \hat{G}^{\dagger} \) is applied to the state \( \ket{\psi} \). More details are given in Section~\ref{sec:genmixed} of the appendix.

Additionally, in Section~\ref{sec:genmixed}, we analyze the noise robustness of our algorithm, which is crucial for practical experimental scenarios where the unknown state may not be exactly \( t \)-compressible, but approximately so. Specifically, our analysis reveals that our algorithm allows for efficient tomography of states that are possibly mixed and (sufficiently) approximately compressible, which are the types of states one would expect to get when running a $t$-doped matchgate circuit on a noisy quantum device.


\subsection{Time complexity lower  bound}

It is natural to wonder whether there exist algorithms for learning $t$-doped Gaussian states with time complexity scaling in $t$ as $\mathcal{O}(\mathrm{poly}(t))$. We establish that the answer is no (see Proposition~\ref{prop:nopolyt} in SM), relying on a widely-believed cryptography assumption. 
Specifically, we show that certain families of pseudorandom quantum states~\cite{brakerski2019pseudo,Ji_2018} can be generated using a polynomial number of local non-Gaussian gates. This implies that if there were an algorithm with polynomial time complexity in $t$ for learning $t$-doped Gaussian states, quantum computers could solve \class{RingLWE}~\cite{LWE0} in polynomial time, which is considered unlikely~\cite{LWE0,regev2024lattices,ShallowSrini,diakonikolas2022cryptographic,aggarwal2022lattice,ananth2023revocable}.
While this rules out the existence of efficient algorithms if $t$ scales polynomially with the number of qubits $n$, it does not yet preclude the existence of efficient algorithms if $t$ grows \emph{slightly} more than logarithmically, for example $t=\mathcal{O}((\log n)^2)$. However, we can rule out this possibility by making the stronger assumption that quantum computers cannot solve \class{RingLWE} in sub-exponential time~\cite{ShallowSrini,diakonikolas2022cryptographic,ananth2023revocable,gupte2022continuous}. This implies that the time complexity of any algorithm to learn $\tilde{\mathcal{O}}(t)$-doped Gaussian states (where $\tilde{\mathcal{O}}(\cdot)$ hides polylogarithmic factors) would necessarily be $\exp(\Omega(t))$. In other words, the following holds. 
\begin{theorem}[Time-complexity lower bound, informal] 
\label{thm:lower_bound_informal}
Assuming that quantum computers cannot solve \class{RingLWE} in sub-exponential time, then there is no time efficient algorithm to learn $\tilde{\omega}(\log(n))$-doped Gaussian state which outputs a description of an efficiently preparable quantum state. Here, $\tilde{\omega}(\log(n))\coloneqq\omega\!\left(\log(n)\mathrm{polyloglog}(n)\right)$.
\end{theorem} 
This would prove that the time complexity in $t$ of our algorithm is essentially optimal, because our algorithm is efficient as long as $t=\mathcal{O}(\log(n))$. We show Theorem~\ref{thm:lower_bound_informal} by efficiently 
encoding the pseudorandom quantum states constructions \cite{zhao2023learning} via a specific qubits-to-fermions mapping~\cite{Kitaev_2006} into other states produced by the same number of gates, all of which are now local non-Gaussian. Crucially for our construction, this mapping sends local qubit operations to local fermionic operations with only a constant overhead in the number of qubits. We refer the reader to the Supplementary Material (Proposition~\ref{prop:exp_lower_bound}) for more details.

\subsection{Testing Gaussian dimension}
We have introduced an algorithm for efficiently learning states with a high Gaussian dimension (or, equivalently, small Gaussian nullity), specifically those promised to be $t$-compressible with a small $t$.
A natural question arises: How can we test the Gaussian dimension of a state? In other words, how can we determine if the underlying state is close or far from the set of $t$-compressible states? In our Supplementary Material, using ideas developed in~\cite{Bittel2024testing} for the case of Gaussian states ($t=0$), we establish that the minimum trace distance between a quantum state $\ket{\psi}$ and the set of $t$-compressible Gaussian states, denoted by $\mathcal{G}_t$, satisfies:
\begin{align}
    \frac{1-\lambda_{t+1}}{2} \le \min_{\ket{\phi_t} \in \mathcal{G}_t} d_{\mathrm{tr}}(\ket{\psi}, \ket{\phi_t}) \leq \sqrt{\sum^{n}_{k=t+1}\frac{1-\lambda_k}{2} },
\end{align}
where $\{\lambda_k\}^n_{k=1}$ represents the normal eigenvalues of the correlation matrix of $\ket{\psi}$ ordered in increasing order. These inequalities imply that $\ket{\psi}$ is close in trace distance to the set $\mathcal{G}_t$ if and only if $\lambda_{t+1}$ is close to one.
In particular, assuming that $\ket{\psi}$ is either a state in $\mathcal{G}_t$ or $\min_{\ket{\phi_t} \in \mathcal{G}_t} d_{\mathrm{tr}}(\ket{\psi}, \ket{\phi_t})\ge \varepsilon$, we can determine with at least $1-\delta$ probability which of the two cases is true by accurately estimating $\lambda_{t+1}$. Specifically, $\mathcal{O}((n^5/\varepsilon^4)\log(n^2/\delta))$ copies of the state suffice for this purpose. Notably, this complexity scaling for testing is independent from $t$, in contrast to learning. 
More details are provided in section~\ref{sec:testing} of the Supplementary Material, along with a generalization to the mixed-state scenario.


\subsection{Numerical simulations: $t$-compressible states in many-body physics}
\label{sec:numericsMAIN}
In the previous section, we introduced an algorithm to learn quantum states that are \( t \)-compressible. In Section~\ref{sec:genmixed} of the Supplementary Material, we extended the algorithm to learn states that are not exactly \( t \)-compressible but \emph{approximately}. This connects to the work of Bravyi and Gosset~\cite{Bravyi_2017}, who demonstrated that ground states of so-called quantum impurity models can be well-approximated by compressible states (see Ref.~\cite{Bravyi_2017}, Corollary 1). In this section we discuss the learnability of such states and of the states produced by the time dynamics of impurity models. 

Quantum impurity models represent a bath of free fermions coupled with a small localized interacting subsystem referred to as an impurity. These models are of paramount importance in condensed matter physics and have been extensively studied over the past half-century. The exploration of such Hamiltonians gained significant traction in the 1960s and 70s, largely due to the pioneering work of Anderson, Kondo, and Wilson. Their focus was on the behavior of magnetic impurities within metals~\cite{AndersonLoc,Kondo1964ResistanceMI,Wilson,AndKond}, providing a theoretical basis for the Kondo effect~\cite{AndKond}, which had been experimentally observed much earlier.
Quantum impurity models are also crucial computational tools for analyzing the electronic structure of strongly correlated materials, including transition metal compounds and high-temperature superconductors. They are integral to the dynamical mean field theory (DMFT)~\cite{Vollhardt_2011}, a method used to study these materials. In DMFT, the impurity typically represents a cluster of atoms within a unit cell, while the bath models the bulk of the material.

Formally, a quantum impurity model Hamiltonian is a Hamiltonian \( H_m \) that can be expressed as \( H_m = H_{\mathrm{free}} + H_{\mathrm{imp},m} \), where \( H_{\mathrm{free}} \) is a free-fermionic Hamiltonian (i.e., a Hamiltonian quadratic in the Majorana operators \( H_{\mathrm{free}} = i\sum_{p,q=1}^{2n} h_{p,q} \gamma_p \gamma_q \), with \( h \) being a real anti-symmetric matrix) and \( H_{\mathrm{imp},m} \) is an arbitrary Hamiltonian, potentially non-free fermionic, acting non-trivially on at most \( m \) distinct Majorana operators, where \( m \) is chosen to be \( O(1) \) relative to the number of fermionic modes \( n \).

Restating Corollary 1 from Ref.~\cite{Bravyi_2017}, Bravyi and Gosset demonstrated that ground states \( \ket{\psi_{\mathrm{GS}}} \) of impurity models can be approximated, for any \( \varepsilon > 0 \), by a \( t \)-compressible Gaussian state with \( t \coloneqq O(\log(\varepsilon^{-1})) \). Specifically, there exists a Gaussian unitary \( G \) and a state \( \ket{\phi} \) supported on \( t \) modes such that \( \norm{\psi_{\mathrm{GS}} - \psi_t}_1 \le \varepsilon \), where \( \ket{\psi_t} \coloneqq G (\ket{\phi} \otimes \ket{0^{n-t}}) \)~\footnote{The asymptotic notation \( t = O(\log(\varepsilon^{-1})) \) hides a logarithmic dependence on the inverse of the spectral gap of the free-fermionic Hamiltonian \( H_{\mathrm{free}} \). For more details, see Ref.~\cite{Bravyi_2017}.}. Thus, our algorithm can effectively be applied to learn ground states of impurity models with arbitrary \(\varepsilon\)-precision, with a sample complexity scaling as \(\mathrm{poly}(n,\varepsilon^{-1})\).

An intriguing question arises: can \textit{evolutions} governed by such Hamiltonians maintain approximate $t$-compressibility up to certain evolution times, thus making our algorithm applicable to such non-equilibrium quantum states? To explore this, we conducted numerical simulations with an interacting impurity embedded in the background of two distinct free models: the 1D \emph{transverse-field Ising} model (TFIM) and an \emph{expander graph} model. The respective Hamiltonians are \( H_{\mathrm{TFIM}} = w Z_1Z_2 + \sum_{j=1}^{n} (g X_j X_{\text{mod}(j+1,n)} + Z_j) \), where \( w \) and \( g \) are real numbers, and \( H_{\mathrm{graph}} = w \gamma_1 \gamma_2 \gamma_3 \gamma_4 + \sum_{e \in \mathcal{G}} i v_e \gamma_{e(1)} \gamma_{e(2)} \), where \(\mathcal{G}\) represents a random 4-regular graph, with coupling constants \( v_e \) uniformly set to 1 for simplicity. Without the impurity term (i.e., for \( w=0 \)), both Hamiltonians become quadratic in the Majorana operators (up to a long \( Z \)-string term, corresponding to a fixed parity contribution), implying that the continuous evolution they realize is Gaussian. We focus on the cases \( g=w=1 \), so that the energy scales of the impurity term and the quadratic part are comparable.

We consider the continuous evolution \( \ket{\psi(T)} = e^{-iHT}\ket{0^n} \) with \( H = H_{\mathrm{graph}} \) and \( H_{\mathrm{TFIM}} \) and investigate whether the resulting state \( \ket{\psi(T)} \) is approximately \( t \)-compressible. To determine this, we check if there exists a \( t \) such that
\begin{equation}
    \sqrt{\sum_{m=t+1}^{n} \frac{1}{2} (1 - \lambda_m)} \leq \varepsilon_{\mathrm{trunc}},
    \label{eq:epsilon_tr_defMAIN}
\end{equation}
where \(\{\lambda_j\}_{j=1}^{n}\) are the normal eigenvalues of the correlation matrix of \(\ket{\psi(T)}\) ordered in increasing order, and $\varepsilon_{\mathrm{trunc}}$ is a fixed truncation error, which captures the approximate $t$-compressiblity of the underlying quantum state. This implies that the state \( \ket{\psi(T)} \) can be learned with our algorithm up to an error \( \varepsilon + \varepsilon_{\mathrm{trunc}} \) in trace distance, where \( \varepsilon \) is the accuracy chosen in the tomography algorithm (see Theorem~\ref{th:joiningpiecesMIXEDApp} of the Supplementary material for more details).  This learning error can be brought closer to $\varepsilon_{\mathrm{trunc}}$ by decreasing $\varepsilon$ at the respective cost. In Figure~\ref{fig:rank_dynamics}, we present the dynamics of the minimal Gaussian nullity \( t \) that allows a compression error below \( \varepsilon_{\mathrm{trunc}} = 0.05 \), as a function of the evolution time \( T \).

For small evolution times \( T \), as shown in Figure~\ref{fig:rank_dynamics} (left panel), \( \ket{\psi(T)} \) exhibits small approximate Gaussian nullity. One might wonder if this approximate $t$-compressibility of impurity dynamics scales well with the system size, or whether, on the contrary, maintaining a fixed error in trace distance becomes increasingly difficult for larger $n$. Our numerical analysis suggests the former. For instance, consider the maximal time $T_4$ that allows maintaining approximate $t$-compressibility for $t=4$ and a fixed $\varepsilon_{\mathrm{trunc}}$. The question is whether $T_4$ decreases to zero with increasing $n$ or stabilizes at a constant value. We find that $T_4$ indeed saturates at a certain constant for both the TFIM and the expander models (Figure~\ref{fig:rank_dynamics}, right panel).

\begin{figure}[t]
    \centering 
    {\includegraphics[width=0.48\linewidth]{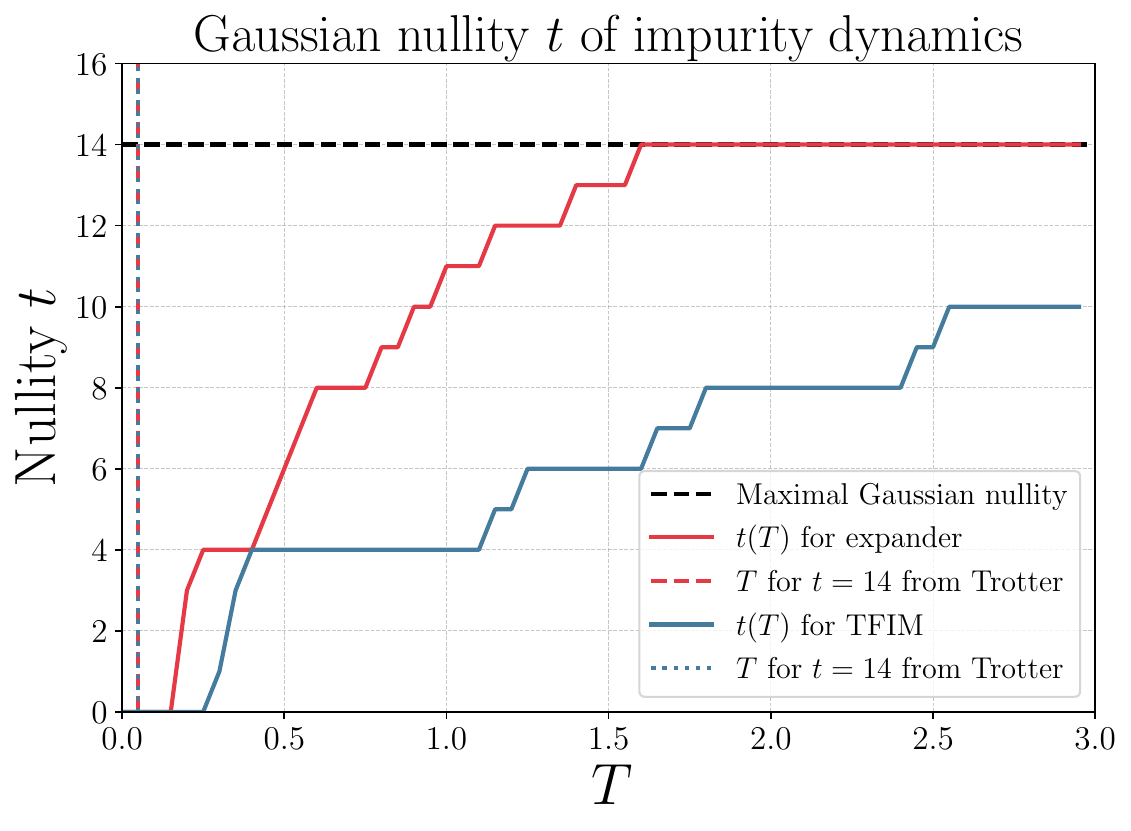}}
    {\includegraphics[width=0.48\linewidth]{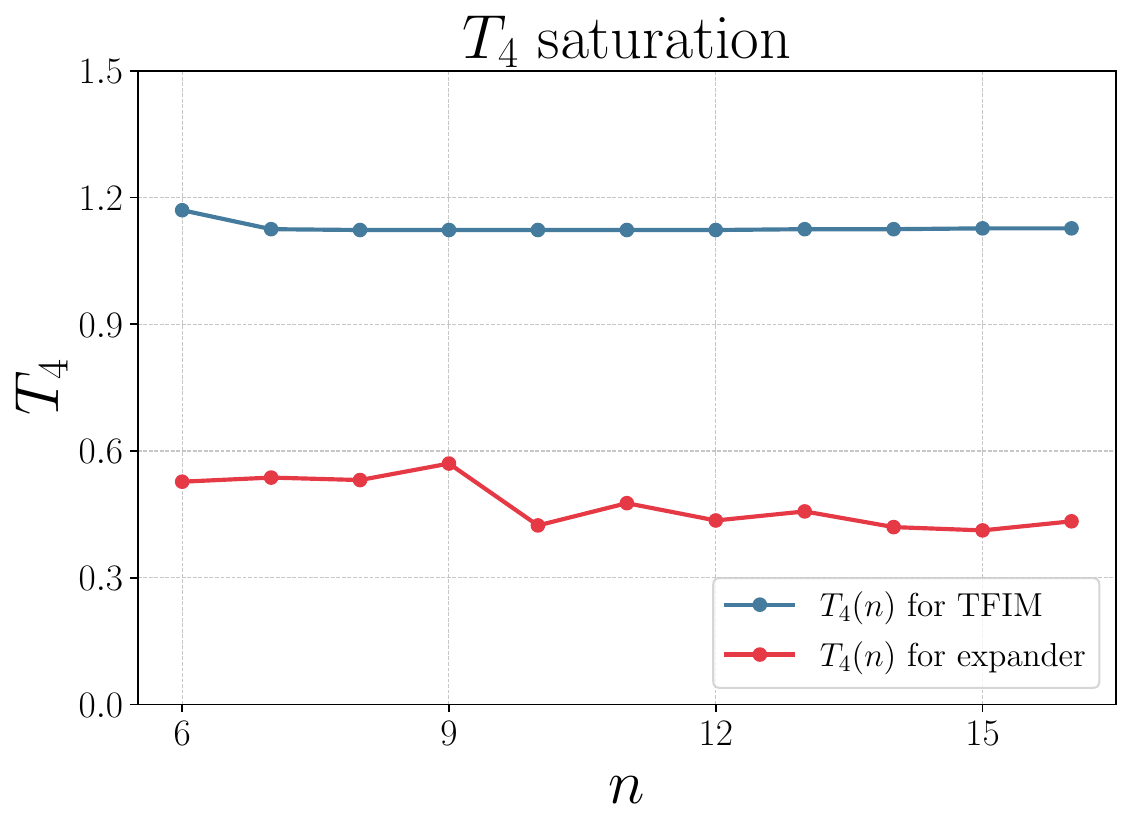}}
    \caption{(left) Dynamics of the Gaussian nullity in an impurity model with expander graph and 1D TFIM backgrounds as defined in the text. The admissible deviation from $t$-compressibility, $\varepsilon_{\mathrm{trunc}}$, is set to $0.05$; system size is $n=14$. The deviation from $t$-compressibility is upper bounded using the normal eigenvalues of the correlation matrix in the state $\ket{\psi(T)}$, as in Eq.~\eqref{eq:epsilon_tr_defMAIN}. Vertical dotted lines show times $T$ which are too large to derive $t<14$ nullity from Trotter approximation and Theorem~\ref{th:1compr}.
    The time axis is resolved up to $\Delta T=0.05$. (right) Maximal evolution time $T_4$ allowing approximate $t=4$ compression (up to $\varepsilon_{\mathrm{trunc}}=0.05$). One observes that $T_4$ saturates at a constant value as a function of $n$ for both models.}
    \label{fig:rank_dynamics}
\end{figure}

We also examine whether the compressibility shown in Figure~\ref{fig:rank_dynamics} can be derived from our Theorem~\ref{th:1compr} for \( t \)-doped circuits using the Trotter approximation. Our numerical analysis rules out this gate-counting hypothesis, because this hypothesis predicts high approximate nullity already for minuscule times (vertical dotted lines in Figure~\ref{fig:rank_dynamics}, left panel). A detailed explanation of this analysis is given in Section~\ref{sec:numerics} of the Supplementary Material. We do not have an analytical derivation of approximate $t$-compressibility of impurity dynamics as observed in our numerical investigation, and leave this question open for future work. The failure of explanation via Trotter expansion highlights the difficulty of this problem. For the TFIM model, the approximate $t$-compressibility may be derivable using 1D Lieb-Robinson bounds. However, such a derivation should not work for the expander graph model -- and at present we do not know how this case could be analytically handled. This open problem also connects to the open questions about the classical simulation of impurity dynamics, raised earlier by Bravyi and Gosset~\cite{Bravyi_2017}.

To run the simulations of \( \ket{\psi(T)} = e^{-iHT} \ket{0^n} \) presented in Figure~\ref{fig:rank_dynamics}, we employed sparse matrix multiplication routines. More details about our numerical simulations are given in Section~\ref{sec:numerics} of the Supplementary Material.

\subsection{Conclusions}
In this work, we have presented an algorithm for efficiently learning $t$-doped fermionic Gaussian states, with sample and time complexity scaling as $\mathcal{O}(\mathrm{poly}(n,2^t))$. Additionally, we have established, under standard cryptography assumptions, that there is no learning algorithm for such class of states with a polynomial dependence on $t$ in the time complexity.
Crucially, our algorithm utilizes solely experimentally feasible single-copy measurements. Its working idea is based on a theorem that we prove, which says that all the non-Gaussianity in a $t$-doped fermionic Gaussian state can be efficiently compressed onto $\mathcal{O}(t)$ qubits through a Gaussian operation. This observation carries potential significance beyond the scope of learning, particularly within the context of quantum many-body theory or within efficient circuit compilation. 
Thus, the results presented in this work, besides being directly relevant to device verification and benchmarking, among other tasks, hold fundamental significance for quantum information theory, as they reveal more about the structure of Gaussian states with fermionic magic gates. Additionally, we introduce a variety of useful analytical techniques, such as new ways to leverage pseudo-random quantum states~\cite{brakerski2019pseudo,Ji_2018} in the context of fermionic systems, which are likely to find applications in future research.
{We also observe that the results presented in our study for fermionic systems have already found application in continuous variable systems and bosonic Gaussian states~\cite{mele2024learning}, where the concept of $t$-doped bosonic Gaussian states is explored, highlighting the broader applicability of our ideas and techniques to a different context.}

Our work offers new directions for further research. 
An open question arising from this work is whether \( t \)-doped Gaussian unitaries can be efficiently learned in a scenario where both the input states to the unitary and the measurements at the end can be chosen. The specific case of \( t = 0 \) has already been addressed in Ref.~\cite{MauroPaper}, and it would be interesting to generalize this to \( t > 0 \).
Another promising direction is the study of the resource theory of fermionic non-Gaussianity, where the notion of Gaussian-nullity we introduced may play a significant role, similar to stabilizer nullity in the Clifford context~\cite{Beverland_2020,Jiang_2023}.
Finally, an analytical justification of the approximate \( t \)-compressibility of impurity model dynamics remains an intriguing open problem, which we intend to explore in future research.





\subsection{Acknowledgments}

We thank Mauro Morales, Lorenzo Leone, Lennart Bittel, Janek Denzler, Tommaso Guaita, Ansgar Burchards, Julio Magdalena de la Fuente, Daniel Liang, Matthias Caro, Alexander Nietner, Marten Folkertsma, Francisco Escudero Gutiérrez, Simona Etinski, Stefano Polla, Jonas Helsen, Barbara Terhal, and Jens Eisert for insightful discussions. We acknowledge the PCMI 2023 program for providing a collaborative space where the authors brainstormed this project.
AAM acknowledges fundings from BMBF (FermiQP, MuniQC-Atoms, DAQC), BMWK (EniQmA), the Quantum Flagship (Millenion, PasQuans2).

\bibliography{ref}
\let\oldaddcontentsline\addcontentsline
\renewcommand{\addcontentsline}[3]{}
\medskip

\bibliographystyle{apsrev4-2}

\let\addcontentsline\oldaddcontentsline

\onecolumngrid

\setcounter{secnumdepth}{2}
\setcounter{equation}{0}
\renewcommand{\thetable}{S\arabic{table}}
\renewcommand{\theequation}{S\arabic{equation}}
\titleformat{\section}[hang]{\normalfont\bfseries}
{Supplementary Material \thesection:}{0.5em}{\centering}
\clearpage
\begin{center}
\textbf{\large Supplementary Material}
\end{center}
\setcounter{equation}{0}
\setcounter{table}{0}

\renewcommand{\theequation}{S\arabic{equation}}
\renewcommand{\bibnumfmt}[1]{[S#1]}
\newtheorem{thmS}{Theorem S\ignorespaces}

\newtheorem{claimS}{Claim S\ignorespaces}

In this supplementary material, we provide a more comprehensive level of detail and explanation for certain statements covered in the main text.

\tableofcontents
\section{Preliminaries}

\subsection{Notation and basics}
In this work, we employ the following notation.
$\mathcal{L}(\mathbb{C}^d)$ denotes the set of linear operators acting on the $d$-dimensional complex vector space $\mathbb{C}^d$. Additionally, we use $[d]$ to represent the set of integers from 1 to $d$, i.e., $[d]  \coloneqq   \{1,\dots, d\}$. We denote with $\mathrm{Mat}(d,\mathbb{F})$ the set of $d \times d$ matrices over the field $\mathbb{F}$.
Let $v \in \mathbb{C}^d$ be a vector, and let $p \in [1,\infty]$. The $p$-norm of $v$ is denoted by $\norm{v}_p$, defined as $\norm{v}_p  \coloneqq    (\sum_{i=1}^d |v_i|^p)^{1/p}$.
The Schatten $p$-norm of a matrix $A\in \mathbb{C}^d$, with $p\in [1,\infty]$, is given by $\norm{A}_p \coloneqq   \Tr((\sqrt{A^\dagger A})^p)^{1/p}$, corresponding to the $p$-norm of the vector of singular values of $A$.
The trace norm and the Hilbert-Schmidt norm are important instances of Schatten $p$-norms, denoted as $\norm{\cdot}_1$ and $\norm{\cdot}_2$ respectively. The Hilbert-Schmidt norm is induced by the Hilbert-Schmidt scalar product $\hs{A}{B} \coloneqq   \Tr(A^\dagger B)$ for $A,B \in \mathcal{L}(\mathbb{C}^d)$.
The infinity norm, $\norm{\cdot}_\infty$, of a matrix is defined as its largest singular value. This norm can be interpreted as the limit of the Schatten $p$-norm of the matrix as $p$ approaches infinity.
For any unitaries $U$ and $V$, and a matrix $A$, we have the unitary invariance property $\norm{UAV}_p=\norm{A}_p$. Also, $\norm{A\otimes B}_p = \norm{A}_p \norm{B}_p$ for $A,B \in \mathcal{L}(\mathbb{C}^d)$. 
We denote with $\mathrm{U}(n)$ the group of $n \times n$ unitary matrices. We denote $\mathrm{O}(2n)$ as the group of real orthogonal $2n \times 2n$ matrices. 
$\mathrm{Sp}(2n,\mathbb{R})$ denotes the group of symplectic matrices over the real field, defined as 
\begin{align}
    \mathrm{Sp}(2n, \mathbb{R}) \coloneqq  \{S \in \mathrm{Mat}(2n,\mathbb{R}) \,:\, S\Omega S^{T}=\Omega\},
\end{align} 
where $\Omega  \coloneqq   \bigoplus_{i = 1}^{n} \begin{pmatrix} 0 & 1 \\ -1 & 0 \end{pmatrix}$. 
The $n$-qubits Pauli operators are represented as elements of the set $\{I,X,Y,Z\}^{\otimes n}$, where $I,X,Y,Z$ represent the standard single qubit Pauli. Pauli operators are traceless, Hermitian, they square to the identity, and form an orthogonal basis with respect to the Hilbert-Schmidt scalar product for the space of linear operators.
We define the set of quantum states as $\mathcal{S}(\mathbb{C}^d) \coloneqq   \{\rho \in \mathcal{L}(\mathbb{C}^d) :\rho \ge 0,\,\Tr(\rho)=1\}$. 
The trace distance between two quantum states $\rho, \sigma$ is defined as $
    d_{\mathrm{tr}}(\rho,\sigma)  \coloneqq   \frac{1}{2}\norm{\rho-\sigma}_1$.
For a function \(f(n)\), if there exists a constant \(c\) and a specific input size \(n_0\) such that $f(n) \leq c \cdot g(n)$ for all \(n \geq n_0\), where \(g(n)\) is a well-defined function, then we express it as \(f(n) = \mathcal{O}(g(n))\). This notation signifies the upper limit of how fast a function grows in relation to \(g(n)\).

For a function \(f(n)\), if there exists a constant \(c\) and a specific input size \(n_0\) such that $f(n) \geq c \cdot g(n)$ for all \(n \geq n_0\), where \(g(n)\) is a well-defined function, then we express it as \(f(n) = \Omega(g(n))\). 
This notation signifies the lower limit of how fast a function grows in relation to \(g(n)\).

For a function \(f(n)\), if for any constant $c$, there exists an input size $n_0$ such that $f(n) > c \cdot g(n)$ for all \(n \geq n_0\), where \(g(n)\) is a well-defined function, then then we express it as \(f(n) = \omega(g(n))\). 
This notation implies that the function grows strictly faster than the provided lower bound.

\subsection{Basics of probability theory}
In this section, we present fundamental results from probability theory useful in our work.
\begin{lemma}[Union bound]
    Let $A_1, A_2, \ldots, A_M$ be events in a probability space. The probability of the union of these events is bounded by the sum of their individual probabilities:
    \[
    \operatorname{Pr}\!\left(\bigcup_{i=1}^{M} A_i\right) \leq \sum_{i=1}^{M} \operatorname{Pr}(A_i). 
    \]
\end{lemma}
\begin{lemma}[Chernoff bound]
\label{th:ChernoffB}
    Consider a set of independent and identically distributed random variables $\{X_i\}_{i=1}^{N}$ with binary outcomes, taking values in $\{0,1\}$. Define $X  \coloneqq   \sum_{i=1}^{N} X_i$ and $\mu  \coloneqq   \mathbb{E}[X]$. For any $\alpha \in (0,1)$, the probability of $X$ being less than $(1-\alpha)$ times its expected value is exponentially bounded as follows:
    \[
    \operatorname{Pr}\left[X \leq (1-\alpha) \mu\right] \leq \exp\!\left(-\frac{\alpha^2 \mu}{2}\right).
    \]
\end{lemma}
\begin{lemma}[Hoeffding's inequality]
    Let $\{X_i\}_{i=1}^{N}$ be independent and identically distributed (i.i.d) random variables with values in $[a,b]\subseteq \mathbb{R}$. For any $\varepsilon>0$, the probability of the deviation of $\hat{X}:= (\sum_{i=1}^{N}X_i)/N$ from its expected value is exponentially bounded as follows:
    \[
    \operatorname{Pr}\!\left(\left| \frac{1}{N}\sum_{i=1}^{N}X_i - \mathbb{E}[\hat{X}]\right| \ge \varepsilon \right) \le 2 \exp\!\left(-\frac{2N \varepsilon^2}{(b-a)^2}\right).
    \]
\end{lemma}
\begin{corollary}
\label{cor:hoff}
    For any $\varepsilon>0$ and $\delta>0$, let $\{X_i\}_{i=1}^{N}$ be i.i.d. random variables with values in $[a,b]\subseteq \mathbb{R}$ and $\hat{X}:= (\sum_{i=1}^{N}X_i)/N$. According to Hoeffding's inequality, a sample size $N$ satisfying
    \[
    N \ge \frac{(b-a)^2}{2\varepsilon^2}\log\!\left(\frac{2}{\delta}\right)
    \]
    suffices to guarantee that $|\frac{1}{N}\sum_{i=1}^{N}X_i- \mathbb{E}[\hat{X}] |< \varepsilon $ with a probability of at least $1-\delta$.
\end{corollary}

\subsection{Fermionic Gaussian states}
In this section, we explore the definitions and essential properties of fermionic Gaussian states. We focus on a system consisting of \(n\) qubits or \(n\) fermionic modes, resulting in a Hilbert space dimension of \(2^n\). More precisely, our work is approached from two perspectives: examining a system of $n$ qubits with 1D matchgates circuits and their magic gates (e.g., SWAP gates), or an equivalent native fermionic system of $n$ modes with states prepared by fermionic Gaussian evolutions and local non-Gaussian evolutions. These perspectives are mathematically connected through the Jordan-Wigner mapping, which we use now for defining Majorana operators in terms of Pauli operators.

\begin{definition}[Majorana operators]
    For each \(k \in \left[n\right]\), Majorana operators are defined as:
    \begin{align}
        \gamma_{2 k-1} \coloneqq  \left(\prod_{j=1}^{k-1} Z_j\right) X_k, \quad \gamma_{2 k} \coloneqq  \left(\prod_{j=1}^{k-1} Z_j\right) Y_k.
    \end{align}
    \label{def:majo}
\end{definition}
Majorana operators can also be defined directly in the fermionic language through their anticommutation relations~\cite{Terhal_2002,Bravyi_2002}.
Majorana operators are Hermitian, traceless, and their squares yield the identity, as deducible from their definitions. Moreover, distinct Majorana operators exhibit anticommutativity and orthogonality with respect to the Hilbert-Schmidt inner product.
\begin{definition}[Majorana ordered products]
    Given a set \(S  \coloneqq   \{\mu_1, \dots, \mu_{|S|}\} \subseteq [2n]\) with \(1 \le \mu_1 < \dots < \mu_{|S|} \le 2n\), we define the Majorana product operator as \(\gamma_S = \gamma_{\mu_1} \cdots \gamma_{\mu_{|S|}}\) if \(S \neq \emptyset\), and \(\gamma_{\emptyset} = I\) otherwise. 
\end{definition}

The \(4^n\) distinct ordered Majorana products are orthogonal to each other with respect the Hilbert-Schmidt inner product, therefore they form a basis for the linear operators $\MatC{2^n}$.
\begin{definition}[Fermionic Gaussian Unitary (FGU)]
    A fermionic Gaussian unitary \(G_O\) is a unitary operator satisfying:
    \begin{align}
        G^\dagger_O\gamma_\mu G_O= \sum^{2n}_{\nu=1} O_{\mu,\nu} \gamma_\nu
    \end{align}
    for any \(\mu \in [2n]\), where \(O\in \mathrm{O}(2n)\) is an orthogonal matrix.
\end{definition}
Since the ordered products of Majorana operators $\gamma_{\mu}$ with $\mu \in [2n]$ form a basis for the linear operators, it is sufficient to specify how a unitary acts under conjugation on the $2n$ Majorana operators $\gamma_{\mu}$, where $\mu \in [2n]$, to uniquely determine the unitary up to a phase. Thus, there is a one-to-one mapping between \(n\)-qubit fermionic Gaussian unitaries (up to a global phase) and orthogonal matrices \(\mathrm{O}(2n)\). In particular, given \(O\in\mathrm{O}(2n)\), it is possible to build the associated unitary using at most \(\mathcal{O}(n(n-1)/2)\) \(2\)-qubit FGU operations. For a more detailed explanation on how to map an \(\mathrm{O}(2n)\) matrix to a fermionic Gaussian unitary, refer to~\cite{dias2023classical,zhao2023PhD,Jiang_2018}.
From the previous definition, it readily follows that \(G^{\dagger}_{O}=G_{O^T}\).
Moreover we have that the product of two Gaussian unitaries is Gaussian, namely $(G_{O_1}G_{O_2})^\dagger \gamma_\mu G_{O_1}G_{O_2}= \sum^{2n}_{\nu=1} (O_1O_2)_{\mu,\nu} \gamma_\nu$.
To streamline notation, we will frequently refer to fermionic Gaussian unitaries $G_O$ simply as $G$ when there is no need to specify the associated orthogonal matrix.
Now, we define a fermionic Gaussian state.
\begin{definition}[Fermionic Gaussian state]
    An $n$-qubit state $\ket{\psi}$ is a (pure) fermionic Gaussian state if it can be expressed as $\ket{\psi}=G \ket{0^n}$, where $G$ is a fermionic Gaussian unitary.
\end{definition}
It is noteworthy that any computational basis state $\ket{x}$ is a Gaussian state. This stems from the observation that each Pauli $X_i$ gate acting on the $i$-th qubit, where $i\in [n]$, is a fermionic Gaussian unitary.
An additional useful identity is $Z_j=-i\gamma_{2j-1}\gamma_{2j}$. 
Thus, the density matrix of a pure fermionic Gaussian state associated to an orthogonal matrix $O\in \mathrm{O}(2n)$ can be written as:
\begin{align}
    \label{eq:densityGaus}G_O\ketbra{0^n}G_O^{\dagger}=G_O\left(\prod^{n}_{j=1}\frac{I-i \gamma_{2j-1}\gamma_{2j}}{2}\right)G_O^\dag=\prod^{n}_{j=1}\left(\frac{I-i \tilde{\gamma}_{2j-1}\tilde{\gamma}_{2j}}{2}\right),
\end{align}
where $\tilde{\gamma}_{\mu}:=G_O \gamma_{\mu} G^{\dag}_O =\sum^{2n}_{\nu=1} O^T_{\mu,\nu} \gamma_\nu $ for each $\mu \in [2n]$.

It can be shown that free-fermionic Hamiltonians, which are quadratic Hamiltonians in the Majorana operators, i.e., $H_{\mathrm{free}} \coloneqq i\sum_{p,q=1}^{2n} h_{p,q} \gamma_p \gamma_q$, where $h$ is a real anti-symmetric matrix, have fermionic Gaussian states as their ground states. Moreover, it can be shown that their time evolutions are fermionic Gaussian unitaries.

We now proceed to define the correlation matrix for any (possibly non-Gaussian) state.
\begin{definition}[Correlation Matrix]
    For any $n$-qubit quantum state $\rho$, its correlation matrix $C(\rho)$ is defined as:
    \begin{align}
        [C(\rho)]_{j,k} \coloneqq  -\frac{i}{2}\Tr\left(\left[\gamma_j,\gamma_k\right]\rho\right),
    \end{align}
    where $j,k\in[2n].$ 
\end{definition}
The correlation matrix of any state is real and anti-symmetric, possessing eigenvalues in pairs $\pm i \lambda_j$ for $j\in[2n]$, where $\lambda_j$ are real numbers such that $|\lambda_j|\le 1$.
The correlation matrix of a quantum state, when evolved using fermionic Gaussian unitaries, undergoes a transformation through conjugation with the corresponding orthogonal matrix, as articulated in the following lemma.
\begin{lemma}[Transformation of the Correlation Matrix under FGU]
    \label{prop:transfFGU}
    For a given $n$-qubit state $\rho$, we have:
    \begin{align}
        C(G_O\rho G_O^{\dagger})= O  C(\rho) O^{T},
    \end{align}
    for any orthogonal matrix $O\in \mathrm{O}(2n)$ and associated fermionic Gaussian unitary $G_O$.
\end{lemma}
This result is readily verified through the definitions of the correlation matrix and fermionic Gaussian unitary.
The state $\ket{x}$ is characterized by a correlation matrix of the form: 
\begin{align}
    \label{eq:compbasisCORR}
    C(\ketbra{x}{x})=\bigoplus_{j = 1}^{n} \begin{pmatrix} 0 &  (-1)^{x_i} \\ -(-1)^{x_i} & 0 \end{pmatrix}.
\end{align}
Hence, for a fermionic Gaussian state $\ket{\psi} \coloneqq  G_O\ket{0^n}$, the correlation matrix takes the form: 
\begin{align}
    C(\ketbra{\psi}{\psi})=O\bigoplus_{j = 1}^{n} \begin{pmatrix} 0 &  1 \\ -1 & 0 \end{pmatrix}O^T.
\end{align}
In the subsequent discussion, we will use $ C(\ket{\psi})$ to denote the correlation matrix of a pure state $\ket{\psi}$. If the state $\ket{\psi}$ is a pure Gaussian state, then each of the eigenvalues of $ C(\ket{\psi})$ is one in absolute value.
Moreover, it is worth noting that every real anti-symmetric matrix can be decomposed in the following form:
\begin{lemma}[Normal form of real anti-symmetric matrices \cite{BookLinAlg}]
\label{le:decSKSYM}
    Any real anti-symmetric matrix $C$ can be decomposed in the so-called `normal-form': 
    \begin{align}
        C=O\bigoplus_{j = 1}^{n} \begin{pmatrix} 0 &  \lambda_j \\ -\lambda_j & 0 \end{pmatrix}O^T,
    \end{align}
    where $O$ is an orthogonal matrix in $\mathrm{O}(2n)$ and $\lambda_j\ge 0 \in \mathbb{R}$, for any $j\in [n]$, are ordered in increasing order. The eigenvalues of $C$ are $\pm i \lambda_j$ where $\lambda_j \in \mathbb{R}$ for any $j\in [n]$.
\end{lemma}

\begin{definition}[Normal eigenvalues]
Given a real-antisymmetric matrix decomposed as in the previous Lemma~\ref{le:decSKSYM}, $\{\lambda_j\}^n_{j=1}$ are dubbed as the `normal eigenvalues' of the matrix.
\end{definition}

\subsection{Particle-number preserving unitaries}
In this section, we introduce the concept of particle-number preserving fermionic unitaries and establish definitions and facts useful for subsequent discussions. We begin by defining creation and annihilation operators.
\begin{definition}[Creation and annihilation operators]
The annihilation operators are defined as:
\begin{align}
    a_{j}  \coloneqq   \frac{\gamma_{2j-1} + i\gamma_{2j}}{2},
\end{align}
for any $j \in [n]$. The creation operators $\{a^{\dagger}_{j}\}^n_{j=1}$ are defined as the adjoints of the annihilation operators.
\end{definition}

\begin{definition}[Particle number operator]
The operator $\hat{N}  \coloneqq   \sum^n_{i=1} a^{\dag}_ia_i$ is denoted as the particle number operator.
\end{definition}

The computational basis forms a set of eigenstates for the particle number operator:
\begin{align}
    \hat{N}\ket{x_1,\dots, x_n} = (x_1+\dots +x_n)\ket{x_1,\dots,x_n},
\end{align}
where $x_1,\dots, x_n \in \{0,1\}$.

\begin{definition}[Particle number preserving unitaries]
\label{def:particle}
A unitary $U$ is said to be particle number preserving if and only if
\begin{align}
    U^{\dag}\hat{N} U = \hat{N},
\end{align}
where $\hat{N}  \coloneqq   \sum^n_{i=1} a^{\dag}_ia_i$ is the particle number operator.
\end{definition}
\begin{definition}[Symplectic group]
The group of real symplectic matrices, denoted as $\mathrm{Sp}(2n,\mathbb{R})$, is defined as
\begin{align}
    \mathrm{Sp}(2n, \mathbb{R})  \coloneqq   \{S \in \mathrm{Mat}(2n,\mathbb{R}) \,:\, S\Omega S^{T}=\Omega\},
\end{align}
where $\Omega  \coloneqq   \bigoplus_{i = 1}^{n} \begin{pmatrix} 0 & 1 \\ -1 & 0 \end{pmatrix}=\bigoplus_{i = 1}^{n} i Y$.
\end{definition}

It is often convenient to express $\Omega$ as $\Omega= I_n \otimes i Y$, where $I_n$ denotes the $n \times n$ identity matrix.
Note that $\Omega$ in the literature is sometimes defined (see, e.g.,~\cite{serafini2017quantum}) as $iY \otimes I_n=\begin{pmatrix} 0_n & I_n \\ -I_n & 0_n \end{pmatrix}$, but the two definitions are equivalent up to orthogonal transformation.

Now, we state an important proposition that will be useful in the subsequent section.
\begin{proposition}[$\mathrm{U}(n)$ is isomorphic to $\mathrm{O}(2n)\cap \mathrm{Sp}(2n,\mathbb{R})$]
\label{prop:isom}
The set of unitaries $\mathrm{U}(n)$ is isomorphic to the set of real symplectic orthogonal matrices $\mathrm{O}(2n)\cap \mathrm{Sp}(2n,\mathbb{R})$.

In particular, any orthogonal symplectic matrix $O \in \mathrm{O}(2n)\cap \mathrm{Sp}(2n,\mathbb{R})$ can be written as follows:
\begin{align}
    O = \Re(u) \otimes I + \Im(u) \otimes iY,
\end{align}
where $u \in \mathrm{U}(n)$ is an $n\times n$ unitary.
\end{proposition}
\begin{proof}
We refer the reader to Appendix B.1 of the book~\cite{serafini2017quantum} for a detailed proof. However, by inspection, it can be verified that the matrix \( O \) defined in this way is both orthogonal and symplectic. This follows from the unitarity of \( u \), which implies the relations:
\begin{align}
\label{eq:relUnOrt}
   \Re(u)\Re(u)^t + \Im(u)\Im(u)^t = I, \quad \quad  \Re(u)\Im(u)^t - \Im(u)\Re(u)^t = 0.
\end{align}
\end{proof}
We now present a Lemma, which shows (some) equivalent definitions of a particle-number preserving fermionic Gaussian unitary.
\begin{lemma}[Particle number preserving Gaussian unitary]
\label{le:equivPPnumber}
Let $G$ be a fermionic Gaussian unitary associated with the orthogonal matrix $O\in \mathrm{O}(2n)$. The following points are equivalents:
\begin{enumerate}
    \item $G$ is particle number-preserving,
    \item $G\ketbra{0^n}G^{\dag}=\ketbra{0^n}$,
    \item  $O$ is symplectic orthogonal, i.e., $O \in \mathrm{O}(2n)\cap \mathrm{Sp}(2n,\mathbb{R})$.
\end{enumerate}
\end{lemma}
\begin{proof} 
If $G$ is particle number-preserving, then
\begin{align}
    G^{\dag} \hat{N} G \ket{0^n} = \hat{N}\ket{0^n} = 0,
\end{align}
where the first equality uses Definition~\ref{def:particle}. This implies $\hat{N} G \ket{0^n} = 0$. Since the ground space corresponding to the zero eigenvalue of the particle number operator $\hat{N}$ is one-dimensional and spanned by $\ket{0^n}$, it follows that $G \ket{0^n}$ is equal to $\ket{0^n}$, up to a phase. Thus, 1. implies 2.

Noting that $\ketbra{0^n}$ is a Gaussian state with a correlation matrix $\Omega = \bigoplus_{j=1}^{n}i Y$, we deduce that the correlation matrix of $G \ketbra{0^n}G^\dag$ is $O \Omega O^T$. Therefore the condition $G\ketbra{0^n}G^{\dag}=\ketbra{0^n}$ is equivalent to $O\Omega O^T = \Omega$, i.e., $O$ is a real symplectic orthogonal matrix. This proves that 2. is equivalent to 3. 

Now, let us assume that $O \in \mathrm{O}(2n)\cap \mathrm{Sp}(2n,\mathbb{R})$. Then, we have for $l\in [n]$:
    \begin{align}
        G^{\dag} a_l G &= G^{\dag} \left(\frac{\gamma_{2l-1}+i\gamma_{2l}}{2}\right) G= \sum^{2n}_{j=1}\left(O_{2l-1,j}  + i O_{2l,j} \right)\frac{\gamma_j}{2}\\
        &=\sum^{n}_{j=1}\left(O_{2l-1,2j-1}  + i O_{2l,2j-1} \right)\frac{\gamma_{2j-1}}{2} + \sum^{n}_{j=1}\left(O_{2l-1,2j}  + i O_{2l,2j} \right)\frac{\gamma_{2j}}{2},\\
        &=\sum^{n}_{j=1}\left(\Re(u)_{l,j} + i \Im(u)_{l,j} \right)\frac{\gamma_{2j-1}}{2} + \sum^{n}_{j=1}\left(-\Im(u)_{l,j} + i \Re(u)_{l,j} \right)\frac{\gamma_{2j}}{2},\\
        &=\sum^{n}_{j=1}u_{l,j} \frac{\left(\gamma_{2j-1}+ i \gamma_{2j}\right)}{2}=\sum^{n}_{j=1}u_{l,j} a_j.
    \end{align}
    where in the fourth equality we used that $O \in \mathrm{O}(2n)\cap \mathrm{Sp}(2n,\mathbb{R})$, and so, because of Proposition~\ref{prop:isom}, it can be written as $
        O=\Re(u)\otimes I + \Im(u) \otimes iY $
    where $u \in \mathrm{U}(n)$ is a $n\times n $ unitary.
    Similarly, we have $G^{\dag} a^{\dag}_l G = \sum^{n}_{j=1}u^*_{l,j} a^{\dag}_j$. This implies 
    \begin{align}
        G^{\dag} \hat{N} G= \sum^n_{l=1}  G^{\dag} a^\dag_l a_l G= \sum^{n}_{l,j,k=1}u^*_{l,j}u_{l,k} a^{\dag}_j a_k=\sum^n_{j=1} a^\dag_j a_j=\hat{N},
    \end{align}
    where we used the unitarity of $u$ in the last step. This proves that 3. implies 1.
    \end{proof}

\section{Structure of $t$-doped Gaussian unitaries and states}
In this section, we analyze the concept of $t$-doped fermionic Gaussian unitaries and states. 
\begin{definition}[$t$-doped fermionic Gaussian unitary]
\label{def:tdopedUN}
    A unitary $U_t$ is a $(t,\kappa)$-doped fermionic Gaussian unitary if it can be decomposed in terms of Gaussian unitaries $\{G_i\}^{t}_{i=0}$ and at most $t$ non-Gaussian $\kappa$-local gates $\{W_i\}^{t}_{i=1}$, specifically
    \begin{align}
        U_t = G_{t}W_t\cdots G_1 W_1 G_0.
    \end{align}
    Here $\kappa$-local refers to the number of 
    distinct Majorana operators that generate each non-Gaussian gate. Informally, a unitary is $t$-doped Gaussian if it is $(t,\kappa)$-doped Gaussian for some fixed constant $\kappa$.
\end{definition}

In our work, we consider non-Gaussian gates $W_{t'}$ for $t' \in [t]$, each generated by $\kappa\le 2n$ in $\kappa$ fixed Majorana operators $\{\gamma_{\mu(t',j)}\}^{\kappa}_{j=1}$, where $\mu(t',1),\dots,\mu(t',\kappa) \in [2n]$. For $\kappa=3$, an example of a non-Gaussian gate is $\exp(\theta\gamma_1 \gamma_2 \gamma_3)=\exp(-i Y_1\theta)$, where $\theta \in \mathbb{R}$, and for $\kappa=4$, an example is the $\operatorname{SWAP}$-gate.

\begin{definition}[$t$-doped fermionic Gaussian state]
    An $n$-qubit state $\ket{\psi}$ is a $(t,\kappa)$-doped (informally, $t$-doped) fermionic Gaussian state if it can be expressed as $\ket{\psi}=U_t \ket{0^n}$, where $U_t$ is a $(t,\kappa)$-doped ($t$-doped) fermionic Gaussian unitary.
\end{definition}
\subsection{Compression of $t$-doped Gaussian unitaries and states}
We start by presenting a Theorem which shows how all non-Gaussianity in a $t$-doped unitary can be `compressed' or `moved' to the first few qubits.  
\begin{theorem}[Compression of non-Gaussianity in $t$-doped unitaries]
\label{th:comprUni}
Any $(t,\kappa)$-doped fermionic Gaussian unitary $U_t$ can be expressed as:
\begin{align}
    U_t = G_A(u_t\otimes I) G_B,
\end{align}
where $G_A$, $G_B$ are Gaussian unitaries, and $u_t$ is a unitary supported exclusively on $\lceil\frac{\kappa t}{2}\rceil $ qubits.
\end{theorem}
\begin{proof}
Let us denote $M \coloneqq  \kappa t$. 
We express the $t$-doped unitary as $U_t=(\prod^t_{t'=1}G_{t'}W_{t'})G_0$. Rearranging it, we have 
\begin{align}
    U_t=\tilde{G}_t\prod^t_{t'=1}\tilde{W}_{t'},
\end{align}
where $\tilde{W}_{t'} \coloneqq  \tilde{G}^\dag_{t'-1} W_{t'}\tilde{G}_{t'-1}$ and $\tilde{G}_{t'} \coloneqq  G_{t'}..G_0$. Informally, the idea behind this rewriting is that Gaussian operations act nicely under conjugation.
Next, we rewrite 
\begin{align}
    U_t=\tilde{G}_tG_{\mathrm{aux}}\prod^t_{t'=1}(G^{\dag}_{\mathrm{aux}}\tilde{W}_{t'}G_{\mathrm{aux}})G^{\dag}_{\mathrm{aux}},
\end{align}
by introducing a Gaussian operation $G_{\mathrm{aux}}$ that we fix later and that will be responsible for moving all the non-Gaussian gates to the first qubits. 
Now, we set:
\begin{align}
     G_A& \coloneqq  \tilde{G}_tG_{\mathrm{aux}}\\
     u_t& \coloneqq  \prod^t_{t'=1}(G^{\dag}_{\mathrm{aux}}\tilde{W}_{t'}G_{\mathrm{aux}})\\
     G_B& \coloneqq  G^{\dag}_{\mathrm{aux}}.
\end{align}
Note that $G_A$ so-defined is Gaussian because the product of Gaussian unitaries is Gaussian, and $G_B$ is clearly Gaussian because the adjoint of a Gaussian unitary is Gaussian. 
We need to show that it is possible to choose $G_{\mathrm{aux}}$ such that $u_t$ is supported only on the first $\lceil M/2 \rceil$ qubits.

More precisely, we will require that $G_{\mathrm{aux}}$ ensures that each $G^{\dag}_{\mathrm{aux}}\tilde{W}_{t'}G_{\mathrm{aux}}$ is generated by the first $M$ Majorana operators alone.  
We will achieve it by ensuring that the Heisenberg evolution under $\tilde{G}_{t'-1}G_{\mathrm{aux}}$, of each Majorana that generates $W_{t'}$, has non-trivial support exclusively on the first $M$ Majorana operators. To find $G_{\mathrm{aux}}$ with the desired property, we will find the associated orthogonal matrix $O_{\mathrm{aux}}$.
Let $\{\mu(t',r)\}_{t'\in [t], r\in [\kappa]}$, with $\mu(t',r)\in [2n]$, be the set of indices of Majorana operators generating the $t'$-th non-Gaussian gate (listed in increasing order). For example, consider $\kappa=4$ and the $t'$-th non-Gaussian gate $W_{t'} \coloneqq  \exp\!\left(i \gamma_2 \gamma_4 \gamma_6 \gamma_9 + i \gamma_4\right)$. In this case, $\mu{(t',3)}=6$.
For each such Majorana operator $\gamma_{\mu(t',r)}$, where $r \in [2n]$, its Heisenberg evolution yields
\begin{align}
    G^{\dag}_{\mathrm{aux}} \tilde{G}^{\dag}_{t'-1} \gamma_{\mu(t',r)} \tilde{G}_{t'-1}G_{\mathrm{aux}}=\sum^{2n}_{m=1} (O_{t'-1} O_{\mathrm{aux}})_{\mu(t',r),m}\gamma_m,
\end{align}
where $O_{t'-1}$ is the orthogonal matrix associated with $\tilde{G}_{t'-1}$. Our demand on the support of Heisenberg-evolved $\gamma_{\mu(t',r)}$ implies 
\begin{align}
    (O_{t'-1} O_{\mathrm{aux}})_{\mu(t',r),m}=(O^T_{\mathrm{aux}} O^T_{t'-1} )_{m,\mu(t',r)}=0
    \label{eq:condi}
\end{align}
for any $m\in \{M+1,\dots, 2n\}$. 

Let us denote as $\{\mathbf{e}_{i}\}^{2n}_{i=1}$ the canonical basis vectors of $\mathbb{R}^{2n}$.
For easy notation, we now denote the unit norm vectors $\{O^T_{t'-1}\mathbf{e}_{{\mu(t',r)}}\}_{t'\in [t], r\in [\kappa]}$ with the set of vectors $\{\mathbf{v}_j\}^{M}_{j=1}$ (remember that $M=\kappa t$).
We can prove the existence of such $O_{\mathrm{aux}}$ by proving the existence of its transpose $O \coloneqq  O^T_{\mathrm{aux}}$. In such notation, the condition in Eq.\eqref{eq:condi} reads as:
\begin{align}
    \mathbf{e}^T_{m} O \mathbf{v}_{j}=0,
\end{align}
for any $j \in [M]$ and $m\in \{M+1,\dots, 2n\}$. In other words, we need to prove the existence of an orthogonal matrix $O$ that maps any given real vectors $\mathbf{v}_1, \dots, \mathbf{v}_M \in \mathbb{R}^{2n}$, where $M \leq 2n$, to the span of the first $M$ canonical basis vectors of $\mathbb{R}^{2n}$. The existence of such a matrix is readily established by selecting an orthonormal basis for $W  \coloneqq   \operatorname{Span}(\mathbf{v}_1, \dots, \mathbf{v}_M)$ and defining the orthogonal matrix that maps this orthonormal basis to the first $\operatorname{dim}(W) \le M$ canonical basis vectors. This concludes the proof.
\end{proof}
The subsequent Theorem~\ref{th:comprStAPP} demonstrates the compression of $t$-doped Gaussian states.

\begin{theorem}[Compression of non-Gaussianity in $t$-doped Gaussian states]
    \label{th:comprStAPP}
    Any $(t,\kappa)$-doped fermionic Gaussian state $\ket{\psi}$ can be represented as:
    \begin{align}
        \ket{\psi} = G(\ket{\phi}\otimes\ket{0^{n-\kappa t}}),
    \end{align}
    where $G$ is a Gaussian unitary, and $\ket{\phi}$ is a state supported exclusively on $\kappa t$ qubits.
\end{theorem}

\begin{proof}
    Let $\ket{\psi} \coloneqq  U_t \ket{0^n}$, where $U_t=(\prod^t_{t'=1}G_{t'}W_{t'})G_0$ is a $t$-doped fermionic Gaussian unitary.
    The proof begins analogously to the one of the previous Theorem~\ref{th:comprUni} and it uses the same notation. In particular, we have:
    \begin{align}
    U_t=\tilde{G}_tG_{\mathrm{aux}}\prod^t_{t'=1}(G^{\dag}_{\mathrm{aux}}\tilde{W}_{t'}G_{\mathrm{aux}})G^{\dag}_{\mathrm{aux}},
\end{align}
where $\tilde{W}_{t'} \coloneqq  \tilde{G}^\dag_{t'-1} W_{t'}\tilde{G}_{t'-1}$ and $\tilde{G}_{t'} \coloneqq  G_{t'}..G_0$.
We set, as before
\begin{align}
     G_A& \coloneqq  \tilde{G}_tG_{\mathrm{aux}}\\
     u_t& \coloneqq  \prod^t_{t'=1}(G^{\dag}_{\mathrm{aux}}\tilde{W}_{t'}G_{\mathrm{aux}})\\
     G_B& \coloneqq  G^{\dag}_{\mathrm{aux}}.
\end{align} 
However, now, we require that $u_t$ has support on the first $M$ qubits, where $M \coloneqq  \kappa t$ (while in the previous proof of Theorem~\ref{th:comprUni} we requested $\lceil M/2 \rceil $), or, equivalently, we request that the generators of $G^{\dag}_{\mathrm{aux}}\tilde{W}_{t'}G_{\mathrm{aux}}$ for any $t'\in [t]$ involve only the first $2M$ Majorana operators.

This time we also impose that $G^{\dag}_{\mathrm{aux}}\ket{0^n}=\ket{0^n}$. 
This implies that $O^T_{\mathrm{aux}}\in \mathrm{Sp}(2n, \mathbb{R})$ (where $O_{\mathrm{aux}}$ is the orthogonal matrix associated to $G_{\mathrm{aux}}$), i.e. $O^T_{\mathrm{aux}}$ must be a symplectic orthogonal matrix, because of Lemma~\ref{le:equivPPnumber}.  We now define $O \coloneqq  O^T_{\mathrm{aux}}$.

Similarly to the previous theorem and using the same notation, we can ensure that $u_t$ is supported only on the first $M$ qubits by demonstrating the existence of an orthogonal, but this time also symplectic, matrix $O$ that satisfies:
\begin{align}
    \mathbf{e}^T_{m} O \mathbf{v}_{j}=0,
\end{align}
for any $j \in [M]$ with arbitrary $\mathbf{v}_1, \dots, \mathbf{v}_M$ real vectors, and $m\in \{2M+1,\dots, 2n\}$.
The existence of such $O$ follows from the subsequent Lemma~\ref{le:exist}, which crucially uses the isomorphism between $2n\times 2n$ symplectic orthogonal real matrices and $n\times n$ unitaries~\cite{serafini2017quantum}.
\end{proof}

\begin{lemma}[Compression via symplectic orthogonal transformations]
\label{le:exist}
    Let $\{\mathbf{e}_{i}\}^{2n}_{i=1}$ be the canonical basis of $\mathbb{R}^{2n}$.
    Let $\mathbf{v}_1, \dots, \mathbf{v}_M \in \mathbb{R}^{2n}$ be a set of unit-norm real vectors, where $M \leq n$. There exists an orthogonal symplectic matrix $O \in \mathrm{O}(2n)\cap \mathrm{Sp}(2n,\mathbb{R})$ such that 
    \begin{align}
        \mathbf{e}^T_{i}O\mathbf{v}_j = 0,
    \end{align} for all $i \in \{2M+1, \dots, 2n\}$ and $j \in [M]$, meaning that all $\{O\mathbf{v}_j\}^{M}_{j=1}$ are exclusively supported on the span of the first $2M$ canonical basis vectors.
\end{lemma}

\begin{proof}
    Orthogonal symplectic matrices $O \in \mathrm{O}(2n) \cap \mathrm{Sp}(2n,\mathbb{R})$ have a bijective correspondence with unitary matrices $U \in \mathrm{U}(n)$ through a well-defined vector space mapping~\cite{serafini2017quantum} (see Proposition~\ref{prop:isom}).
    Specifically, for a $2n$-dimensional real vector $\mathbf{w}  \coloneqq   (w_1, \dots, w_{2n})$, there exists a bijective mapping to an $n$-dimensional complex vector $\mathbf{f}(\mathbf{w})  \coloneqq   (w_1 - i w_2, \dots, w_{2n-1} - i w_{2n})$.
    A unitary transformation $U$ in this $n$-dimensional complex space corresponds to an orthogonal symplectic transformation $O$ in the corresponding $2n$-dimensional real space, and vice versa.
    Thus, finding a unitary $U$ that maps the span of $\mathbf{f}(\mathbf{v}_1), \dots, \mathbf{f}(\mathbf{v}_M)$ to the span of the first $M$ canonical basis vectors of this $n$-dimensional complex space implies the existence of a symplectic orthogonal matrix $O$ that maps $\mathbf{v}_1, \dots, \mathbf{v}_M$ to the span of the first $2M$ canonical basis vectors of the $2n$-dimensional real space.
    To establish the existence of such a unitary, consider the complex vector subspace $W  \coloneqq   \operatorname{Span}(\mathbf{f}(\mathbf{v}_1), \dots, \mathbf{f}(\mathbf{v}_M))$ of dimension at most $M$. By selecting an orthonormal basis for this subspace, we can construct a unitary matrix $U$ that maps this basis to the first $\mathrm{dim}(W) \le M$ canonical basis vectors.
    Hence, the existence of the required unitary matrix is confirmed, implying the existence of a symplectic orthogonal matrix $O$ that maps $\mathbf{v}_1, \dots, \mathbf{v}_M$ to the span of the first $2M$ canonical basis vectors. This concludes our proof.
\end{proof}

It is noteworthy that the `compressed size' obtained for $t$-doped unitaries, which is $\lceil \kappa t /2 \rceil $, is less than $\kappa t$ which we proved for $t$-doped Gaussian states.

\subsection{Compression of $t$-doped Gaussian particle-number preserving unitaries}
The subsequent Proposition~\ref{prop:ComprCONS} demonstrates that the compression of $t$-doped Gaussian unitaries can also be achieved in the particle-number preserving case.
\begin{proposition}[Particle number preserving $t$-doped unitaries]
    \label{prop:ComprCONS}
    Let $U_t$ be a $t$-doped fermionic Gaussian unitary, as per Definition~\ref{def:tdopedUN}, where all the unitaries that compose $U_t$ are particle-number preserving. Then $U_t$ can be decomposed as:
    \begin{align}
        U_t  \coloneqq   G_A (u_t \otimes I) G_B,
    \end{align}
    where $G_A$ and $G_B$ are Gaussian unitaries which preserve the number of particles (see Definition~\ref{def:particle}) and $u_t$ is a particle-number preserving possibly non-Gaussian unitary supported on $\kappa t$ qubits. 
\end{proposition}
\begin{proof}
    The proof of such proposition follows the same lines as the one of Theorem~\ref{th:comprStAPP}. In fact, by inspecting the proof, it readily follows that the so-defined $G_A$ is particle-number preserving. The fact that $G_B$ and $u_t$ are particle-number preserving follows from the condition $G^{\dag}_{\mathrm{aux}}\ket{0^n}=\ket{0^n}$, where $G_B \coloneqq  G^{\dag}_{\mathrm{aux}}$ and by Lemma~\ref{le:equivPPnumber}.
\end{proof}

\subsection{Circuit complexity of $t$-doped Gaussian unitaries and states}
The circuit complexity of a unitary (state) is defined as the minimum number of $\mathcal{O}(1)$-local gates needed for implementing the unitary (state).
We will consider locality both in the qubit and in the fermionic sense; in the latter case it refers to the number of distinct Majorana operators that generate each non-Gaussian gate.  Our subject of interest is the scaling of the complexity of a $t$-doped unitary. Throughout this section, we assume $\kappa=\mathcal{O}(1)$ and let $t=t(n)$ change in some way with $n$.
By Definition~\ref{def:tdopedUN}, a $t$-doped Gaussian unitary $U_t$ can be written as $U_t = G_{t}W_t\cdots G_1 W_1 G_0$, where $\{G_i\}^{t}_{i=0}$ are Gaussian unitaries and $\{W_i\}^{t}_{i=1}$ are, possibly non-Gaussian, $\kappa$-local fermionic gates. 
From this definition and using the fact that any Gaussian unitary can be decomposed as the product of $\le 2n(2n-1)/2$ 
(fermionic) $2$-local gates~\cite{dias2023classical,zhao2023PhD,Zhao_2024}, the fermionic circuit complexity of $t$-doped Gaussian unitaries is upper-bounded by $\mathcal{O}(n^2t)$. The same can be shown for qubit circuit complexity (see below).
But more importantly, we have proven earlier that a $t$-doped Gaussian unitary can be decomposed as $U_t=G_A (u_t \otimes I) G_B$, where $G_A, G_B$ are Gaussians and $u_t$ is a unitary on $\lceil\frac{\kappa t}{2}\rceil$ qubits. In the following, we show that such decomposition reveals an improved upper bound on the circuit complexity of $t$-doped Gaussian unitaries.
\begin{proposition}[Circuit complexity of $t$-doped Gaussian unitaries]
\label{prop:complexity}
The circuit complexity $\mathcal{C}(U_t)$ of a $t$-doped Gaussian unitary $U_t$ is (both in the qubit and fermionic sense):
\begin{align}
    \mathcal{C}(U_t)=\begin{cases}
        \mathcal{O}(n^2+t^3), & \text{if $\kappa t\le n$}\\
        \mathcal{O}(n^2t) ,& \text{otherwise.}
    \end{cases}
\end{align}
\end{proposition}
\begin{proof}
Let us assume that $\kappa t\le n $. Then, $U_t$ can be written as $U_t=G_A (u_t \otimes I) G_B$, where $G_A, G_B$ are Gaussians and $u_t$ is a unitary on $\lceil\frac{\kappa t}{2}\rceil$ qubits. 
In fact, $u_t$ is itself a $(t,\kappa)$-doped Gaussian unitary on $\lceil\frac{\kappa t}{2}\rceil$ qubits. It is not directly obvious, but will be shown momentarily; this will imply the desired circuit complexity $\mathcal{O}(t^3)$. 
We recall our definitions used in the proof of Theorem~\ref{th:comprUni}. We have
$\tilde{G}_{t'} \coloneqq G_{t'}..G_0$ for $t'\in [t]$ and a Gaussian unitary $G_{\mathrm{aux}}$, and set
\begin{align}
     u_t& \coloneqq \prod^t_{t'=1}(G^{\dag}_{\mathrm{aux}}\tilde{G}^\dag_{t'-1} W_{t'}\tilde{G}_{t'-1}G_{\mathrm{aux}})=\prod^t_{t'=1} w_{t'},
\end{align} 
where we defined the unitaries $w_{t'}\coloneqq G^{\dag}_{\mathrm{aux}}\tilde{G}^\dag_{t'-1} W_{t'}\tilde{G}_{t'-1}G_{\mathrm{aux}}$ which act only on the first $\lceil\frac{\kappa t}{2}\rceil$ qubits (equivalently, fermionic modes).
We note that $w_{t'}$ is generated by $\kappa$ Majorana operator superpositions of form $\sum^{\kappa t}_{i=1} [\tilde{\mathbf{v}}_j]_i \gamma_i$, $j\in[\kappa (t'-1)+1,\kappa t']$; here  $\tilde{\mathbf{v}}_{j}\coloneqq O\mathbf{v}_j$ (cf. notation $O$ and $\mathbf{v}_j$ from the proof of Theorem~\ref{th:comprUni}).
Hence, for each of these (non-local) non-Gaussian unitaries $w_{t'}$, we can find a Gaussian operation $g_{t'}$ on the first $\lceil\frac{\kappa t}{2}\rceil$ qubits whose associated orthogonal matrix rotates vectors $\tilde{\mathbf{v}}_j$ into the span of the first $\kappa$ basis vectors.
As a result, we have $w_{t'} = g_{t'}^\dag \tilde{w}_{t'} g_{t'}$, where $\tilde{w}_{t'}$ is now a $\kappa$-local non-Gaussian unitary generated by the first $\kappa$ Majorana operators. By implication, it is also a local qubit gate acting on the first $\lceil\frac{\kappa}{2}\rceil$ qubits. From this it follows that the circuit complexity of each $w_{t'}$ scales as that of a Gaussian $g_{t'}$. As $g_{t'}$ acts on the first $O(t)$ qubits/fermionic modes, its circuit complexity is $O(t^2)$ both in the qubit and the fermionic sense. Therefore, the circuit complexity of $u_t$ is $t\mathcal{O}(t^2)=\mathcal{O}(t^3)$. Moreover, the circuit complexity (both qubit and fermionic) to implement $G_A$ and $G_B$ is $\mathcal{O}(n^2)$. Putting the above observations together, it follows that the circuit complexity of $U_t$ is $\mathcal{O}(n^2+t^3)$. As long as $\kappa t\le n $, this upper bound is tighter than the one which proof follows from the $t$-doped definition, namely $\mathcal{O}(n^2 t)$.

The qubit (and not only fermionic) circuit complexity of $\mathcal{O}(n^2 t)$ for $ \kappa t > n $ can be found in a similar way as the complexity of $\mathcal{O}(t^3)$ we showed for $u_t$ above. In particular, consider any $\kappa$-local non-Gaussian fermionic unitary $W_{t'}$ which participates in $U_t$. Using auxilliary Gaussian rotations, its generating Majorana operators can be mapped to $\{\gamma_1,..,\gamma_\kappa\}$, resulting in a unitary supported by the first $\lceil\frac{\kappa}{2}\rceil$ qubits alone. The asymptotic qubit complexity of $U_t$ is thus determined by that of remaining $t$ Gaussian layers, yielding $\mathcal{O}(n^2 t)$ as promised.
\end{proof}
This Proposition reveals that $t$-doped fermion Gaussian unitaries allow not only a `spatial compression for the magic', but also a compression of the circuit depth.
Since our proof of Theorem~\ref{th:comprUni} is constructive, this provides a useful method for compiling magic matchgate circuits, which might be used in practice to reduce the circuit depth.

\subsection{$t$-compressible Gaussian states}
We now introduce the notion of $t$-compressible fermionic Gaussian state, a class of states that includes the one of $t$-doped Gaussian states. We now reiterate Definition~\ref{def:tcompr_maintext} for convenience. Throughout this section, we assume that $t\in [n]$.
\begin{definition}[$t$-compressible Gaussian state]
\label{def:tcompr}
A state $\ket{\psi}$ is a $t$-compressible (Gaussian) state if and only if it can be represented as $\ket{\psi}=G(\ket{\phi}\otimes \ket{0^{n-t}})$, where $G$ is a Gaussian operation, and $\ket{\phi}$ is a pure state supported solely on the first $t$ qubits.
\end{definition}
A $t$-doped Gaussian state is also a $\kappa t$-compressible Gaussian state because of Theorem~\ref{th:comprStAPP}. However, the reverse is not true because of circuit complexity arguments: $t$-doped Gaussian states exhibit a circuit complexity of at most $\mathcal{O}(n^2t)$. In contrast, a $t$-compressible state features a circuit complexity of $\mathcal{O}(n^2+\exp(t))$, representing the complexity needed for implementing a single Gaussian operation and preparing a generic state supported on $t$ qubits.

In the subsequent Proposition, we elucidate the structure of the correlation matrix of any \(t\)-compressible state, such as \(t\)-doped states. 
\begin{proposition}[Correlation matrix of a \(t\)-compressible Gaussian state]
\label{prop:corrMatSM}
The correlation matrix \(C(\ket{\psi})\) of a \(t\)-compressible Gaussian state \(\ket{\psi}\) can be expressed as:
\begin{align}
    C(\ket{\psi}) = O \bigoplus_{j = 1}^{n} \begin{pmatrix} 0 &  \lambda_j \\ -\lambda_j & 0 \end{pmatrix} O^T,
    \label{eq:tdopedCORR}
\end{align}
where \(\lambda_j \le 1\) for \(j \in [t]\) and \(\lambda_j = 1\) for \(j \in \{t+1,\dots,n\}\), and \(O \in \mathrm{O}(2n)\) is an orthogonal matrix.
\end{proposition}
\begin{proof}
    As per Definition~\ref{def:tcompr}, we represent \(\ket{\psi}\) as \(\ket{\psi} = G(\ket{\phi} \otimes \ket{0^{n-t}})\), where \(G\) is a Gaussian operation and \(\ket{\phi}\) is a pure state supported solely on the first \(t\) qubits.
    Utilizing Lemma~\ref{le:decSKSYM}, we can express the correlation matrix \(C(\ket{\psi})\) as follows:
    \begin{align}
        C(\ket{\psi}) &= Q C(\ket{\phi} \otimes \ket{0^{n-t}}) Q^T \\
        &= Q \left(C(\ket{\phi}) \oplus C(\ket{0^{n-t}})\right) Q^T,
    \end{align}
    where $Q\in \mathrm{O}(2n)$ is the orthogonal matrix associated to the Gaussian unitary $G$. By Eq.\eqref{eq:compbasisCORR}, we have:
    \begin{align}
        C(\ket{0^{n-t}}) = \bigoplus_{j =1}^{n-t} \begin{pmatrix} 0 & 1 \\ -1 & 0 \end{pmatrix}.
    \end{align}
    Since \(C(\ket{\phi})\) is an antisymmetric real matrix, we can decompose it into its normal form (Lemma~\ref{le:decSKSYM}):
    \begin{align}
        C(\ket{\phi})= O_t\bigoplus_{j=1}^{t} \begin{pmatrix} 0 & \lambda_j \\ -\lambda_j & 0 \end{pmatrix} O_t^T
    \end{align}
    The proof concludes by definining $O \coloneqq  Q ( O_t \oplus I_{2n-2t} )$. 
\end{proof}
The previous proposition reveals that $t$-compressible states exhibit at least $n-t$ normal eigenvalues which are exactly one. 
This motivates the following definition, in analogy to the stabilizer dimension and nullity defined in the stabilizer case~\cite{grewal2023efficient,Beverland_2020,Jiang_2023}, which found applications in resource theory of magic.
\begin{definition}[Gaussian dimension and Gaussian nullity of a state]
\label{def:tcomprGDIM}
The Gaussian dimension of a state is defined as the number of the normal eigenvalues of its correlation matrix which are equal to one, while the Gaussian nullity is defined as the number of the normal eigenvalues of its correlation matrix which are strictly less than one. 
\end{definition}

In the following, we show that this is also a sufficient condition for a state to be Gaussian \(t\)-compressible.
\begin{lemma}[Sufficient condition for $t$-compressibility]
\label{prop:suffcond}
Let $\ket{\psi}$ be an $n$-qubit quantum state. If $\ket{\psi}$ has Gaussian dimension $\ge n-t$ (or, equivalently, Gaussian nullity $\le t $), then $\ket{\psi}$ is a $t$-compressible Gaussian state.
\end{lemma}
\begin{proof}
The correlation matrix of $\ket{\psi}$ can be written in its normal form as $C(\ket{\psi})=O\Lambda O^T$, where $\Lambda \coloneqq  \bigoplus_{j = 1}^{n} \begin{pmatrix} 0 & \lambda_j \\ -\lambda_j & 0 \end{pmatrix}$, and $O\in \mathrm{O}(2n)$. The $\{\lambda_j\}^n_{i=1}$ are the normal eigenvalues such that the last $n-t$ are equal to one.
Consider the state $\ket{\psi^{\prime}} \coloneqq  G^{\dag}_{O}\ket{\psi}$ where $G_O$ is the Gaussian unitary associated with $O$. Then, $C(\psi^{\prime})=O^T C(\ket{\psi}) O=\Lambda$. In particular, 
\begin{align}
\Tr(\ketbra{\psi^{\prime}}Z_k)=C(\psi^{\prime})_{2k-1,2k}=\Lambda_{2k-1,2k}=1,
\end{align}
for each $k\in \{t+1,\dots,n\}$. Therefore, $\ket{\psi^{\prime}}$ must be of the form $\ket{\psi^{\prime}}=\ket{\phi}\otimes\ket{0^{n-t}}$, where $\ket{\phi}$ is an arbitrary state on the first $t$ qubits. Therefore, we have $\ket{\psi}=G_O(\ket{\phi}\otimes\ket{0^{n-t}})$, which is a $t$-compressible state.
\end{proof}
Hence, Proposition~\ref{prop:corrMatSM} and Lemma~\ref{prop:suffcond} prove the following.
\begin{proposition}[Equivalence between $t$-compressibily and $n-t$ Gaussian dimension]
\label{prop:equiv}
    A $n$-qubit state is $t$-compressible if and only if its Gaussian dimension is at least $n-t$ (or, equivalently, its Gaussian nullity is at most $t$). 
\end{proposition}
Note that Proposition~\ref{prop:equiv} also proves that a quantum state is a pure Gaussian state if and only if its Gaussian dimension is $n$.
Furthermore, as a direct consequence of the proof of Lemma~\ref{prop:suffcond}, we establish that the Gaussian unitary associated with a $t$-compressible Gaussian state can be selected as the Gaussian unitary corresponding to any orthogonal matrix placing its correlation matrix in the normal form (Lemma~\ref{le:decSKSYM}). This is summarized as follows:
\begin{lemma}
\label{fact:formGQ}
    Every $t$-compressible Gaussian state $\ket{\psi}$ can be written as $\ket{\psi } \coloneqq  G_O (\ket{\phi}\otimes \ket{0^{n-t}})$, where $G_O$ is chosen as the Gaussian unitary associated with an orthogonal matrix $O \in \mathrm{O}(2n)$ that arranges its correlation matrix in the normal form described in Lemma~\ref{le:decSKSYM}, and $\ket{\phi}$ is a state supported on $t$ qubits.  
\end{lemma}

\section{Tomography algorithm}
In this section, we present a detailed and rigorous analysis of the tomography algorithm for $t$-compressible states outlined in the main text (Algorithm~\ref{alg:algo}). Throughout this section, we assume that $t\in [n]$.

\subsection{Useful lemmas and subroutines}
Let us start with a Lemma, which gives a sample complexity upper bound to estimate the correlation matrix of a state using single-qubit Pauli-basis measurements.  We recall that the correlation matrix of a state $\rho$ is a real antisymmetric matrix, defined as:
\begin{align}
    [C(\rho)]_{j,k} = \mathrm{Tr}( O^{(j,k)}\rho),
\end{align}
where $O^{(j,k)} \coloneqq  -i \gamma_j\gamma_k$, for $j<k \in [2n]$ (and the other elements are given by the antisymmetricity of the matrix). Note that $O^{(j,k)}$ are Pauli observables. Thus, we have a total of $M  \coloneqq   n(2n-1)$ Pauli expectation values to estimate.
\begin{lemma}[Sample complexity for estimating the correlation matrix by Pauli measurements]
\label{le:samplecompAPP}
Let $\varepsilon_c, \delta > 0$. Assume to have access to $N \ge N_c(n,\varepsilon_c,\delta)$, with
\begin{align}
N_c(n,\varepsilon_c,\delta) \coloneqq  \left\lceil\frac{8n^3(2n-1)}{\varepsilon_c^{2}}\log\!\left(\frac{2n(2n-1)}{\delta}\right)\right\rceil,
\end{align} copies of an $n$-qubit state $\rho$. Utilizing only $N$ single-copies measurements in the Pauli basis, with probability $\ge 1-\delta$, we can construct an anti-symmetric real matrix $\hat{C}$ such that it satisfies:
\begin{align}
    \norm{\hat{C}-C(\rho)}_\infty \le \varepsilon_c.
\end{align}
\end{lemma}
\begin{proof}
Let $\varepsilon>0$ an accuracy parameter to be fixed.
For each $j<k \in [2n]$, we measure $N'$ copies of $\rho$ in the Pauli basis corresponding to $O^{(j,k)}$, obtaining outcomes $\{X^{(j,k)}_m\}^{N'}_{m=1}$, where $X^{(j,k)}_m \in \{-1,+1\}$. Let $\hat{C}_{j,k}  \coloneqq   \frac{1}{N'}\sum^{N'}_{m=1} X^{(j,k)}_m$. Hoeffding's inequality (specifically Corollary~\ref{cor:hoff}) implies that $N' \ge (4/(2 \varepsilon^2))\log(2M/\delta)$ suffices to guarantee that, with probability at least $1-\delta/M$, we have $\lvert \hat{C}_{j,k} -  \mathrm{Tr}(O^{(j,k)} \rho) \rvert < \varepsilon$.
By using the union bound, we conclude that the probability that this holds for any $j<k \in [2n]$ is at least $1-\delta$. More specifically: 
\begin{align}
    \operatorname{Pr}\!\left(\forall\, j<k \in [2n]: \lvert \hat{C}_{j,k} -  \mathrm{Tr}(O^{(j,k)} \rho) \rvert < \varepsilon \right) & = 1 - \operatorname{Pr}\!\left(\exists\, j<k \in [2n]: \lvert \hat{C}_{j,k} - \mathrm{Tr}(O^{(j,k)} \rho ) \rvert \ge \varepsilon \right) \\
    &\ge 1 - \sum_{j<k \in [2n]}\operatorname{Pr}\!\left(\lvert \hat{C}_{j,k} - \mathrm{Tr}(O^{(j,k)} \rho ) \rvert \ge \varepsilon \right) \\
    &\ge 1-\delta.
\end{align}
Therefore, the total number of measurements needed is $N=N^{\prime}M$. Now, we can conclude by transferring the error to the operator norm. Let $A \coloneqq  \hat{C}-C(\rho)$. For the definition of the operator norm, we have $\norm{A}_\infty \coloneqq  \sup_{\ket{\psi}}\sqrt{\left|\bra{\psi}A^{\dag}A\ket{\psi}\right|}$. Thus, we have:
\begin{align}  
    \left|\bra{\psi}A^{\dag}A\ket{\psi}\right|&\le\sum^{2n}_{i,j,k=1}\left|\bra{\psi}\ketbra{i}{i}A^{\dag}\ketbra{j}{j}A\ketbra{k}{k}\ket{\psi}\right|=2n\varepsilon^2 \sum^{2n}_{i,k=1}\left|\braket{\psi}{i}\right|\left|\braket{k}{\psi}\right|\le 4n^2\varepsilon^2 
\end{align}
where, in the first step, we inserted the resolution of the identity and applied the triangle inequality, in the second step, we applied the upper bound on each matrix element, and, in the last step, we used the Cauchy-Schwartz inequality. Hence, we have $\norm{\hat{C}-C(\rho)}_{\infty}\le 2n \varepsilon$. We conclude by choosing $\varepsilon_c \coloneqq  \varepsilon/2n$.
\end{proof}
While the sequential estimation of individual correlation matrix entries by measurements in the Pauli basis, as described above, may not be the most sample-efficient approach, it might be convenient to adopt in an experiment because of its easy implementation scheme. 
However, instead of independently estimating each correlation matrix entry, one could choose to simultaneously measure mutually commuting observables~\cite{PartitionBabbush} or utilize the fermionic classical shadow protocol introduced in~\cite{Zhao_2021,zhao2023PhD,wan2023matchgate}. This refinement would lead to a reduction in sample complexity by a factor of \(n\), at the cost of implementing a slightly more complicated measurement scheme.

For completeness, we present now a Lemma which gives a sample complexity upper bound for estimating the correlation matrix using a commuting observables measurement scheme. 
The idea is to partition the observables $O^{(j,k)} \coloneqq  -i \gamma_j\gamma_k$, for $j<k \in [2n]$ into disjoint sets of commuting observables. Subsequently, one employs the fact that commuting Pauli observables can be measured simultaneously via a Clifford measurement~\cite{PartitionBabbush,miller2022hardwaretailored}. 
A crucial observation is that two different Pauli observables of the form $-i \gamma_j\gamma_k$ commute if and only if they are associated with different Majorana operators. Using this observation, we can partition these $M=(2n-1)n$ observables into $2n-1$ disjoint sets, each containing $n$ commuting Pauli observables. We refer to~\cite{PartitionBabbush} Appendix C for details of such a partition, and we omit repeating the construction here. However, we point out that the required Clifford transformations can be chosen to be Gaussian as well. 
\begin{lemma}[Sample complexity for estimating the correlation matrix by grouping commuting observables]
\label{le:samplecompAPPcommuting}
Let $\varepsilon_c, \delta > 0$. Assume to have access to $N \ge N_c(n,\varepsilon_c,\delta)$, with
\begin{align}
N_c(n,\varepsilon_c,\delta) \coloneqq  \left\lceil \frac{8n^2(2n-1)}{\varepsilon_c^{2}}\log\!\left(\frac{2n(2n-1)}{\delta}\right)\right\rceil,
\end{align} copies of an $n$-qubit state $\rho$. Utilizing $N$ single-copy (Gaussian) measurements, with probability $\ge 1-\delta$, we can construct an anti-symmetric real matrix $\hat{C}$ such that it satisfies:
\begin{align}
    \norm{\hat{C}-C(\rho)}_\infty \le \varepsilon_c.
\end{align}
\end{lemma}
\begin{proof}
For each of the $2n-1$ sets of commuting Pauli, we find the Clifford $U$ that allows us to simultaneously measure such commuting Pauli in the given set, i.e. we map each of the $n$ Pauli to $\{Z_k\}^{n}_{k=1}$.  Now this Clifford can also be chosen to be Gaussian. Indeed, the key constraint on $U$ is that each of the different Paulis of the form $-i \gamma_j\gamma_k$ with $j<k \in [2n]$ 
(where pairs $(j,k)$ are non-overlapping since these Paulis commute) 
is mapped to $\{Z_k\}^{n}_{k=1}$ with $Z_k:=-i \gamma_{2k-1}\gamma_{2k} $. This constraint can be satisfied by using a Gaussian operation associated to the orthogonal matrix which is a permutation of the Majorana indices from the commuting Paulis into the Majorana indices from the $Z$-Paulis. Consequently, we measure $N'$ copies of $U\rho U^{\dagger}$ in the computational basis.  Thus, for each $O^{(j,k)}$, we obtain outcomes $\{X^{(j,k)}_m\}^{N'}_{m=1}$, where $X^{(j,k)}_m \in \{-1,+1\}$. The unbiased estimators are $\hat{C}_{j,k}  \coloneqq   \frac{1}{N'}\sum^{N'}_{m=1} X^{(j,k)}_m$. As before, Hoeffding's inequality and union bound imply that $N' \ge (2/ \varepsilon^2)\log(2M/\delta)$ suffices to guarantee that the probability of $\lvert \hat{C}_{j,k} -  \mathrm{Tr}(O^{(j,k)} \rho) \rvert < \varepsilon $ holding for each $j<k\in [2n]$ is at least $1-\delta$.
Therefore, the total number of measurements needed is $N=N^{\prime} (2n-1)$. We can conclude as in the previous Lemma.
\end{proof}

\begin{lemma}[Perturbation bounds on the normal eigenvalues of correlation matrices]
\label{le:lbantisymm}
Let $A$ and $B$ be two $2n\times 2n$ anti-symmetric real matrices with normal eigenvalues $\{\lambda_{k}(A)\}^{n}_{k=1}$ and $\{\lambda_{k}(B)\}^{n}_{k=1}$ respectively ordered in increasing order. Then, we have:
\begin{align}
    |\lambda_{k}(A)- \lambda_{k}(B)|\le \norm{A-B}_{\infty},
\end{align}
for any $k\in [n]$.
\end{lemma}
\begin{proof}
This follows from the fact that $C \coloneqq  iA$ and $D \coloneqq  iB$ are Hermitian matrices. Applying Weyl's Perturbation Theorem (see Ref.~\cite{bhatia1996matrix}, section VI), which states that two $2n\times 2n$ Hermitian matrices $C$ and $D$, with eigenvalues $c_1\le\dots\le c_{2n}$ and $d_1\le \dots\le d_{2n}$, satisfy: 
\begin{align}
    \norm{C-D}_{\infty}\ge \max_{j\in [n]}|c_{j}- d_{j}|.
\end{align}
Since $A$ and $B$ are antisymmetric, their eigenvalues are $\{\pm i \lambda_{k}(A)\}^{n}_{k=1}$ and $\{\pm i \lambda_{k}(B)\}^{n}_{k=1}$ respectively. Hence, the eigenvalues of $C$ and $D$ are $\{\pm \lambda_{k}(A)\}^{n}_{k=1}$ and $\{\pm \lambda_{k}(B)\}^{n}_{k=1}$ respectively.
This implies that:
\begin{align}
    \norm{A-B}_{\infty}= \norm{C-D}_{\infty}\ge \max_{j\in [2n]}|c_{j}- d_{j}|=\max_{k\in [n]}|\lambda_{k}(A)- \lambda_{k}(B)|.
\end{align}
\end{proof}
To formalize our learning algorithm, it is useful to invoke the following well-known lemma.

\begin{lemma}[Quantum Union Bound \cite{Gao_2015,odonnell2021quantum,aaronson2006qmaqpoly}]
\label{le:qunionbound}
Let $\varepsilon_1, \dots, \varepsilon_M > 0$, where $M \in \mathbb{N}$. Let $\{P_i\}_{i=1}^{M}$ be projectors, and $\rho$ be a quantum state. If $\operatorname{Tr}(P_i \rho) \geq 1 - \varepsilon_i$ for all $i \in [M]$, then
\begin{align}
    \left\|\rho - \frac{P_M \dots P_1 \rho P_1 \dots P_M}{\operatorname{Tr}(P_M \dots P_1 \rho P_1 \dots P_M)}\right\|_{1} \leq 2\sqrt{\sum_{i\in [M]}\varepsilon_i}\,.
\end{align}
\end{lemma}
We now leverage this known lemma to prove the following.
\begin{lemma}
\label{le:learningLEMMA}
Let $\ket{\psi}$ be a $t$-compressible Gaussian state. Given an estimate $\hat{C}$ for the correlation matrix $C(\ket{\psi})$, there exists a Gaussian operation $\hat{G}$ such that:
\begin{align}
    d_{\mathrm{tr}}(\ket{\phi}\!\otimes\!\ket{0^{n-t}}\!,\hat{G}^{\dag}\!\ket{\psi}) \le \sqrt{(n-t)  \norm{\hat{C}-C(\ket{\psi})}_\infty\!},
\end{align}
where $\ket{\phi}\otimes \ket{0^{n-t}}$ corresponds to the post-measurement state obtained by measuring the last $n-t$ qubits of the state $\hat{G}^{\dag}\ket{\psi}$ in the computational basis and obtaining the outcome corresponding to $\ket{0^{n-t}}$. This event occurs with a probability of at least $1-(n-t)  \norm{\hat{C}-C}_\infty$.
\end{lemma}
\begin{proof}
    According to Proposition~\ref{prop:corrMatSM}, the correlation matrix $C \coloneqq   C(\ket{\psi})$ can be put in the form $C=O \Lambda O^{\rm T}$, where $O\in \mathrm{O}(2n)$ and $\Lambda=i\bigoplus^{n}_{j=1}\lambda_j(C) Y$. Here, $\lambda_j\le 1$ for $j \in [t]$ and $\lambda_j=1$ for $j \in \{t+1,\dots n\}$, and $Y$ represents the $Y$-Pauli matrix. Let $\varepsilon_c \coloneqq  \norm{\hat{C}-C}_\infty$, then we have (because of Lemma~\ref{le:lbantisymm}) that $|\lambda_{j}(\hat{C})-\lambda_{j}(C)|\le \varepsilon_c$, where $\{ \lambda_{j}(\hat{C})\}^{n}_{j=1}$ and $\{ \lambda_{j}(C)\}^{n}_{j=1}$ are the normal eigenvalues of the matrices $\hat{C}$ and $C$ respectively. Thus, we have:
    \begin{align}
        \lambda_{m}(\hat{C})\ge 1 - \varepsilon_c,
        \label{eq:lambdaineq}
    \end{align}
    for $m \in \{t+1,\dots n\}$.
    We can now express the real anti-symmetric matrix $\hat{C}$ in its normal form $\hat{C}=\hat{O}\hat{\Lambda} \hat{O}^{T}$, where $\hat{O}\in \mathrm{O}(2n)$ is an orthogonal matrix and $\hat{\Lambda}$ is a matrix of the form $\hat{\Lambda} = i\bigoplus^{n}_{j=1} \lambda_{j}(\hat{C}) Y$, with $\lambda_{j}(\hat{C})\in \mathbb{R}$ for any $j\in[n]$. 
    Next, consider the state $|\psi^{\prime}\rangle \coloneqq  \hat{G}^{\dag}\ket{\psi}$, where $\hat{G}$ is the Gaussian unitary associated to $\hat{O}^{T}$. It holds that $\lvert C(\psi^{\prime})_{j,k}-(\hat{\Lambda} )_{j,k} \rvert\le \varepsilon_c$, where we used that $C(\psi^{\prime})=\hat{O}^T C(\ket{\psi})\hat{O}$, $\hat{\Lambda}=\hat{O}^T \hat{C}\hat{O}$, Cauchy-Schwarz and the definition of infinity norm.
    Therefore, we have $
         C(\psi^{\prime})_{j,k}\ge (\hat{\Lambda} )_{j,k} -  \varepsilon_c$. In particular, for $m\in\{t+1,\dots, n\}$, we get:
    \begin{align}
        \Tr(Z_m \psi^{\prime}) & =C(\psi^{\prime})_{2m-1,2m} \\
                               & \ge (\hat{\Lambda} )_{2m-1,2m} -  \varepsilon_c \\
                               & =  \lambda_{m}(\hat{C}) -  \varepsilon_c \\
                               & \ge 1- 2 \varepsilon_c,
    \end{align}
    where $Z_m=-i\gamma_{2m-1}\gamma_{2m}$ is the $Z$-Pauli operator acting on the $m$-th qubit and in the last step we used Eq.\eqref{eq:lambdaineq}. Therefore, we also have $
        \Tr(\ketbra{0}{0}_m \psi^{\prime})\ge 1- \varepsilon_c$.
    By using the Quantum Union Bound (Lemma~\ref{le:qunionbound}), we have:
    \begin{align}
        d_{\mathrm{tr}}(\ket{\psi^{\prime}}, \ket{\phi}\otimes \ket{0^{n-t}}) \le \sqrt{(n-t)  \varepsilon_c},
    \end{align}
    where $\ket{\phi}\otimes \ket{0^{n-t}}$ is the post-measurement state after having measured the outcomes corresponding to $\ket{0^{n-t}}$ in the last $n-t$ qubits. By Lemma~\ref{le:qunionbound}, this scenario occurs with probability at least $1-(n-t)\varepsilon_c$.
\end{proof}
In the following, we mention the guarantees of a full pure state tomography algorithm, which demonstrates optimal dependence on the number of qubits and uses only single-copies measurements, albeit with a trade-off in accuracy compared to other algorithms~\cite{Haah_2017,odonnell2015efficient}. This is an example of a procedure that we can utilize in our $t$-qubits full state tomography step of our learning algorithm. 
\begin{lemma}[Fast state tomography \cite{FastFranca}]
\label{le:fasttom}
For any unknown $n$-qubit pure state $|\psi\rangle$, there exists a quantum algorithm that, utilizing $N_{\mathrm{tom}}(n,\varepsilon, \delta) \coloneqq  \mathcal{O}\!\left(2^n n \log (1 / \delta) \varepsilon^{-4}\right)$ copies of $|\psi\rangle$ and $T_{\mathrm{tom}}(n,\varepsilon, \delta) \coloneqq  \mathcal{O}\!\left(4^n n^3 \log (1 / \delta) \varepsilon^{-5}\right)$ time, generates a classical representation of a state $|\tilde{\psi}\rangle$ that is $\varepsilon$-close to $|\psi\rangle$ in trace distance with probability at least $1-\delta$. Furthermore, the algorithm requires only single-copy Clifford measurements and classical post-processing.
\end{lemma}

Next, we provide a lemma that is useful in the proof of the subsequent Theorem~\ref{th:joiningpieces}.
\begin{lemma}[Boosting the probability of success]
\label{le:boosting}
Let $\delta > 0$ and $N' \in \mathbb{N}$. Consider an algorithm $\mathcal{A}$ that succeeds with a probability of $p_{\mathrm{succ}} \geq \frac{3}{4}$. If we execute $\mathcal{A}$ a total of $m \ge \lceil 2N' + 24\log\!\left(\frac{1}{\delta}\right)\rceil$ times, then $\mathcal{A}$ will succeed at least $N'$ times with a probability of at least $1 - \delta$.
\end{lemma}

\begin{proof}
    We will employ a Chernoff bound~\ref{th:ChernoffB} to establish this result. Define binary random variables $\{X_i\}^{m}_{i=1}$ as follows:
    \begin{align}
        X_i = \begin{cases}
            1 & \text{if $\mathcal{A}$ succeeds}, \\ 
            0 & \text{if $\mathcal{A}$ fails}.
        \end{cases}
    \end{align}
    Define $\hat{X}  \coloneqq   \sum^{m}_{i=1}X_i$. We have $\mathbb{E}[\hat{X}] = m p_{\mathrm{succ}}$. Moreover, we aim to upper bound by $\delta$ the probability that $\mathcal{A}$ succeeds fewer than $N'$ times, which is $\operatorname{Pr}\!\left(\hat{X} \leq N'\right)$. We first write it as $\operatorname{Pr}\!\left(\hat{X} \leq N'\right) = \operatorname{Pr}\!\left(\hat{X} \leq \left(1 - \alpha\right) \mathbb{E}[\hat{X}]\right)$,
    where we defined $\alpha  \coloneqq   1 - \frac{N'}{m p_{\mathrm{succ}}}$. Note that $\alpha$ satisfies $\alpha \ge \frac{1}{3}$, if 
    \begin{align}
    \label{eq:condm1}
        m \ge 2N',
    \end{align}
    exploiting the fact that $p_{\mathrm{succ}} \geq \frac{3}{4}$.
    Applying the Chernoff bound, we obtain:
    \begin{align}
        \operatorname{Pr}\!\left(\hat{X} \leq (1-\alpha) \mathbb{E}[\hat{X}]\right) \leq \exp\!\left(-\frac{\alpha^2 \mathbb{E}[\hat{X}]}{2}\right) = \exp\!\left(-\frac{\alpha^2}{2} p_{\mathrm{succ}}m\right).
    \end{align}
    This is upper bounded by $\delta$ if
    \begin{align}
    \label{eq:condm2}
        m \ge \frac{2}{p_{\mathrm{succ}}\alpha^2}\log\!\left(\frac{1}{\delta}\right).
    \end{align}
    Therefore, choosing $m$ as follows satisfies Eq. \eqref{eq:condm1} and Eq. \eqref{eq:condm2}:
    \begin{align}
        m \ge 2N' + \frac{2}{\left(\frac{3}{4}\right)\left(\frac{1}{3}\right)^2}\log\!\left(\frac{1}{\delta}\right) = 2N' + 24\log\!\left(\frac{1}{\delta}\right),
    \end{align}
    where we used the fact that $p_{\mathrm{succ}} \geq \frac{3}{4}$ and $\alpha \ge \frac{1}{3}$.
\end{proof}

\subsection{Joining the pieces: proof of correctness}
We now present the main theorem which puts together the lemmas we have discussed. It demonstrates that to learn $t$-doped fermionic Gaussian states, or more generally $t$-compressible Gaussian states, with $t=\mathcal{O}\!\left(\log(n)\right)$, only resources scaling polynomially in the number of qubits are required.
\begin{theorem}[Efficient learning of $t$-compressible Gaussian states]
\label{th:joiningpieces}
Let $\ket{\psi}$ be a $t$-compressible Gaussian state, and consider $\varepsilon, \delta \in (0,1]$. By utilizing 
\begin{align}
    N & \ge \frac{256 n^5}{\varepsilon^{4}}\log\!\left(\frac{12n^2}{\delta}\right) + 2N_{\mathrm{tom}}\!\left(t,\frac{\varepsilon}{2},\frac{\delta}{3}\right) + 24 \log\!\left(\frac{3}{\delta}\right)
\end{align}
single-copy measurements and
\begin{align}
    T & \ge \mathcal{O}(n^3) + T_{\mathrm{tom}}\!\left(t,\frac{\varepsilon}{2},\frac{\delta}{3}\right)
\end{align}
computational time, Algorithm~\ref{alg:algo} yields a classical representation of a state $|\hat{\psi}\rangle$, satisfying $d_{\mathrm{tr}}(|\hat{\psi}\rangle, \ket{\psi}) \le \varepsilon$ with probability $\ge 1-\delta$.

Here, $N_{\mathrm{tom}}(t,\frac{\varepsilon}{2},\frac{\delta}{3})$ and $T_{\mathrm{tom}}(t,\frac{\varepsilon}{2},\frac{\delta}{3})$ respectively denote the number of copies and computational time sufficient for full state tomography of a $t$-qubit state with an $\varepsilon/2$ accuracy and a failure probability of at most $\delta/3$ (using the notation of Lemma~\ref{le:fasttom}).
\end{theorem}

\begin{proof}[Proof]
The learning procedure is outlined in Algorithm~\ref{alg:algo} in the main text. We now establish its efficiency and correctness.
According to Lemma~\ref{le:samplecompAPPcommuting}, $ N_c(n,\varepsilon_c,\delta/3)$ single copies of $\ket{\psi}$ are sufficient to construct an anti-symmetric real matrix $\hat{C}$ such that $\norm{\hat{C}-C}_\infty\le\varepsilon_c$ with a probability of at least $1-\delta/3$, where $C$ is the correlation matrix of $\rho$. Here, we set $\varepsilon_c \coloneqq  \varepsilon^2/(4(n-t))$.
Then, we can find the orthogonal matrix $\hat{O} \in \mathrm{O}(2n)$ such that it puts $\hat{C}$ in its normal form Eq.\eqref{eq:decomAntisym}, which can be performed in $\mathcal{O}\!\left(n^3\right)$ time. Employing this, we construct the associated Gaussian unitary $\hat{G}$ (a task achievable in time $\mathcal{O}\!\left(n^3\right)$, see~\cite{dias2023classical,zhao2023PhD}). Subsequently, we consider the state $\hat{G}^{\dag}\ket{\psi}$.
As per Lemma~\ref{le:learningLEMMA}, we have 
\begin{align}
    d_{\mathrm{tr}}(\ket{\phi}\otimes\ket{0^{n-t}},\hat{G}^{\dag}\ket{\psi}) \le \frac{\varepsilon}{2},
    \label{eq:tr1}
\end{align}
where $\ket{\phi}\otimes\ket{0^{n-t}}$ corresponds to the post-measurement state obtained by measuring the last $n-t$ qubits of the state $\hat{G}^{\dag}\ket{\psi}$ in the computational basis and obtaining the outcome corresponding to $\ket{0^{n-t}}$. The probability of such an occurrence, as per Lemma~\ref{le:learningLEMMA}, is at least $1-\varepsilon^2/4\ge 3/4$.
Thus, the algorithm proceeds iteratively by querying a total of $m$ copies of $\ket{\psi}$. In each iteration, it applies the unitary $\hat{G}^{\dag}$ to $\ket{\psi}$ and computational basis measurements on the last $n-t$ qubits.  By choosing $m \coloneqq  \lceil 2N_{\mathrm{tom}}(t,\varepsilon/2,\delta/3) + 24 \log(3/\delta)\rceil$, it is guaranteed that the measurements outcome corresponding to $\ket{0^{n-t}}$ occurred at least $N_{\mathrm{tom}}(t,\varepsilon/2,\delta/3)$ with probability at least $1-\delta/3$ (this follows by Lemma~\ref{le:boosting}). 
Applying the tomography algorithm of Lemma~\ref{le:fasttom} to the first $t$ qubits of all the copies where we obtained the outcome corresponding to $\ket{0^{n-t}}$, we obtain an output state $|\hat{\phi}\rangle$ such that it is guaranteed that:
\begin{align}
    d_{\mathrm{tr}}(|\hat{\phi}\rangle,\ket{\phi})\le \frac{\varepsilon}{2},
    \label{eq:tr2}
\end{align}
with a probability of at least $1-\delta/3$. Our output state is $|\hat{\psi}\rangle  \coloneqq   \hat{G}(|\hat{\phi}\rangle \otimes \ket{0^{n-t}})$, and the information about such a state is provided in the output by providing the orthogonal matrix $\hat{O} \in \mathrm{O}(2n)$, which identifies $\hat{G}$, and the $t$-qubit state $|\hat{\phi}\rangle$.

Considering the trace distance between $|\hat{\psi}\rangle$ and $\ket{\psi}$ and applying the triangle inequality with $\hat{G}(|\phi\rangle\otimes \ket{0^{n-t}})$ as the reference state, we have:
\begin{align}
\label{eq:finaleq}
    d_{\mathrm{tr}}(|\hat{\psi}\rangle, \ket{\psi}) \le d_{\mathrm{tr}}(|\hat{\phi}\rangle,\ket{\phi}) + d_{\mathrm{tr}}(\ket{\phi}\otimes \ket{0^{n-t}},\hat{G}^{\dag}\ket{\psi}),
\end{align}
where in the last step we use the unitary-invariance of the trace distance and $d_{\mathrm{tr}}(|\hat{\phi}\rangle\otimes \ket{0^{n-t}},|\phi\rangle\otimes \ket{0^{n-t}})=d_{\mathrm{tr}}(|\hat{\phi}\rangle,\ket{\phi})$.
The algorithm's overall failure probability is contingent on the potential failure of any of the three subroutines—specifically, correlation matrix estimation, measurement of the last $n-t$ qubits, and the tomography protocol. Each subroutine is associated with a failure probability of at most $\delta/3$. Consequently, by the union bound, the algorithm's total failure probability is at most $\delta$.
Utilizing Eqs.\eqref{eq:tr1},\eqref{eq:tr2},\eqref{eq:finaleq} and assuming the case in which the algorithm does not fail, we deduce $d_{\mathrm{tr}}(|\hat{\psi}\rangle, \ket{\psi})\le \varepsilon$.
The overall sample complexity is determined by the number of copies needed to estimate the correlation matrix \(N_c(n,\varepsilon_c,\delta/3)\) plus the copies \(m\) for tomography, i.e., a total number of copies
\begin{align}
    N =\left\lceil \frac{256 n^5}{\varepsilon^{4}}\log\!\left(\frac{12n^2}{\delta}\right) + 2N_{\mathrm{tom}}\!\left(t,\frac{\varepsilon}{2},\frac{\delta}{3}\right) + 24 \log\!\left(\frac{3}{\delta}\right)\right\rceil,
\end{align}
suffices, which is \(\mathcal{O}(\mathrm{poly}(n)+\exp(t))\). On the other hand, the time complexity involves post-processing of the estimated correlation matrix, requiring \(\mathcal{O}(n^3)\) time, and the time-complexity for full-state tomography which is \(\mathcal{O}(\exp(t))\).
\end{proof}

\begin{remark}
    The output state of Algorithm~\ref{alg:algo} is \(|\hat{\psi}\rangle  \coloneqq   \hat{G}(|\hat{\phi}\rangle \otimes \ket{0^{n-t}})\). Specifically, to provide a classical representation of \(|\hat{\psi}\rangle\) in the output, it suffices to give the orthogonal matrix \(\hat{O} \in \mathrm{O}(2n)\) associated with \(\hat{G}\) and the classical description of the \(t\)-qubit state \(|\hat{\phi}\rangle\).
    Therefore, the memory necessary to store the classical description of the state outputted by Algorithm~\ref{alg:algo} is \(\mathcal{O}(\mathrm{poly}(n,2^t))\), similarly to its time and sample complexity.
\end{remark}

\section{Testing $t$-compressible states}
\label{sec:testing}
In this section, we address property testing problem, i.e. the problem of determining whether a state is close or far from the set of states $t$-compressible states (or equivalently we test the Gaussian dimension of a state). We begin by establishing an upper bound on the trace distance between a state and the set of $t$-compressible Gaussian states. Subsequently, we present a lower bound on the same quantity. Finally, we leverage these two bounds to develop a testing algorithm for $t$-compressible Gaussian states. 

We note that the testing problem in the context of fermionic Gaussian states ($t=0$) was also unsolved. However, a forthcoming paper~\cite{Bittel2024testing} addresses the testing problem for general, possibly mixed, fermionic Gaussian states. Here, we generalize the results presented in~\cite{Bittel2024testing} regarding Gaussian testing to the scenario of $t$-compressible Gaussian states.

\subsection{Upper bound on the trace distance from the set of $t$-compressible states}
We observed in Lemma~\ref{prop:suffcond} that when $n-t$ normal eigenvalues of a state's correlation matrix are precisely one, the state is a $t$-compressible state. However, when these eigenvalues are close to one, we may inquire about the existence of a $t$-compressible state in close proximity. This inquiry is formalized in the subsequent Proposition~\ref{prop:apporxtcompr}. 
\begin{proposition}[Check the closeness to a $t$-compressible Gaussian state]
\label{prop:apporxtcompr}
Let $\ket{\psi}$ be a quantum state. Let $\{\lambda_i\}^{n}_{i=1}$ be the normal eigenvalues of its correlation matrix ordered in increasing order. 
Then, there exists a $t$-compressible Gaussian state $\ket{\psi_t}$ such that:
\begin{align}
    d_{\mathrm{tr}}(\ket{\psi_t },\ket{\psi}) \leq \sqrt{\sum^{n}_{k=t+1}\frac{1}{2} (1-\lambda_k)}.
\end{align}
In particular, $\ket{\psi_t }$ can be chosen as $\ket{\psi_t } \coloneqq  G_O (\ket{\phi}\otimes \ket{0^{n-t}})$, where $G_O$ is the Gaussian unitary associated with the orthogonal matrix $O\in \mathrm{O}(2n)$ that puts the correlation matrix of $\ket{\psi}$ in its normal form (Lemma~\ref{le:decSKSYM}), and $\ket{\phi}\otimes \ket{0^{n-t}}$ is the state obtained by projecting the state $G^{\dag}_O\ket{\psi}$ onto the zero state on the last $n-t$ qubits.
\end{proposition}

\begin{proof}
We can define the state $\ket{\psi^{\prime}} \coloneqq  G^{\dag}_{O}\ket{\psi}$. We have: 
\begin{align}
    \Tr(\ketbra{\psi^{\prime}}Z_k)=C(\psi^{\prime})_{2k-1,2k}=\Lambda_{2k-1,2k}=\lambda_k = 1 - (1-\lambda_k),
\end{align}
for each $k\in \{t+1,\dots,n\}$. By using that $Z_k=2\ketbra{0}_k - I $, we have:
\begin{align}
    \Tr(\ketbra{\psi^{\prime}}\ketbra{0}_k)=1 - \frac{(1-\lambda_k)}{2}
\end{align}
By Quantum Union Bound (Lemma~\ref{le:qunionbound}), we have:
\begin{align}
    d_{\mathrm{tr}}(\ket{\psi^{\prime} }, \ket{\phi}\otimes \ket{0^{n-t}}) \leq \sqrt{\sum^{n}_{k=t+1} \frac{(1-\lambda_k)}{2}}.
\end{align}
Therefore, by using the unitarity invariance of the trace-norm, we can conclude. 
\end{proof}
From this, it readily follows that trace distance between the state and the set of $t$-compressible pure Gaussian states $\mathcal{G}_t$ is upper bounded by:
\begin{align}
    \min_{\ket{\phi_t} \in \mathcal{G}_t }d_{\mathrm{tr}}(\ket{\psi}, \ket{\phi_t}) \le \sqrt{(n-t) \frac{(1-\lambda_{t+1})}{2}}.
\end{align}

\subsection{Lower bound on the trace distance from the set of $t$-compressible states}

In this section, we establish a lower bound on the trace distance of a state from the set of $t$-compressible Gaussian states. We denote the set of pure $t$-compressible Gaussian states as $\mathcal{G}_t$.  In the following proof, we follow the derivation presented in Ref.~\cite{Bittel2024testing} for the case of pure Gaussian states ($t=0$), extending it to $t$-compressible states.

\begin{proposition}
\label{prop:lbtrace}
Let $\ket{\psi}$ be a quantum state, and let $\{\lambda_i\}^{n}_{i=1}$ be the normal eigenvalues of its correlation matrix, ordered in increasing order. The lower bound on the trace distance between the state and the set of $t$-compressible pure Gaussian states $\mathcal{G}_t$ is given by:
\begin{align}
    \min_{\ket{\phi_t} \in \mathcal{G}_t }d_{\mathrm{tr}}(\ket{\psi}, \ket{\phi_t}) \ge \frac{1}{2}(1-\lambda_{t+1})
\end{align}
\end{proposition}

\begin{proof}
Consider an arbitrary operator $O$ with $\|O\|_\infty \leq 1$, to be fixed later. Let $\ket{\phi_t} \in \mathcal{G}_t $. Then, we have:
\begin{align}
    d_{\mathrm{tr}}(\ket{\psi}, \ket{\phi_t}) &= \frac{1}{2}\| \ketbra{\psi} - \ketbra{\phi_t} \|_1 \\
    &\ge  \frac{1}{2}|\Tr(O(\ketbra{\psi}- \ketbra{\phi_t}))|, \label{eq:trace_distance_lower_bound}
\end{align}
where in the last step, we used Holder's inequality.

Let $C(\ket{\psi})$ and $C(\ket{\phi_t})$ be the correlation matrices of $\ket{\psi}$ and $\ket{\phi_t}$, respectively. Since $C(\ket{\psi})-C(\ket{\phi_t})$ is real and antisymmetric, it can be brought into its normal form $C(\ket{\psi})-C(\ket{\phi_t})=Q^{T} \Lambda' Q$, where $\Lambda'=\bigoplus^n_{j=1} i \sigma'_j Y$ and $\{ \sigma'_j\}^{n}_{j=1}$ are the normal eigenvalues of $C(\ket{\psi})-C(\ket{\phi_t})$ and $Q \in \mathrm{O}(2n)$ is an orthogonal matrix (Lemma~\ref{le:decSKSYM}).

Now, choose the operator $O$ in the form $O=U^{\dagger}_{Q}i\gamma_{2k-1}\gamma_{2k}U_{Q}$, where $ k\in[n]$ and $U_{Q}$ is a Gaussian unitary associated with the orthogonal matrix $Q \in \mathrm{O}(2n)$. Note that $\|O\|_\infty=1$. Fix $k$ as a value of $j$ that maximizes $|\sigma_j^{\prime}|$. Thus, we have:
\begin{align}
    |\Tr(O(\ketbra{\psi}- \ketbra{\phi_t}))| &= 
    |[C(U_Q\ket{\psi})-C(U_Q\ket{\phi_t})]_{2k-1,2k}| \\
    &= 
    |[Q(C(\ket{\psi})-C(\ket{\phi_t}))Q^T]_{2k-1,2k}| \\
    &=  
    |\sigma_k^{\prime}| \\
    &= \max_{j\in [n]} 
      |\sigma_j^{\prime}|\\
    &= \|C(\ket{\psi})-C(\ket{\phi_t})\|_\infty,
\end{align}
where in the first step we used the definition of correlation matrix, in the second step we used Lemma~\ref{prop:transfFGU} and in the last step we used the fact that the largest normal eigenvalues of an anti-symmetric matrix corresponds to the infinity norm of the matrix.
Therefore, combining with \eqref{eq:trace_distance_lower_bound}, we have:
\begin{align}
\label{eq:vardeftr}
    d_{\mathrm{tr}}(\ket{\psi}, \ket{\phi_t}) \ge \frac{1}{2}  \|C(\ket{\psi})-C(\ket{\phi_t})\|_\infty.
\end{align}
Now, by applying Lemma~\ref{le:lbantisymm}, we have:
\begin{align}
    \|C(\ket{\psi})-C(\ket{\phi_t})\|_\infty \ge |\lambda(C(\ket{\psi})_{t+1}- \lambda(C(\ket{\phi_t})_{t+1}|,
\end{align}
where $\lambda(C)_{t+1}$ denotes the $t+1$-th smallest normal eigenvalue of a correlation matrix $C$.
Since $\ket{\phi_t}$ is a $t$-compressible Gaussian state, its Gaussian dimension is $n-t$ (because of Proposition~\ref{prop:equiv}), hence its $t+1$-th smallest normal eigenvalue must be one. Therefore, the desired lower bound is obtained by taking the minimum over $\ket{\phi_t} \in \mathcal{G}_t$ .
\end{proof}

\subsection{Testing the Gaussian dimension of a state}
We present an efficient algorithm (Algorithm~\ref{alg:algoTEST}) for property testing of $t$-compressible states, where $\mathcal{G}_t$ represents the set of $n$-qubits $t$-compressible Gaussian states, or equivalently the set of states with $n-t$ Gaussian dimension. The algorithm takes copies of a state $\ket{\psi}$ as input and determines whether $\min_{\ket{\phi_t} \in \mathcal{G}_t } d_{\mathrm{tr}}(\ket{\psi}, \ket{\phi_t})\le \varepsilon_A$ or $\min_{\ket{\phi_t} \in \mathcal{G}_t } d_{\mathrm{tr}}(\ket{\psi}, \ket{\phi_t})\ge \varepsilon_B$, with the promise that one of these cases is true, where $\varepsilon_B>\varepsilon_A\ge 0$. 

We provide the details of the testing algorithm in Algorithm~\ref{alg:algoTEST}. The correctness of this algorithm is established by the following theorem, the proof of which follows the same steps as proofs presented in~\cite{Bittel2024testing}.
\begin{algorithm}[h]
\caption{Property testing algorithm for $t$-compressible Gaussian states}
\label{alg:algoTEST}
\KwIn{
    Error thresholds $\varepsilon_A, \varepsilon_B$, failure probability $\delta$. $N:=\lceil 16(n^3/\varepsilon_{\mathrm{corr}}^{2}) \log(4n^2/\delta)\rceil$ copies of the state $\ket{\psi}$, where $\varepsilon_{\mathrm{corr}}=(\frac{\varepsilon^2_B}{n-t}-\varepsilon_A)$. Let $\varepsilon_{\mathrm{test}}:=\left(\frac{\varepsilon^2_B}{n-t}+\varepsilon_A\right)$.
}
\KwOut{
    Output either ``$\varepsilon_A$-close to the set of $t$-compressible states" or ``$\varepsilon_B$-far from $t$-compressible states set".
}

\textbf{Step 1:} Estimate the entries of the correlation matrix using $N$ single-copy measurements, resulting in the estimated $2n \times 2n$ matrix $\hat{\Gamma}$\;

\textbf{Step 2:} Find $\{\hat{\lambda}_k\}^n_{k=1}$, which corresponds to the ordered normal eigenvalues of $\hat{\Gamma}$\;

\textbf{Step 3:} 
\If{$\hat{\lambda}_{t+1} \ge 1 - \varepsilon_{\mathrm{test}}$}{
    \textbf{Output:} ``$\varepsilon_A$-close to the set of $t$-compressible states".
}
\Else{
    \textbf{Output:} ``$\varepsilon_B$-far from $t$-compressible states set."
}
\end{algorithm}

\begin{theorem}[Efficient $t$-compressible Gaussian testing]\label{th:proofofefficientpurestatetesting}
    Let $\ket{\psi}$ be an $n$-qubit pure state. Assume $\varepsilon_B, \varepsilon_A \in [0,1]$ such that $\varepsilon_B > \sqrt{(n-t)\varepsilon_A}$, $\delta \in (0,1]$, and $\varepsilon_{\mathrm{corr}}= (\frac{\varepsilon^2_B}{n-t}-\varepsilon_A)$. Assume that $\ket{\psi}$ is such that $\min_{\ket{\phi_t} \in \mathcal{G}_t } d_{\mathrm{tr}}(\ket{\psi}, \ket{\phi_t})\le \varepsilon_A$ or $\min_{\ket{\phi_t} \in \mathcal{G}_t } d_{\mathrm{tr}}(\ket{\psi}, \ket{\phi_t})> \varepsilon_B$.  
    Then, Algorithm~\ref{alg:algoTEST} can discriminate between these two scenarios using $N=\lceil 16(n^3/\varepsilon_{\mathrm{corr}}^{2}) \log(4n^2/\delta)\rceil $ single-copy measurements of the state $\ket{\psi}$ with a probability of success at least $1-\delta$. 
\end{theorem}

\begin{proof}
    Let $\varepsilon_{\mathrm{corr}}>0$ be an accuracy parameter to be fixed later. By Lemma~\ref{le:samplecompAPPcommuting}, with $N \ge 8(n^3/\varepsilon_{\mathrm{corr}}^{2}) \log(4n^2/\delta)$ single-copy measurements, we can find a matrix $\hat{\Gamma}$ such that, with probability at least $1-\delta$, it holds that $\norm{\hat{\Gamma}-\Gamma(\ket{\psi})}_\infty \le \varepsilon_{\mathrm{corr}}$. This implies that for all $k\in[n]$, we have $|\hat{\lambda}_k - \lambda_k| \le \varepsilon_{\mathrm{corr}}$, where $\{ \hat{\lambda}_k \}^{n}_{k=1}, \{\lambda_k\}^{n}_{k=1}$ are the normal eigenvalues of $\hat{\Gamma}$ and $\Gamma(\ket{\psi})$ respectively. 
    Let $\varepsilon_{\mathrm{test}}$ be a parameter to fix later. If $ \hat{\lambda}_{t+1}\ge 1-\varepsilon_{\mathrm{test}}$, we aim to show that $\min_{\ket{\phi_t} \in \mathcal{G}_t } d_{\mathrm{tr}}(\ket{\psi}, \ket{\phi_t})\le \varepsilon_B$, otherwise, we aim to show that $\min_{\ket{\phi_t} \in \mathcal{G}_t } d_{\mathrm{tr}}(\ket{\psi}, \ket{\phi_t})>  \varepsilon_A$. Thus, we first assume that $ \hat{\lambda}_{t+1}\ge 1-\varepsilon_{\mathrm{test}}$.
   From Lemma~\ref{prop:apporxtcompr}, we have that:
\begin{align}
    \min_{\ket{\phi_t} \in \mathcal{G}_t } d_{\mathrm{tr}}(\ket{\psi}, \ket{\phi_t}) 
    \le \sqrt{\frac{(n-t)}{2} (1-\lambda_{t+1})}\le \sqrt{\frac{(n-t)}{2} (1-\hat{\lambda}_{t+1} + \varepsilon_{\mathrm{corr}} )}\le\sqrt{\frac{(n-t)}{2} (\varepsilon_{\mathrm{test}} + \varepsilon_{\mathrm{corr}} )} .
\end{align}
 Therefore we need to ensure that $\sqrt{\frac{(n-t)}{2} (\varepsilon_{\mathrm{test}} + \varepsilon_{\mathrm{corr}} )} \le \varepsilon_B$.
Let us analyze now the case in which $\hat{\lambda}_{t+1}< 1-\varepsilon_{\mathrm{test}}$. From Lemma~\ref{prop:lbtrace}, we have:
\begin{align}
    \min_{\ket{\phi_t} \in \mathcal{G}_t }d_{\mathrm{tr}}(\ket{\psi}, \ket{\phi_t}) \ge \frac{1}{2}(1-\lambda_{t+1})\ge   \frac{1}{2}(1-\hat{\lambda}_{t+1}-\varepsilon_{\mathrm{corr}})> \frac{1}{2}(\varepsilon_{\mathrm{test}}-\varepsilon_{\mathrm{corr}})
\end{align}
Therefore, we impose that $\frac{\varepsilon_{\mathrm{test}} - \varepsilon_{\mathrm{corr}}}{2} \ge \varepsilon_{A}$.
The two mentioned inequalities are satisfied by choosing $\varepsilon_{\mathrm{test}}= (\frac{\varepsilon^2_B}{n-t}+\varepsilon_A)$ and $\varepsilon_{\mathrm{corr}}= (\frac{\varepsilon^2_B}{n-t}-\varepsilon_A)$, by assuming that $\varepsilon_B > \sqrt{(n-t)\varepsilon_A}$.
\end{proof}

\section{Pseudorandomness from $t$-doped Gaussian states and time complexity lower bound}
We have presented an algorithm for learning $t$-doped fermionic Gaussian states with a time complexity scaling as $\mathcal{O}(\mathrm{poly}(n,2^t))$ (assuming that the Majorana locality $\kappa$ of each non-Gaussian gate is constant in the number of qubits). Thus, as long as $t=\mathcal{O}(\log(n))$, the algorithm is efficient, i.e., the total run-time is polynomial in the number of qubits.
In this section, we delve into establishing lower bounds on the time complexity for learning $t$-doped fermionic Gaussian states. We begin by demonstrating that, under a well-believed cryptography assumption~\cite{LWE0,regev2024lattices,ShallowSrini,diakonikolas2022cryptographic,aggarwal2022lattice,ananth2023revocable}, no algorithm can learn $t$-doped fermionic Gaussian states with time complexity scaling in $t$ as $\mathcal{O}(\mathrm{poly}(t))$. This rules out efficient algorithms if $t$ scales polynomially with the number of qubits $n$.
However, under a stronger cryptography assumption~\cite{ShallowSrini,diakonikolas2022cryptographic,ananth2023revocable,gupte2022continuous}, we can establish that when $t$ scales \emph{slightly} faster than logarithmically in the number of qubits $n$, i.e., $t=\tilde{\omega}(\log(n))$, there exists no efficient algorithm to learn $t$-doped fermionic Gaussian states, where we defined $\tilde{\omega}(\log(n))\coloneqq\omega\!\left(\log(n)\mathrm{polyloglog}(n)\right)$.

\subsection{$t$-doped Gaussian states cannot be learned in polynomial time in $t$}
Our cryptography assumption relies on the conjecture that a specific problem, namely ``Learning With Errors over Rings'' \class{RingLWE}~\cite{LWE0}, is hard to solve by quantum computers~\cite{LWE0,regev2024lattices,ShallowSrini,diakonikolas2022cryptographic,aggarwal2022lattice,ananth2023revocable}. Detailed definitions and discussions about \class{RingLWE} can be found in~\cite{LWE0}. Informally, \class{RingLWE} is a variant of the more general ``Learning With Errors" (LWE) problem specialized to polynomial rings over finite fields, where the LWE problem is to distinguish random linear equations, perturbed by a small amount of noise, from truly uniform ones.

Crucial to our proof is a lemma, adapted from~\cite{zhao2023learning} (which strongly relies on previous work~\cite{ShallowSrini,brakerski2019pseudo,Ji_2018}), presented below:

\begin{lemma}[Adapted from Theorem 14 in~\cite{zhao2023learning}]
\label{le:caro1}
    Assume that \class{RingLWE} cannot be solved by quantum computers in polynomial time. Then, there exists a set $\mathcal{S}_{\mathrm{PRS}}$ of $k$-qubits pure quantum states, known as pseudorandom quantum states, with the following properties:
    \begin{enumerate}
        \item Any state in $\mathcal{S}_{\mathrm{PRS}}$ can be prepared using $\tilde{\mathcal{O}}(k)$ Toffoli and Hadamard gates (here, $\tilde{\mathcal{O}}(\cdot)$ hides $\mathrm{polylog(\cdot)}$ factors).
        \item States in the set $\mathcal{S}_{\mathrm{PRS}}$ cannot be learned in time complexity $\mathcal{O}(\mathrm{poly}(t))$ by quantum computers. More specifically, let $\rho \in \mathcal{S}_{\mathrm{PRS}}$ be an unknown quantum state. Then there is no quantum algorithm that, using $\mathcal{O}(\mathrm{poly}(k))$ copies of $\rho$ and computational time, with probability at least $2/3$ outputs a classical description of a state $\hat{\rho}$ which can be prepared in polynomial time on a quantum computer such that $d_{\mathrm{tr}}(\rho,\hat{\rho})\le 1/8$.
    \end{enumerate}
\end{lemma}

Now we use this Lemma to show that there is no algorithm for learning $t$-doped states with a $\mathcal{O}(\mathrm{poly}(t))$ computational time scaling in $t$.

\begin{proposition}[No $\mathrm{poly}(t)$ algorithm to learn $t$-doped states]
\label{prop:nopolyt}
    Assume that \class{RingLWE} cannot be solved by quantum computers in polynomial time.
    Then there is no quantum algorithm that, given access to copies of a $t$-doped fermionic Gaussian $n$-qubits state $\rho$ and with a time complexity scaling in $t$ as $\mathcal{O}(\mathrm{poly}(t))$, outputs a classical description of a quantum state $\hat{\rho}$ such that it can be prepared in polynomial time on a quantum computer and, with probability at least $2/3$, it holds that $d_{\mathrm{tr}}(\rho,\hat{\rho})\le 1/8$.
\end{proposition}

\begin{proof}
Consider an $n$-qubit state of the form $\ket{\phi}\otimes \ket{0^{n-k}}$, where $\ket{\phi}$ is a $k$-qubit state in the set $\mathcal{S}_{\mathrm{PRS}}$ defined in Lemma~\ref{le:caro1}. Let $\ket{\phi}=U_{\mathrm{PRS}}\ket{0^k}$, prepared by a unitary $U_{\mathrm{PRS}}$ that can be implemented using $\tilde{\mathcal{O}}(k)$ Hadamard and Toffoli gates, as per Lemma~\ref{le:caro1}.
The Toffoli gates can be implemented in turn using Hadamard, CNOT, T-gates, and T-gates inverse. For each gate in the system, utilizing a standard SWAP-exchange trick, we can make any gate acting on some qubits of the system to act only on the first two qubits through a cascade of SWAP gates. This incurs a total overhead factor of $\mathcal{O}(k)$ in the total number of gates.
The SWAP gates between nearest-neighbor qubits are local non-Gaussian gates, since they can be expressed as $\text{SWAP}_{i,i+1}=e^{-i\frac{\pi}{4}}\exp(i\frac{\pi}{4}(X_iX_{i+1}+Y_iY_{i+1}+Z_iZ_{i+1}))$, which have Majorana locality equal to $4$ (due to the Jordan-Wigner mapping). Also the other gates, now acting on the first two qubits, are (possibly) non-Gaussian gates with Majorana locality at most $4$. This is because the Pauli operators in the generator of each gate can be expressed in terms of Majorana operators via the Jordan-Wigner transformation.
This results in a total of $t=\tilde{\mathcal{O}}(k^2)$ local non-Gaussian gates required to prepare the state. Thus, the $\ket{\phi}$ is a $t$-doped Gaussian state with $\kappa=4$ Majorana local non-Gaussian gates. Now, due to Lemma~\ref{le:caro1}, there exists no learning algorithm to learn such a state in time $\mathcal{O}(\mathrm{poly}(k))=\mathcal{O}(\mathrm{poly}(t))$, with error less than $1/8$ and probability of success at least $2/3$.
\end{proof}

\subsection{Learning $\tilde{\omega}(\log(n))$-doped Gaussian states is hard}
The previous Proposition rules out efficient algorithms when $t=\Omega(\mathrm{poly}(n))$. However, our proposed algorithm is not anymore efficient when $t=\tilde{\omega}(\log(n))$, while it is efficient for $t=\mathcal{O}(\log(n))$. 
If we were to make the stronger assumption, namely that \class{RingLWE} cannot be solved by quantum computers in sub-exponential time~\cite{ShallowSrini,diakonikolas2022cryptographic,ananth2023revocable,gupte2022continuous}, then we can rule out that efficient algorithm for $t=\tilde{\omega}(\log(n))$ exists (which means faster than $\omega(\log(n))$ but hiding $\mathrm{polyloglog}$ factors). For showing this, we need first the following Lemma adapted by~\cite{zhao2023learning}.
\begin{lemma}[Adapted from Theorem 14 and 15 in~\cite{zhao2023learning}]
\label{le:caro2}
    Assume that \class{RingLWE} cannot be solved by quantum computers in sub-exponential time. Then, there exists a set $\mathcal{S}_{\mathrm{PRS}}$ of $k$-qubit pure quantum states, known as pseudorandom quantum states, with the following properties:
    \begin{enumerate}
        \item Any state in $\mathcal{S}_{\mathrm{PRS}}$ can be prepared using $\tilde{\mathcal{O}}(k)$ Toffoli and Hadamard gates (here, $\tilde{\mathcal{O}}(\cdot)$ hides $\mathrm{polylog}(\cdot)$ factors).
        \item Any algorithm to learn states from the set $\mathcal{S}_{\mathrm{PRS}}$ must have $\exp(\Omega(k))$ time complexity. More specifically, let $\rho \in \mathcal{S}_{\mathrm{PRS}}$ be an unknown quantum state. Then any quantum algorithm that, by querying copies of $\rho$, with probability at least $2/3$ outputs a classical description of a state $\hat{\rho}$ which can be prepared in sub-exponential time on a quantum computer such that $d_{\mathrm{tr}}(\rho,\hat{\rho})\le 1/8$, must have $\exp(\Omega(k))$ time-complexity.
    \end{enumerate}
\end{lemma}
Now we use this to prove the following Proposition, also stated informally in the main text as Theorem~\ref{thm:lower_bound_informal}. It reaches stronger conclusion than Proposition~\ref{prop:nopolyt}, but at the cost of stronger cryptography assumption. The idea of the following proof is to use a more compact qubits-to-fermion mapping than Jordan-Wigner, namely a modification of the one introduced by Kitaev \cite{Kitaev_2006}, which would allow to create pseudorandom quantum states with no overhead in the number of non-Gaussian gates compared to the number of Toffoli and Hadamard gates. 
In the following, we recall that $\tilde{\omega}(\log(n))$ is defined as $\tilde{\omega}(\log(n)):=\omega(\log(n)\mathrm{polyloglog}(n))$).
\begin{proposition}[Learning $\tilde{\omega}(\log(n))$-doped Gaussian states is hard]
\label{prop:exp_lower_bound}
Assume that \class{RingLWE} cannot be solved by quantum computers in sub-exponential time. Then, there is no efficient (i.e., $\mathcal{O}(\mathrm{poly}(n))$ time) quantum algorithm that, by querying copies of a $\tilde{\omega}(\log(n))$-doped Gaussian state $\rho$, with probability at least $2/3$ outputs a description of a state $\hat{\rho}$ which can be prepared in polynomial time on a quantum computer such that $d_{\mathrm{tr}}(\rho,\hat{\rho})\le 1/8$. 
\end{proposition}

\begin{proof}
Following a similar approach as in the proof of Proposition~\ref{prop:nopolyt}, we begin by defining a state $\ket{\phi}\otimes \ket{0^{n-k}}$, where $\ket{\phi}=U_{\mathrm{PRS}}\ket{0^k}$ represents a $k$-qubit state in the set $\mathcal{S}_{\mathrm{PRS}}$ as defined in Lemma~\ref{le:caro2}. This state can be efficiently prepared using a unitary $U_{\mathrm{PRS}}$ with $\tilde{\mathcal{O}}(k)$ Hadamard and Toffoli gates, which, in turn, can be implemented using Hadamard gates, CNOTs, T-gates, and their inverses.
In contrast to the previous proof in Proposition~\ref{prop:nopolyt}, an application of the same argument would not yield the desired conclusion due to the unfavorable quadratic overhead introduced by the SWAP-exchange trick in the number of non-Gaussian gates. That trick was necessary due to the use of Jordan-Wigner transformation. However, we can employ a more efficient qubits-to-fermions mapping, specifically a modified version of the one introduced by Kitaev \cite{Kitaev_2006}. Hereafter, we will refer to it as the ``Kitaev encoding."
Our objective is to construct a fermionic state encoding $\ket{\phi}$ using a circuit of size $\tilde{\mathcal{O}}(k)$ composed of Gaussian and $\kappa=4$ local non-Gaussian gates.
We employ a mapping of $k$ qubits into $2k$ fermionic modes using Majorana operators $\{\gamma_{\alpha,j}\,|\, j\in[k],\, \alpha\in\{0,\,x,\,y,\,z\}\}$. These $4k$ Majorana operators are defined in terms of $4k$ Pauli strings via Jordan-Wigner, with an arbitrarily fixed operator ordering.
We are now going to leverage the formalism and basics of stabilizer codes, for more in-depth information, refer to~\cite{Roffe_2019}.
The Kitaev encoding involves defining a stabilizer code of $2k$ physical qubits encoded in $k$ logical qubits, characterized by the following $k$ stabilizer generators $\{s_j\}^k_{j=1}$ and logical Pauli operators $\{X^{\rm KE}_j,Z^{\rm KE}_j\}^k_{j=1}$:
\begin{align}
s_j&:=\gamma_{0,j} \gamma_{x,j} \gamma_{y,j} \gamma_{z,j},\\
X^{\rm KE}_j&:=i\gamma_{y,j} \gamma_{z,j},\\
Z^{\rm KE}_j&:=i\gamma_{x,j} \gamma_{y,j}.
\end{align}
for each $j\in[k]$. Note that these operators explicitly satisfy the algebraic conditions on stabilizer generators and logical Pauli operators.
The Kitaev encoding is associated to a Clifford transformation $V_{\rm KE}$ such that: 
\begin{align}
    V_{\rm KE}X_j V_{\rm KE}^{\dagger} &= X^{\rm KE}_j,\\
    V_{\rm KE}Z_j V_{\rm KE}^{\dagger} &= Z^{\rm KE}_j,
    \label{eq:Zmap}
\end{align}
for each $j\in[k]$, and:
\begin{align}
    V_{\rm KE}Z_j V_{\rm KE}^{\dagger} &= s_j,
\end{align}
for each $j\in \{k+1,\dots, 2k\}$. The last equation ensures that $V_{\rm KE}\ket{0^{2k}}$ is an eigenstate with $+1$ eigenvalue for each of the stabilizer generators $\{s_j\}^k_{j=1}$, while Eq.\eqref{eq:Zmap} implies that $V_{\rm KE}\ket{0^{2k}}$ is an eigenstate with $+1$ eigenvalue for each $\{Z^{\rm KE}_j\}^k_{j=1}$. Thus, $V_{\rm KE}\ket{0^{2k}}$ is a valid ``logical zero" stabilizer state. 
Exploiting the fact that $V_{\rm KE}\ket{0^{2k}}$ is an eigenstate with $+1$ eigenvalues of $\{Z^{\rm KE}_j\}^k_{j=1}$ and $\{s_j Z^{\rm KE}_j\}^k_{j=1}$, its density matrix can be written as:
\begin{align}
 V_{\rm KE}\ketbra{0^{2k}}V^{\dag}_{\rm KE}=\prod^{k}_{j=1} \left(\frac{I+ s_jZ^{\rm KE}_j }{2}\right)  \prod^{2k}_{j=k} \left(\frac{I+ Z^{\rm KE}_j }{2}\right)&=\prod^{k}_{j=1} \left(\frac{I- i\gamma_{0,j}   \gamma_{z,j} }{2}\right)\prod^{2k}_{j=k} \left(\frac{I+ i\gamma_{x,j} \gamma_{y,j} }{2}\right).
\end{align}
From this, we observe that $V_{\rm KE}\ket{0^{2k}}$ is a fermionic Gaussian state because it can be written in the form of Eq.\eqref{eq:densityGaus}, noting that signed permutation matrices are orthogonal matrices. Moreover, we denote $Y^{\rm KE}_j:=-iZ^{\rm KE}_j X^{\rm KE}_j=i\gamma_{x,j} \gamma_{z,j}$.
The Kitaev encoding, as defined, ensures that the local qubit gates are mapped onto local fermionic ones. In particular, the gates of the circuit $U_{\rm PRS}$ that prepares the pseudorandom state $\ket{\phi}=U_{\rm PRS}\ket{0^k}$--- the Hadamard $\operatorname{H}$, $\operatorname{CNOT}$, and $\operatorname{T}$-gate --- are, up to an overall phase, mapped onto
\begin{align}
   \mathrm{H}^{\rm KE}_{j}&:=  V_{\rm KE} \mathrm{H}_{j} V_{\rm KE}^{\dagger}=V_{\rm KE}\left( Z_j \frac{I+ i Y_j}{\sqrt{2}}\right) V_{\rm KE}^{\dagger}= Z^{\rm KE}_j \frac{I+ i Y^{\rm KE}_j}{\sqrt{2}}=e^{-\frac{\pi}{2} \gamma_{x,j} \gamma_{y,j}}e^{-\frac{\pi}{4} \gamma_{z,j} \gamma_{x,j}},\\
   \mathrm{CNOT}_{j,l}^{\rm KE}&:= V_{\rm KE} \mathrm{CNOT}_{j,l} V_{\rm KE}^{\dagger}= e^{i\frac{\pi}{4} (1-Z^{\rm KE}_j)(I-X^{\rm KE}_l)}=e^{i\frac{\pi}{4} (I-i\gamma_{x,j} \gamma_{y,j})(I-i\gamma_{y,l} \gamma_{z,l})},\\
    \mathrm{T}^{\rm KE}_{j}&:=V_{\rm KE} \mathrm{T}_{j} V_{\rm KE}^{\dagger}= e^{i\frac{\pi}{8} Z^{\rm KE}_j} = e^{-\frac{\pi}{8} \gamma_{x,j} \gamma_{y,j}},
\end{align}
for each $j \neq l \in [k]$.
The only non-Gaussian among these is the encoding of the $\mathrm{CNOT}$ gate, which has Majorana locality $\kappa=4$.
This implies that the encoding of the circuit $U_{\mathrm{PRS}}$, i.e. $  U^{\rm KE}_{\mathrm{PRS}}:= V_{\rm KE} U_{\mathrm{PRS}}V^{\dag}_{\rm KE}$, is a $t$-doped fermionic Gaussian unitary with local non-Gaussian gates and $t=\tilde{\mathcal{O}}(k)$.
Thus, we have that the Kitaev encoding of the pseudorandom state $\ket{\phi}=U_{\mathrm{PRS}}\ket{0^k}$ is:
\begin{align}
    \ket{\phi}_{\rm KE}:=V_{\rm KE}\ket{\phi}\otimes \ket{0^k}=V_{\rm KE}U_{\mathrm{PRS}} \ket{0^{2k}}= U^{\rm KE}_{\mathrm{PRS}}V_{\rm KE}\ket{0^{2k}}.
\end{align}
Since $V_{\rm KE}\ket{0^{2k}}$ is a Gaussian state and $U^{\rm KE}_{\mathrm{PRS}}$ is a $t$-doped Gaussian unitary with $t=\tilde{\mathcal{O}}(k)$, then $ V_{\rm KE}\ket{\phi}\otimes \ket{0^k}$ is a $t$-doped Gaussian state with $t=\tilde{\mathcal{O}}(k)$. The Majorana locality of each non-Gaussian gate is at most $4$.

Using an arbitrary algorithm $\mathcal{A}$ for learning a $t$-doped fermionic Gaussian state, we will now define a protocol for learning the pseudorandom state $\ket{\phi}$. Given a copy of a state $\ket{\phi}$ on $k$ qubits, we use $k$ auxiliary qubits in state $\ket{0}$ and apply the Clifford transformation $V_{\rm KE}$. This means it can be produced by a circuit with $O(k^2)$ $2$-qubit gates \cite{Dehaene_2003}. The resulting $\ket{\phi}_{\rm KE}$ can be input to $\mathcal{A}$ as a copy of a $t$-doped fermionic Gaussian state for $t=\tilde{\mathcal{O}}(k)$. Using the number of copies of $\ket{\phi}$ given by sample complexity of $\mathcal{A}$, we learn a description of a state $\hat{\rho}_{\rm KE}$ which, with probability at least $2/3$, satisfies:
\begin{align}
    d_{\mathrm{tr}}(\hat{\rho}_{\rm KE},\rho_{\rm KE})\le \frac{1}{8},
\end{align}
where we defined $\rho_{\rm KE}$ to be the density matrix associated with $\ket{\phi}_{\rm KE}$. 
By defining $\hat{\rho}:=\Tr_{\{k+1,\dots,2k\}}\left(V^{\dag}_{\rm KE}\hat{\rho}_{\rm KE}V_{\rm KE}\right)$, where $\Tr_{\{k+1,\dots,2k\}}(\cdot)$ indicates the partial trace with respect to the qubits $\{k+1,\dots,2k\}$, we also have:
\begin{align}
    d_{\mathrm{tr}}( \hat{\rho},\ketbra{\phi}) \le d_{\mathrm{tr}}(V^{\dag}_{\rm KE}\hat{\rho}_{\rm KE}V_{\rm KE},\ketbra{\phi}\otimes\ketbra{0^k}{0^k})=d_{\mathrm{tr}}(V^{\dag}_{\rm KE}\hat{\rho}_{\rm KE}V_{\rm KE},V^{\dag}_{\rm KE}\rho_{\rm KE}V_{\rm KE})=d_{\mathrm{tr}}(\hat{\rho}_{\rm KE},\rho_{\rm KE})\le \frac{1}{8},
\end{align}
where in the first step we used that the partial trace does not increase the trace distance between two states~\cite{Wilde_2017}. Hence, we found a state $\hat{\rho}$ which is in trace distance close to the target state $\ketbra{\phi}$.
To recap, we produced the learning algorithm for a pseudorandom state $\ket{\phi}$ from a learning algorithm for $t$-doped fermionic Gaussian states. The pseudorandom state learning algorithm has the same sample complexity as the fermionic one, and the time complexity $T_{\rm PRS}=S_{\rm f}\cdot O(k^2) + T_{\rm f}$ where $T_{\rm f}$ and $S_{\rm f}$ are time and sample complexity of the fermionic learner. If there is a fermionic learner whose time and sample complexity scale subexponentially in $k$, the same property carries over to the pseudorandom states learner. By Lemma~\ref{le:caro2}, this would contradict the cryptographic assumption that \class{RingLWE} cannot be solved by quantum computers in sub-exponential time. Hence, the time complexity of the fermionic learner needs to be $\exp(\Omega(k))$. If $k=\omega(\log(n))$, then this implies that any algorithm to learn $t$-doped fermionic Gaussian states with $t=\tilde{\mathcal{O}}(k)=\omega(\log(n)\mathrm{polyloglog}(n))$ must be inefficient, i.e., its time complexity must be $\omega(\mathrm{poly}(n))$.
\end{proof}

\section{Generalization to the mixed-state scenario and noise-robustness of the algorithm}
\label{sec:genmixed}
Many of the concepts introduced in our work can be extended to the more realistic mixed-state scenario. We start by providing the definition for $t$-compressible mixed states.
\begin{definition}[Mixed $t$-compressible state]
   A possibly mixed state $\rho$ is a $t$-compressible (Gaussian) state if and only if it can be represented as 
   \begin{align}
        \rho = G(\sigma \otimes \ketbra{0^{n-t}})G^{\dag},
   \end{align}
    where $G$ is a Gaussian operation, and $\sigma$ is a state supported solely on the first $t$ qubits.
\end{definition}

As in the pure-state scenario (Proposition~\ref{prop:equiv}), saying that a quantum state is $t$-compressible is equivalent to the fact that at least $n-t$ normal eigenvalues of the state's correlation matrix are equal to one (i.e., its Gaussian dimension is $\ge n-t$ and its Gaussian nullity is $\le t$).
\begin{proposition}[Equivalent definition of $t$-compressible state]
\label{prop:equiv2}
    A possibly mixed $n$-qubit state $\rho$ is $t$-compressible if and only if at least $n-t$ normal eigenvalues of the correlation matrix of $\rho$ are equal to one.
\end{proposition}
The proof is analogous to that of Proposition~\ref{prop:equiv}.

For the tomography algorithm, analogous guarantees extend from the $t$-compressible pure-case to the $t$-compressible mixed-state scenario. In this case, the unknown state is promised to be a $t$-compressible Gaussian (possibly mixed) state. However, in this mixed-state scenario, the algorithm is simpler, since measuring the last $n-t$ qubits in the computational basis is not necessary because we do not require the output state to be pure, i.e., step 5 in Algorithm~\ref{alg:algo} is not necessary. Instead, the full state tomography of the first $t$ qubits can be applied right after having applied the Gaussian operation to the state. We explicitly show this in Algorithm~\ref{alg:algoSMappr} of the following subsection, whose proof of correctness is shown in the (more general) Theorem~\ref{th:joiningpiecesMIXEDApp}.

Furthermore, the following inequalities can be helpful to test how much a possibly mixed state $\rho$ deviates (in trace distance) from the set of mixed $t$-compressible Gaussian states.
\begin{proposition}[Minimum distance between a possibly mixed state and the set of $t$-compressible states]
\label{prop:lbubmixedtrace}
Let $\rho$ be a possibly mixed quantum state, and let $\{\lambda_i\}_{i=1}^{n}$ be the normal eigenvalues of its correlation matrix, ordered in increasing order. The following trace distance upper and lower bounds between the state and the set of $t$-compressible mixed Gaussian states $\mathcal{G}_t$ hold:
\begin{align}
   \sqrt{\sum_{k=t+1}^{n} \frac{(1-\lambda_k)}{2}}  \ge \min_{\sigma_t \in \mathcal{G}_t } d_{\mathrm{tr}}(\rho, \sigma_t) \ge \frac{1}{2}(1-\lambda_{t+1}).
\end{align}
\end{proposition}
The proof of this theorem is analogous to those of Proposition~\ref{prop:apporxtcompr} and Proposition~\ref{prop:lbtrace}, which held for the pure-state scenario but readily apply to the more general mixed-state scenario.

Thus, with the previous inequalities in hand, we can also provide a property testing algorithm that applies to a possibly mixed quantum state.
\begin{theorem}[Efficient $t$-compressible Gaussian testing]\label{th:proofofefficientpurestatetestingMIXED}
    Let $\rho$ be an $n$-qubit possibly mixed quantum state. Assume $\varepsilon_B, \varepsilon_A \in [0,1]$ such that $\varepsilon_B > \sqrt{(n-t)\varepsilon_A}$, $\delta \in (0,1]$, and $\varepsilon_{\mathrm{corr}}= \left(\frac{\varepsilon^2_B}{n-t}-\varepsilon_A\right)$. Assume that $\rho$ is such that $\min_{\sigma_t \in \mathcal{G}_t } d_{\mathrm{tr}}(\rho, \sigma_t)\le \varepsilon_A$ or $\min_{\sigma_t \in \mathcal{G}_t } d_{\mathrm{tr}}(\rho, \sigma_t)> \varepsilon_B$.  
    Then, Algorithm~\ref{alg:algoTEST} can discriminate between these two scenarios using $N=\lceil 16(n^3/\varepsilon_{\mathrm{corr}}^{2}) \log(4n^2/\delta)\rceil$ single-copy measurements of the state $\rho$ with a probability of success at least $1-\delta$. 
\end{theorem}
The proof of this theorem is similar to that of Theorem~\ref{th:proofofefficientpurestatetesting}.
\subsection{Tomography of approximately $t$-compressible states}
In this subsection, we address the problem of tomography for `approximate $t$-compressible states', i.e., states with $n-t$ normal eigenvalues of their correlation matrix that are not exactly one, but \emph{almost}.

This approach is particularly useful for the tomography of states where most of the normal eigenvalues are very close to one, such as states obtained when considering impurity models~\cite{Bravyi_2017}, as we will explore in more detail in the next section. Moreover, when preparing a target $t$-compressible state in a quantum device (e.g., a target pure state prepared by a 1D matchgate circuit doped with a few SWAP gates), noise can cause an effective state $\rho$ to be prepared instead of the target $t$-compressible state. Therefore, an algorithm that allows for perturbations from the set of $t$-compressible states is experimentally motivated.

We provide such an algorithm in Table~\ref{alg:algoSMappr}. Note that this differs from Algorithm~\ref{alg:algo} because in this mixed-state scenario we do not need to measure the last $n-t$ qubits in the computational basis before performing full-state tomography on the first $t$ qubits.

\begin{algorithm}
\label{alg:algoSMappr}
\caption{Learning algorithm for approximately $t$-compressible fermionic Gaussian states (possibly mixed)}
\KwIn{Accuracy $\varepsilon$, failure probability $\delta$, $N = \mathcal{O}\!\left(\frac{n^5}{\varepsilon^{4}}\log\!\left(\frac{n^2}{\delta}\right)\right)+ N_{\mathrm{tom}}\!\left(t,\frac{\varepsilon}{2},\frac{\delta}{2}\right)$ copies of the state $\rho$, where $N_{\mathrm{tom}}$ is the number of copies needed for $t$-qubit state tomography with accuracy $\frac{\varepsilon}{2}$ and failure probability $\frac{\delta}{2}$.}
\KwOut{A classical description of $\hat{\rho}$, ensuring $d_{\mathrm{tr}}(\hat{\rho}, \rho)\le \varepsilon + \varepsilon_t$ with probability at least $1-\delta$. Here, $\varepsilon_t\coloneqq \sqrt{\sum^{n}_{j=t+1}\frac{(1-\lambda_j)}{2}}$, 
where $\{\lambda_j\}^n_{j=1}$ are the normal eigenvalues of the correlation matrix of $\rho$ ordered in increasing order.}

Estimate the correlation matrix of $\rho$ using $\mathcal{O}\!\left(\frac{n^5}{\varepsilon^{4}}\log\!\left(\frac{n^2}{\delta}\right)\right)$ single-copy measurements (see Lemma~\ref{le:samplecompAPPcommuting}), obtaining $\hat{C}$\;

Express $\hat{C}$ in its normal form $\hat{C}=\hat{O}\hat{\Lambda}\hat{O}^T$ (Eq.\eqref{eq:decomAntisym}) and find the Gaussian unitary $G_{\hat{O}}$ associated with $\hat{O} \in \mathrm{O}(2n)$\;

Using $N_{\mathrm{tom}}(t,\frac{\varepsilon}{2},\frac{\delta}{2})$ copies, perform full-state tomography on the first $t$ qubits of the state $G^{\dag}_{\hat{O}}\rho G_{\hat{O}}$, and let $\hat{\sigma}$ be the $t$-qubit state output by tomography\;

\Return $\hat{O}$ and $\hat{\sigma}$, which identify $\hat{\rho} \coloneqq G_{\hat{O}}(\hat{\sigma}\otimes \ketbra{0^{n-t}})G^{\dag}_{\hat{O}}$\;
\end{algorithm}

\begin{theorem}[Efficient learning of approximate $t$-compressible Gaussian states]
\label{th:joiningpiecesMIXEDApp}
Let $\varepsilon, \delta \in (0,1]$. Let $\rho$ be a quantum state. By utilizing $N = \mathcal{O}\!\left(\frac{n^5}{\varepsilon^{4}}\log\!\left(\frac{n^2}{\delta}\right)\right)+ N_{\mathrm{tom}}\!\left(t,\frac{\varepsilon}{2},\frac{\delta}{2}\right)$ single-copy measurements and $T = \mathcal{O}(n^3) + T_{\mathrm{tom}}\!\left(t,\frac{\varepsilon}{2},\frac{\delta}{2}\right)$ computational time, Algorithm~\ref{alg:algoSMappr} yields a classical representation of a state $\hat{\rho}$, satisfying 
\begin{align}
    d_{\mathrm{tr}}(\hat{\rho}, \rho) \le \varepsilon + \varepsilon_t,
\end{align}
with probability $\ge 1-\delta$. Here, $\varepsilon_t\coloneqq \sqrt{\sum^{n}_{j=t+1}(1-\lambda_j)/2}$, where $\{\lambda_j\}^n_{j=1}$ are the normal eigenvalues of the correlation matrix of $\rho$ ordered in increasing order. Moreover, $N_{\mathrm{tom}}(t,\frac{\varepsilon}{2},\frac{\delta}{2})$ and $T_{\mathrm{tom}}(t,\frac{\varepsilon}{2},\frac{\delta}{2})$ respectively denote the number of copies and computational time sufficient for full-state tomography of a mixed $t$-qubit state with an $\varepsilon/2$ accuracy and a failure probability of at most $\delta/2$.
\end{theorem}
\begin{proof}
The learning procedure is outlined in Algorithm~\ref{alg:algoSMappr} and we now establish its correctness.
According to Lemma~\ref{le:samplecompAPPcommuting}, $ N_c(n,\varepsilon_c,\delta/2)$ single copies of $\rho$ are sufficient to construct an anti-symmetric real matrix $\hat{C}$ such that $\norm{\hat{C}-C}_\infty\le\varepsilon_c$ with a probability of at least $1-\delta/2$, where $C$ is the correlation matrix of $\rho$. 
Then, we can put $\hat{C}$ in its normal form $\hat{C}=\hat{O}\hat{\Lambda} \hat{O}^{T}$, where $\hat{O}\in \mathrm{O}(2n)$ is an orthogonal matrix and $\hat{\Lambda}$ is a matrix of the form $\hat{\Lambda} = i\bigoplus^{n}_{j=1} \lambda_{j}(\hat{C}) Y$, where $\{ \lambda_{j}(\hat{C})\}^{n}_{j=1}$ are the normal eigenvalues of $\hat{C}$. We can then construct~\cite{dias2023classical,zhao2023PhD,Zhao_2024} the Gaussian unitary $\hat{G}$ associated to $\hat{O}$. Subsequently, we consider the state $\rho^{\prime}\coloneqq \hat{G}^{\dag}\rho \hat{G}$.
For $m\in\{t+1,\dots, n\}$, we have:
    \begin{align}
        \Tr(Z_m \rho^{\prime}) & =C(\rho^{\prime})_{2m-1,2m} \\
                               & \ge (\hat{\Lambda} )_{2m-1,2m} -  \varepsilon_c \\
                               & =  \lambda_{m}(\hat{C}) -  \varepsilon_c \\
                               & \ge \lambda_{m}(C)- 2 \varepsilon_c,
    \end{align}
    where $Z_m=-i\gamma_{2m-1}\gamma_{2m}$ is the $Z$-Pauli operator acting on the $m$-th qubit. In the second step we used that that $\lvert C(\psi^{\prime})_{j,k}-(\hat{\Lambda} )_{j,k} \rvert\le \varepsilon_c$, as follows by $C(\rho^{\prime})=\hat{O}^T C(\rho)\hat{O}$, $\hat{\Lambda}=\hat{O}^T \hat{C}\hat{O}$, Cauchy-Schwarz and the definition of infinity norm.
    In the last step we used that, because of Lemma~\ref{le:lbantisymm}, we have $|\lambda_{j}(\hat{C})-\lambda_{j}(C)|\le \varepsilon_c$ for each $j\in [n]$, where $\{ \lambda_{j}(C)\}^{n}_{j=1}$ are the normal eigenvalues of $C$. 
    Therefore, we also have 
    \begin{align}
        \Tr(\ketbra{0}{0}_m \rho^{\prime})=\frac{1}{2}+\frac{1}{2}\Tr(Z_m \rho^{\prime})\ge
        1 -  (\varepsilon_c + \frac{1-\lambda_{m}(C)}{2}).
    \end{align}
    By using the Quantum Union Bound (Lemma~\ref{le:qunionbound}), we have:
    \begin{align}
        d_{\mathrm{tr}}(\rho^{\prime}, \rho^{\prime}_{\mathrm{post}}\otimes \ketbra {0^{n-t}}) \le \sqrt{(n-t)  \varepsilon_c + \varepsilon^2_t},
        \label{eq:Qunbound}
    \end{align}
    where $\varepsilon_t\coloneqq \sqrt{\sum^{n}_{j=t+1}(1-\lambda_j(C))/2}$, and $ \rho^{\prime}_{\mathrm{post}} \otimes \ketbra {0^{n-t}}$ is the post-measurement state after having measured the outcomes corresponding to $\ket{0^{n-t}}$ in the last $n-t$ qubits. Let $\rho^{\prime}_t$ be the reduced state of $\rho^{\prime}$ over the first $t$-qubits. We have:
\begin{align}
    d_{\mathrm{tr}}(\rho^{\prime}, \rho^{\prime}_t\otimes \ketbra {0^{n-t}}) &\le 
    d_{\mathrm{tr}}(\rho^{\prime}, \rho^{\prime}_{\mathrm{post}}\otimes \ketbra {0^{n-t}}) + d_{\mathrm{tr}}(\rho^{\prime}_{\mathrm{post}}\otimes \ketbra{0^{n-t}}, \rho^{\prime}_t\otimes \ketbra {0^{n-t}})\\
    &= 
    d_{\mathrm{tr}}(\rho^{\prime}, \rho^{\prime}_{\mathrm{post}}\otimes \ketbra {0^{n-t}}) + d_{\mathrm{tr}}(\rho^{\prime}_{\mathrm{post}}, \rho^{\prime}_t)\\
    &\le  
    2 d_{\mathrm{tr}}(\rho^{\prime}, \rho^{\prime}_{\mathrm{post}}\otimes \ketbra {0^{n-t}}) \\
    &\le  
    2 \sqrt{(n-t)  \varepsilon_c + \varepsilon^2_t},
    \label{eq:chainAPP}
\end{align}
where in the first step we used the triangle inequality, in the second step we used the unitary invariance of the trace distance, in the third step we used the data processing inequality~\cite{Wilde_2017} (with respect to the operation of tracing out a subsystem) i.e., $d_{\mathrm{tr}}(\rho^{\prime}_{\mathrm{post}}, \rho^{\prime}_t)\le d_{\mathrm{tr}}(\rho^{\prime}_{\mathrm{post}}\otimes \ketbra{0^{n-t}}, \rho^{\prime})$, and in the last step we used Eq.\eqref{eq:Qunbound}.
Using now $N_{\mathrm{tom}}\!\left(t,\frac{\varepsilon}{2},\frac{\delta}{2}\right)$ copies of the state $\rho^{\prime}=\hat{G}^{\dag}\rho \hat{G}$, we perform full-state tomography on the first $t$-qubits, obtaining $\hat{\sigma}$ such that 
\begin{align}
d_{\mathrm{tr}}(\hat{\sigma}, \rho^{\prime}_t)\le \varepsilon/2,    
\label{eq:tomAPP}
\end{align}
with $\ge 1-\delta/2$ probability.
Our output state is $\hat{\rho}  \coloneqq   \hat{G}(\hat{\sigma} \otimes \ketbra{0^{n-t}})\hat{G}^{\dag}$, and the information about such a state is provided in the output by providing the orthogonal matrix $\hat{O} \in \mathrm{O}(2n)$, which identifies $\hat{G}$, and the $t$-qubit state $\hat{\sigma}$.
Considering the trace distance between $\hat{\rho}$ and $\rho$, we have:
\begin{align}
\label{eq:finaleqAPP}
    d_{\mathrm{tr}}(\hat{\rho}, \rho)&=d_{\mathrm{tr}}(\hat{\sigma} \otimes \ketbra{0^{n-t}}, \hat{G}^{\dag} \rho \hat{G}) \\ 
    &=d_{\mathrm{tr}}(\hat{\sigma} \otimes \ketbra{0^{n-t}}, \rho^{\prime})\\
    &\le d_{\mathrm{tr}}(\hat{\sigma} \otimes \ketbra{0^{n-t}}, \rho^{\prime}_t \otimes \ketbra{0^{n-t}}) + d_{\mathrm{tr}}(\rho^{\prime}_t \otimes \ketbra{0^{n-t}}, \rho^{\prime})\\
    &\le d_{\mathrm{tr}}(\hat{\sigma}, \rho^{\prime}_t) + d_{\mathrm{tr}}(\rho^{\prime}_t \otimes \ketbra{0^{n-t}}, \rho^{\prime})\\
    & \le \frac{\varepsilon}{2} + 2 \sqrt{(n-t)  \varepsilon_c + \varepsilon_t} ,
\end{align}
where in the last step we used Eq.\eqref{eq:chainAPP} and Eq.\eqref{eq:tomAPP}.
By setting $\varepsilon_c \coloneqq  \varepsilon^2/(16(n-t))$, we get:
\begin{align}
    d_{\mathrm{tr}}(\hat{\rho}, \rho) \le \frac{\varepsilon}{2} + 2 \sqrt{ \frac{\varepsilon^2}{16} + \varepsilon^2_t} \le \varepsilon +   \varepsilon_t,
\end{align}
where we used that $\sqrt{a+b}\le \sqrt{a}+ \sqrt{b}$ for each $a,b\ge 0$. The total failure probability of the algorithm, by union bound, is $\le \delta$.
\end{proof}

As a consequence of the previous theorem, we present the following proposition, which states that if the unknown quantum state is sufficiently close in trace distance to the set of $t$-compressible Gaussian states, then our learning algorithm (Algorithm~\ref{alg:algoSMappr}) can still be effectively applied.
\begin{theorem}[Noise robustness of learning $t$-compressible Gaussian states]
\label{th:joiningpiecesMIXEDApp2}
Let $\varepsilon, \delta \in (0,1]$. Let $\rho$ be a quantum state such that
\begin{align}
    \min_{\sigma_t \in \mathcal{G}_t } d_{\mathrm{tr}}(\rho, \sigma_t) \le \frac{\varepsilon^2}{4(n-t)},
\end{align}
i.e., it is sufficiently close to the set of possibly mixed $t$-compressible states $\mathcal{G}_t$.
Then, there exists a learning algorithm which, utilizing 
\begin{align}
    N = \mathcal{O}\left(\frac{n^5}{\varepsilon^{4}}\log\left(\frac{n^2}{\delta}\right)\right) + N_{\mathrm{tom}}\left(t,\frac{\varepsilon}{4},\frac{\delta}{2}\right)
\end{align}
single-copy measurements, yields a classical representation of a state $\hat{\rho}$, satisfying $d_{\mathrm{tr}}(\hat{\rho}, \rho) \le \varepsilon$ with probability $\ge 1-\delta$. 

Here, $N_{\mathrm{tom}}(t,\frac{\varepsilon}{4},\frac{\delta}{2})$ denotes the number of copies sufficient for full-state tomography of a mixed $t$-qubit state with an $\varepsilon/4$ accuracy and a failure probability of at most $\delta/2$.
\end{theorem}

\begin{proof}
By applying the previous Theorem~\ref{th:joiningpiecesMIXEDApp} with accuracy $\varepsilon/2$ (instead of $\varepsilon$), we get:
\begin{align}
    d_{\mathrm{tr}}(\hat{\rho}, \rho) \le \frac{\varepsilon}{2} + \varepsilon_t,
    \label{eq:nstart}
\end{align}
where $\varepsilon_t \coloneqq \sum^{n}_{j=t+1}(1-\lambda_j)/2$, and $\{\lambda_j\}_{j=1}^n$ are the normal eigenvalues of the correlation matrix of $\rho$ ordered in increasing order. Let us upper bound $\varepsilon_t$. 
First of all, by assumption, it follows that there exists a $t$-compressible Gaussian state $\sigma_t$ such that 
\begin{align}
    d_{\mathrm{tr}}(\sigma_t, \rho) \le \frac{\varepsilon^2}{4(n-t)}.
    \label{eq:noiss}
\end{align}
Then, we have:
\begin{align}
    \varepsilon_t &\coloneqq  \sqrt{ \sum^{n}_{j=t+1}\frac{(1-\lambda_j)}{2}} \\
    &\le \sqrt{(n-t)\frac{(1-\lambda_{t+1})}{2}} \\
    &\le \sqrt{\frac{(n-t)}{2} \norm{C(\sigma_t)-C(\rho)}_{\infty}} \\
    &\le \sqrt{(n-t) d_{\mathrm{tr}}(\sigma_t, \rho)} \\
    &\le  \frac{\varepsilon}{2},
    \label{eq:chain3}
\end{align}
where in the third step we used Lemma~\ref{le:lbantisymm}, in the fourth step we used that the one-norm difference between two (possibly non-Gaussian) quantum states upper bounds the infinity norm difference of their correlation matrix, i.e.:
\begin{align}
    \norm{C(\sigma_t)-C(\rho)}_{\infty} \le 2 d_{\mathrm{tr}}(\sigma_t, \rho),
\end{align}
as it can be seen by the variational definition of trace distance (see the derivation of Eq.~\eqref{eq:vardeftr} for an analogous and explicit proof of the latter inequality). In the last step we used Eq.~\eqref{eq:noiss}. Thus, substituting in Eq.~\eqref{eq:nstart}, we conclude.
\end{proof}

\section{Details of numerical simulations}
\label{sec:numerics}
This section concerns a few additional details about the numerical results on impurity model dynamics, presented in Section~\ref{sec:numericsMAIN} and Figure~\ref{fig:rank_dynamics} of the main text.

In the main text we touch upon the question: can the compressibility shown in Figure~\ref{fig:rank_dynamics} be derived from our Theorem~\ref{th:1compr} for $t$-doped circuits simply using the Trotter approximation? If true, this would render our numerical results less surprising; the reasoning for such a derivation could go as follows. The evolution unitary $e^{-iHT}$ can be approximated with a product formula using some number $t'$ of Trotter steps, which for any impurity model yields a $t'$-doped circuit. Such a Trotter-approximated state $\ket{\psi(T)}_{\mathrm{Trot}; t'}$ would then be $t$-compressible for $t = \kappa t'$ (where $\kappa$ is the locality of the impurity term; for $H_{\mathrm{graph}}$ and $H_{\mathrm{TFIM}}$, $\kappa = 4$) due to Theorem~\ref{th:1compr}. 
To test whether this explains the observed compressibility, we numerically check the number of Trotter steps needed to approximate $\ket{\psi(T)}$ such that $d_{\mathrm{tr}}(\ket{\psi(T)}, \ket{\psi(T)}_{\mathrm{Trot}; t'}) \le \varepsilon_{\mathrm{trunc}}^2/(4(n-\kappa t'))$. Because of Eq.~\eqref{eq:chain3}, this would imply that $\sqrt{\sum_{m=t+1}^{n} \frac{1}{2} (1 - \lambda_m)} \leq \varepsilon_{\mathrm{trunc}}$ (i.e., Eq.~\eqref{eq:epsilon_tr_defMAIN}) is satisfied.
We find that the prediction from the Trotter approximation and compression theorem, as explained above, does not capture the dynamics of the Gaussian nullity we have observed. For example, at minimal resolved time $T = 0.05$ for both impurity models, the Trotter approximation already predicts $t = 14$ as the minimal Gaussian nullity (see vertical dotted lines in Figure~\ref{fig:rank_dynamics}, left panel). In reality, however, the expander model dynamics still allows approximate $t$-compression with $t = 4$ all the way up to $T=0.4$ (and the TFIM model even up to $T=1.1$). 

One may ask, why exactly is this analysis, by Trotterization and Theorem~\ref{th:1compr}, coming so short of explaining the observed compressibility. In our opinion, this should be related to the fact that the non-Gaussian gates in the Trotterized evolution circuits are close to identity. One intuitively expects that states produced in such circuits would be much more compressible (approximately), than if the involved non-Gaussian gates were generic, far from identity. However, Theorem~\ref{th:1compr} does not account for this difference, and gives the same compressibility guarantee if the number of non-Gaussian gates is the same; also, importantly, Theorem~\ref{th:1compr} only deals with exact and not approximate compressibility. It is at present unclear how to account for the nature of non-Gaussian evolution beyond gate-counting, so that our numerical results on compressibility of impurity model dynamics could be explained analytically.

\begin{figure}[h]
    \centering 
    \includegraphics[width=0.69\linewidth]{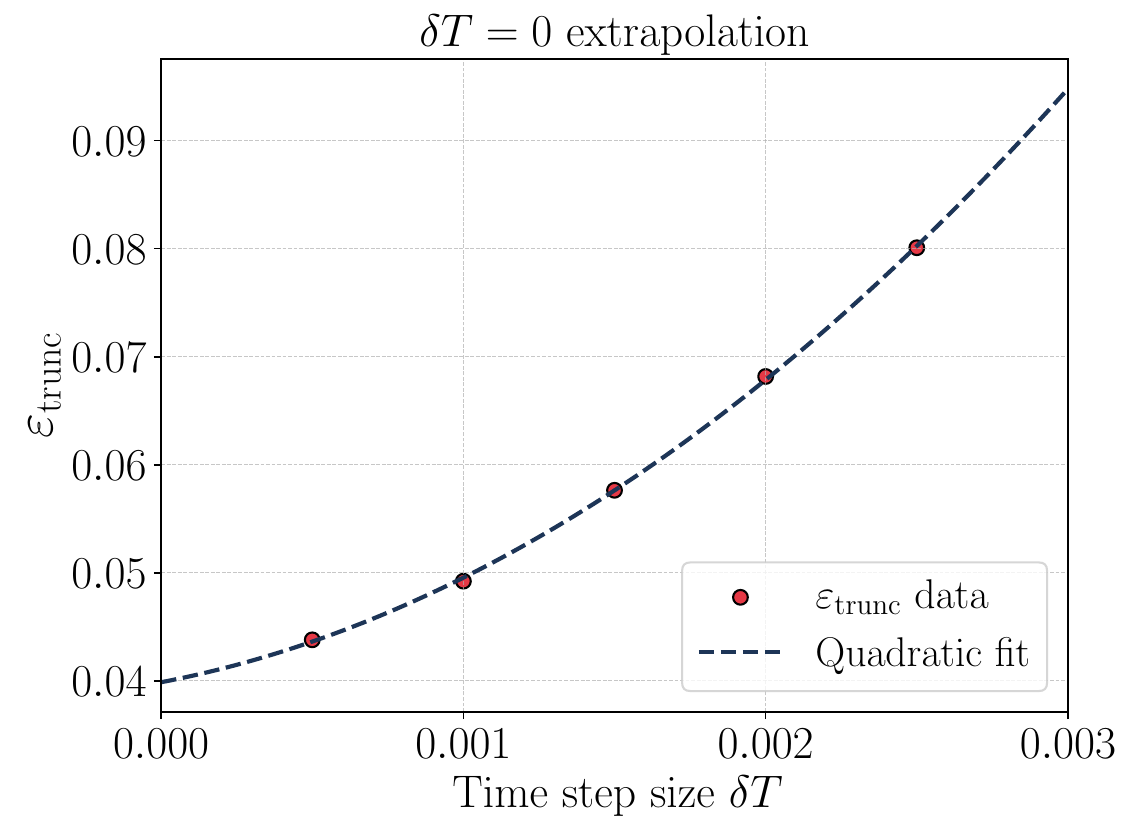}
    \caption{Extraction of truncation error $\varepsilon_{\mathrm{trunc}}$ from sparse simulation via $\delta T\rightarrow 0$ extrapolation, used to produce Figure~\ref{fig:rank_dynamics}. The displayed extrapolation is for the most challenging case of the expander model on $n=16$ qubits. Here the evolution time is $T=0.4$ and the truncation nullity is $t=4$. The relevant scale for the error is the truncation value $\varepsilon_{\mathrm{trunc}}=0.05$, used in Figure~\ref{fig:rank_dynamics}.}
    \label{fig:dt0_ext}
\end{figure}

To run the simulations of $\ket{\psi(T)}=e^{-iHT} \ket{0^n}$ presented in Figure~\ref{fig:rank_dynamics}, we employed sparse matrix multiplication routines. Specifically, we approximated the exponential operator as $e^{-iHT} \simeq (1-iH\delta T)^{T/\delta T}$, followed by state renormalization. To ensure that our results are as accurate as possible, the value of $\varepsilon_{\rm trunc}$ for each $\ket{\psi(T)}$ was obtained from a $\delta T\rightarrow 0$ extrapolation. In particular, it was computed $\varepsilon_{\rm trunc}$ for $\delta T\in \{ 0.0005,0.001, 0.0015, 0.002,0.0025 \}$ and then extended to $\delta T=0$ via a parabolic fit (see Figure~\ref{fig:dt0_ext}). This approach enabled simulations for systems up to $n=16$ qubits in Figure~\ref{fig:rank_dynamics}.


\end{document}